\documentclass[a4paper,12pt]{amsart} 
\usepackage{amssymb, amsmath, amsthm, color, amstext, tikz, appendix, mathrsfs, mathabx}
\usepackage{dsfont}
\usepackage[colorlinks=true, breaklinks=true, linkcolor=black, citecolor=blue, urlcolor=red]{hyperref} 
\usepackage[english]{babel}

\usepackage [all,cmtip]{xy}

\numberwithin{equation}{section}
\newtheorem{letterthm}{Theorem}

\newtheorem{theo}{Theorem}[section]
\newtheorem{lemma}[theo]{Lemma}
\newtheorem{corollary}[theo]{Corollary}
\newtheorem{prop}[theo]{Proposition}
\newtheorem{proposition}[theo]{Proposition}

\theoremstyle{definition} 
\newtheorem{definition}[theo]{Definition}
\newtheorem{notation}[theo]{Notation}
\newtheorem{remark}[theo]{Remark}
\newtheorem{example}[theo]{Example}

\DeclareMathOperator{\Leb}{Leb}
\DeclareMathOperator{\Gr}{Gr}

\DeclareMathOperator{\id}{id}

\newcommand{\act}{\curvearrowright}
\newcommand{\al}{\alpha}
\newcommand{\oscrA}{{\overline{\mathscr{A}}}}
\DeclareMathOperator{\Ad}{Ad}
\DeclareMathOperator{\alg}{alg}
\DeclareMathOperator{\Aut}{Aut}
\newcommand{\scrB}{\mathscr B}
\DeclareMathOperator{\Calg}{C*-alg}
\newcommand{\scrC}{\mathscr C}
\newcommand{\cC}{\mathcal C}
\newcommand{\C}{\mathbf C}
\newcommand{\Ch}{\textup{Char}}
\newcommand{\D}{\mathbf D}
\newcommand{\cD}{\mathcal D}
\newcommand{\dn}{d_1=0<d_2<\cdots<d_n<d_{n+1}=1}
\newcommand{\de}{\delta}
\DeclareMathOperator{\dis}{dis}

\newcommand{\cF}{\mathcal F}
\newcommand{\cAF}{\mathcal{AF}}
\newcommand{\fH}{\mathfrak H}
\newcommand{\scrH}{\mathscr H}
\DeclareMathOperator{\Hom}{Hom}
\DeclareMathOperator{\Hilb}{Hilb}
\newcommand{\ga}{\gamma}
\newcommand{\Ga}{\Gamma}
\newcommand{\Gd}{\oplus_\D G} 
\newcommand{\cI}{\mathcal I}

\newcommand{\la}{\lambda}
\newcommand{\scrM}{\mathscr M}
\newcommand{\mbD}{m_\beta^\D}
\newcommand{\scrN}{\mathscr N}
\newcommand{\N}{\mathbf{N}}
\newcommand{\ot}{\otimes}

\newcommand{\ovt}{\overline\otimes}
\DeclareMathOperator{\Prob}{Prob}
\newcommand{\scrQ}{\mathscr Q}
\newcommand{\R}{\mathbf{R}}
\DeclareMathOperator{\Rot}{Rot}
\newcommand{\fS}{\mathbf{S}}
\DeclareMathOperator{\Set}{Set}
\newcommand{\cSF}{\mathcal{SF}}

\DeclareMathOperator{\supp}{supp}
\newcommand{\fT}{\mathfrak T}
\newcommand{\utriv}{\textup{triv}}
\DeclareMathOperator{\Span}{span}
\DeclareMathOperator{\target}{target}
\DeclareMathOperator{\tr}{tr}
\DeclareMathOperator{\Tr}{Tr}
\newcommand{\cU}{\mathcal U}
\newcommand{\varep}{\varepsilon}
\newcommand{\wh}{\widehat}
\newcommand{\Z}{\mathbf{Z}}

\begin{document}

\title[Gauge theory and Thompson's groups]{Operator-algebraic construction of gauge theories and Jones' actions of Thompson's groups}
\thanks{Both authors were partially supported by European Research Council Advanced Grant 669240 QUEST. AS was supported by Alexander-von-Humboldt Foundation through a Feodor Lynen Research Fellowship.
AB is supported by a University of New South Wales Sydney starting grant.}
\author{Arnaud Brothier and Alexander Stottmeister}
\address{Arnaud Brothier\\ School of Mathematics and Statistics, University of New South Wales, Sydney NSW 2052, Australia\\
Tel.: +612 9385 7077}
\email{arnaud.brothier@gmail.com\endgraf
\url{https://sites.google.com/site/arnaudbrothier/}}
\address{Alexander Stottmeister\\Department of Mathematics, University of Rome Tor Vergata, Via della Ricerca Scientifica 00133 Roma, Italy}
\email{alexander.stottmeister@gmail.com}

\thispagestyle{empty}

\begin{abstract}
Using ideas from Jones, lattice gauge theory and loop quantum gravity, we construct 1+1-dimensional gauge theories on a spacetime cylinder.
Given a separable compact group $G$, we construct localized time-zero fields on the spatial torus as a net of C*-algebras together with an action of the gauge group that is an infinite product of $G$ over the dyadic rationals and, using a recent machinery of Jones, an action of Thompson's group $T$ as a replacement of the spatial diffeomorphism group. \\
Adding a family of probability measures on the unitary dual of $G$ we construct a state and obtain a net of von Neumann algebras carrying a state-preserving gauge group action.
For abelian $G$, we provide a very explicit description of our algebras. 
For a single measure on the dual of $G$, we have a state-preserving action of Thompson's group and semi-finite von Neumann algebras.
For $G=\fS$ the circle group together with a certain family of heat-kernel states providing the measures, we obtain hyperfinite type III factors with a normal faithful state providing a nontrivial time evolution via Tomita-Takesaki theory (KMS condition). 
In the latter case, we additionally have a non-singular action of the group of rotations with dyadic angles, as a subgroup of Thompson's group $T$, for geometrically motivated choices of families of heat-kernel states.
\end{abstract}

\maketitle

\section*{Introduction}
Jones recently discovered a very general process of constructing actions of Thompson's groups $F<T<V$ and more generally for any group of fractions, see \cite{Jones17-Thompson,Jones16-Thompson}.
This discovery arose from an attempt to build a conformal field theory directly from a subfactor using the formalism of Jones' planar algebras:
Thompson's group $T$ (a group of local scale transformations and translations) is seen as a replacement of the positive diffeomorphisms of the circle that is the classical symmetry group of a chiral conformal field theory. A similar construction has been invoked in a particular approach to the canonical quantization of general relativity, namely loop quantum gravity, to implement in parts the concept of background independence, i.e.~covariance w.r.t.~(spatial) diffeomorphisms,  see \cite{Thiemann-07-QuantumGR, Baez-96-Spin-Networks, Ashtekar-95-diffeomorphism-invariant-theories} and references therein.
{Note that Thompson's group $T$ is only a weak substitute of the diffeomorphism group. Indeed, the action of the rotation subgroup of $T$ (rotations with dyadic rationals for angle) is generically very discontinuous, see \cite{Jones16-Thompson} and \cite{Brot-Jones18-2} for explicit computations.}
This project, which the present article is part of, has been motivated by the aforesaid similarities between Jones' machinery and constructions in loop quantum gravity, in particular the operator-algebraic treatment by one of the authors \cite{Stottmeister-Thiemann-16-Coherent2, Stottmeister-Thiemann-16-Coherent3}. As the latter are partially motivated by classical ideas in lattice gauge theory, see \cite{Kogut-83-Lattice-Gauge, Creutz-85-Quarks}, we start our construction from much simpler representation-theoretical data, i.e.~a compact group, than the more general approach by Jones involving subfactors in a rather different direction, see \cite{Jones18-Hamiltonian}. Nevertheless, this leads to nontrivial results: the construction of one of the simplest gauge theories, Yang-Mills theory in 1+1 dimensions (YM$_{1+1}$), cf.~for example \cite{Witten-91-Quantum-Gauge, Sengupta-01-Yang-Mills}, together with an action by Thompson's group $V$ or its rotation subgroup.
In this article, we present a precise operator-algebraic construction of the model and focus on a detailed exposition of mathematical results following closely the notation used by Jones. In our companion article \cite{Brot-Stottmeister-Phys}, we give a complementary description of the construction in the context of lattice gauge theory and focus on physics related aspects of the model, especially renormalization, and its relation to existing treatments of YM$_{1+1}$, see \cite{Dimock-96-Yang-Mills, Ashtekar-97-Yang-Mills, Driver-99-Yang-Mills}. Moreover, we point out that our rather complete operator-algebraic formulation of YM$_{1+1}$ allows to understand the differences among those previous accounts in a very precise manner. Our construction also bears certain similarities as well as critical differences to more recent operator-algebraic treatments of lattice gauge theory \cite{Grundling-Rudolph-13-Infinite-Lattice-QCD, Grundling-Rudolph-17-Dynamics-QCD, Arici-Stienstra-Suijekom18}, which we explain in detail in our companion article.\\[0.25cm]
Let us give a more detailed summary of the main construction and the article's content:\\[0.1cm]
Given a separable compact group $G$ (e.g.~a compact Lie group), we consider the weakly dense C*-subalgebra $M\subset B(L^2(G))$ generated by the continuous function $\scrC(G)$ acting by pointwise multiplications and the group ring $\C[G]$ acting by convolution.
The choice of $M$ is driven by the construction of a kinematical C*-algebraic theory and the necessity of having unital completely positive maps between algebras (a condition that is natural w.r.t.~renormalization and spatial locality, see {\cite[Sections 2.3 and 4]{Brot-Stottmeister-Phys}} for details).
If we only desired to have a von Neumann-algebraic model, we could directly pass to $M=B(L^2(G))$ or any weakly dense *-subalgebra, {see Remark \ref{rem:subalg}.}
For a standard dyadic partition (s.d.p.~in short) $t$ of the torus $\R/\Z$, we define the C*-algebra at level $t$ as $M_t:=\ot_I M$ a (minimal) tensor product of copies of $M$ indexed by the intervals $I$ associated with the partition.
This comes along with a gauge group $\prod_{\partial I} G$ acting on $M_t$, where the product is indexed by the left boundary $\partial I$ of the intervals. 
We fix a unital embedding $R:M\to M\ot M$ which provides unital embeddings $\iota_t^s:M_t\to M_s$ for $s$ a refinement of the partition $t$, i.e.~$t\leq s.$
The map $R$ is chosen such that it reflects the group structure of $G$, or more precisely the commutation relations between $\scrC(G)$ and $\C[G]$, and is equivariant w.r.t.~the gauge group action. This way, we obtain an inductive-limit C*-algebra $\scrM_0$ and, using Jones' machinery, an action of Thompson's group $V$ on $\scrM_0$ which is one of the main point of our construction.
Additionally, we have an action of the gauge group $\prod_\D G$ that is the infinite product of $G$ indexed by the dyadic rationals $\D$. 
The C*-algebra $\scrM_0$ is precisely described by a (discrete) crossed-product $\scrC(\oscrA)\rtimes \oplus_\D G$, where $\oscrA$ can be interpreted as a compact Hausdorff space of generalized connections \cite{Ashtekar-Lewandowski-95-Projective-Differential} and $\oplus_\D G$ is the infinite direct sum of $G$ over the dyadic rationals serving as a space of compatible conjugate momenta. 
Additionally, there are explicit formulas for the actions of Thompson's group $V$ and the gauge group $\prod_\D G$ on $\scrM_{0}$, see sections \ref{sec:M_0} and \ref{sec:gauge}.
Moreover, because of the aforementioned unitality, it is possible to localize our construction w.r.t.~connected open subsets $O$ of the torus, thus giving us a net of C*-algebras $\scrM_0(O)$ corresponding to time-zero fields localized in some region of space.
The localization is realized by considering at each finite level $t$ the C*-subalgebra $M_t(O)\subset M_t:= \ot_I M$ that is given by the unital diagonal embedding of the tensor product of $M$ indexed by those intervals $I$ contained inside $O$.
It follows that the localized algebras are stable under the action of the gauge group and geometrically covariant w.r.t.~Thompson's group $T$ that is the restriction to $V$ of elements that are homeomorphisms of the torus. Elements of $V$ only act continuously on the torus up to finitely many points and, thus, $vO$ is not necessarily connected for $v\in V$.
The C*-algebraic framework provides a kinematical model with a notion of (infinitely divisible, homogeneous, Cantor-like) space subject to spatial covariance given by Thompson's group $T$ and compatible inner degrees of freedom subject to the action of the gauge group. Similar to the point of view advocated in \cite{Jones17-Thompson}, we consider $\scrM_{0}$ as a \textit{semi-continuum} limit of field algebras, which fits with the Cantor-like structure of the dyadic rationals $\D$.\\[0.1cm]
Notably, the net of C*-algebras $\scrM_0(O)$ satisfies axioms similar to those of a conformal net \cite{Gabbiani-Froehlich-93-Conformal}, see Proposition \ref{prop:observables}.\\[0.25cm]
As hinted at above, our choice of unital embeddings $\iota_t^s:M_t\to M_s$, $t\leq s$, is partially motivated by considerations from Wilson's approach to the renormalization group \cite{Wilson-71-Renormalization1, Wilson-71-Renormalization2}, see also \cite{Fernandez-92-Random-Walks}, and, thus, with guidance from the setting of statistical and Euclidean field theory, we assume that a \textit{full continuum limit} requires additional data in terms of a collection of states $\omega_{t}$ as replacement for the measures in the commutative setting, see {\cite[Section 4]{Brot-Stottmeister-Phys}} for further details. Furthermore, we should be able to obtain a natural notion of time evolution in the (semi-)continuum and, thus, a 1+1 dimensional theory via Tomita-Takesaki theory whenever the states give rise to a limit state $\varpi$ on $\scrM_{0}$ satisfying the Kubo-Martin-Schwinger (KMS) condition.\\[0.1cm]
Thus, our next step consists in providing said states which give us von Neumann algebras together with a faithful state at every level $t$ via the Gelfand-Naimark-Segal (GNS) construction.
To this end we consider a family of strictly positive probability measures $m_d\in\Prob(\wh G)$ on the discrete, unitary dual of $G$ indexed by the dyadic rational $d\in\D$.
Here, our leading example involves a compact Lie group $G$ together with a family of heat-kernel states providing the measures, which are just the Gibbs (or KMS) states associated with the Kogut-Susskind Hamiltonian \cite{Kogut-Susskind75} of YM$_{1+1}$ at level $t$ (viewing the partition $t$ as a 1-dimensional lattice). The family of measures $m_d$ provides a family of faithful states $\omega_t$ on $M_t$ such that the inclusion maps $\iota_t^s:M_t\to M_s$ for $t\leq s$ are state-preserving. By standard reasoning, we have a limit state $\varpi:=\varprojlim_t \omega_t$ on the C*-limit $\scrM_0:=\varinjlim_t M_t$. Performing the GNS construction, we obtain a von Neumann algebra equipped with a normal state $(\scrM,\varpi)$. Although, it is not obvious under which conditions on the family of states $\omega_t$ the state $\varpi$ extends to a faithful state on the von Neumann algebra $\scrM$, a simple argument shows that the GNS representation of $\varpi$ is faithful, see Section \ref{sec:state-prelim}. 
Moreover, the gauge group acts in a state-preserving way for any choice of measures $m_d$.
If each measure $m_d$ equals a fixed one, we can show that the Jones' action $V\act\scrM_0$ on the C*-algebra $\scrM_0$ extends to a state-preserving action on the von Neumann algebra $(\scrM,\varpi)$. However, besides this restricted case, there appears to be in general no simple way to decide whether the Jones' action $V\act \scrM_0$ extends to an action on $\scrM$.\\[0.1cm]
Picking up the renormalization theme once more, it interesting to observe that it is possible to cast our construction into a form that is very similar to the \textit{multi-scale entanglement renormalization ansatz} (MERA) \cite{Vidal-08-MERA, Evenbly-Vidal-15-TNR, Evenbly-Vidal-16-TNR-MERA} by mapping our inductive limit $\iota_t^s:(M_t, \omega_t)\to (M_s,\omega_s)$, $t\leq s$, to its standard form corresponding to its Bratteli diagram \cite{Evans_Kawahigashi_book}. In the language of \cite{Vidal-08-MERA}, the system of unitaries $U_t$ associated with this mapping serves as disentanglers while the maps induced by the GNS construction from the standard unital embeddings of the Bratteli diagram yield the isometries, see {\cite[Sections 3.4 and 4]{Brot-Stottmeister-Phys}} for further explanations.\\
Following this condensed outline of the main construction, let us outline the article's main results:\\[0.1cm]
In order to further analyze the issues raised above, we specialize in this article to the case when $G$ is abelian and leave the general case for future research. Thereby, we are able to give a precise description of $(\scrM,\varpi)$ and to check when the state $\varpi$ is faithful and when Thompson's group acts on $\scrM$.\\[0.1cm]
Firstly, we consider the case, when all measures $m_d$ are equal to a single one $m\in\Prob(\wh G)$. Then, we obtain a precise statement about the structure of $(\scrM,\varpi)$ that is a crossed-product von Neumann algebra, see Theorem \ref{theo:singlemes}.
We recall here two opposite cases where $\wh G$ is torsion free and where $\wh G$ is finite equipped with its Haar measure. 
\begin{letterthm}\label{theo:intro-I}
Assume that $G$ is a separable compact abelian group and $m_d$ is equal to a fixed measure $m\in\Prob(\wh G)$ for any $d\in\D.$\\
The von Neumann algebra $\scrM$ is semifinite.
Both, the action of the gauge group and the action of Thompson's group $V$ extend to state-preserving actions on $(\scrM,\varpi)$.\\
If $\wh G$ is torsion free, $(\scrM,\varpi)$ will be isomorphic to a type I$_\infty$ von Neumann algebra with diffuse center equipped with a non-faithful state.\\
If $\wh G$ is finite and $m=m_d$ is the Haar measure, $(\scrM,\varpi)$ will be the hyperfinite II$_1$ factor equipped with its trace.
\end{letterthm}
We note that it is not possible to obtain nontrivial time evolution using Tomita-Takesaki theory in the setup of Theorem \ref{theo:intro-I} since we do not have a faithful non-tracial state $\varpi$ even though the states $\omega_{t}$ are faithful at any finite level $t$. Moreover, all relevant compressions will provide either a non-faithful state or a faithful tracial state resulting in trivial modular time evolutions. Nevertheless, it is possible to obtain a sensible Hamiltonian $H$ affiliated with the algebra $\scrM$ when the state $\varpi$ is constructed from the limiting case $h_{d} = 1$ that results from the degenerate probability measure $m(\{\pi\}) = \delta_{\pi,\pi_{\utriv}}$, see {\cite[Section 3.4]{Brot-Stottmeister-Phys}}. This particular limit state $\varpi$ yields an irreducible representation of $\scrM_{0}$ and is known as the strong-coupling vacuum in lattice gauge theory or Ashtekar-Isham-Lewandowski state in loop quantum gravity, see \cite{Kogut-83-Lattice-Gauge, Ashtekar-Lewandowski-94-Representation-Theory}. \\[0.1cm]
Secondly, we consider a more involved construction with a non-constant family of measures.
For clarity of the presentation, we assume that $G=\fS$ is the circle group (physically corresponding to a pure electromagnetic field) but many of the arguments apply to any compact, separable abelian group, see Remark \ref{rem:generalM}.
As indicated above, we choose a physically motivated class of probability measures on the Pontyagin dual $\Z$ such that the corresponding states are given by the Gibbs or KMS states, with possibly inhomogeneous temperature distribution, of the Kogut-Susskind Hamiltonian \cite{Kogut-Susskind75}. In the 1+1-dimensional setting, the latter is essentially given by bi-invariant Laplacian or quadratic Casimir $\Delta_{G}$ of $G$, and, thus, in the specific case at hand the measures are discrete Gaussians with Fourier transform given by Jacobi theta functions.
Indeed, we have $m_d(\{n\}) = Z_{\beta(d)}^{-1} \exp(-n^2 \beta(d)/2)$ for $d\in\D$, where $\beta:\D\to \R_{>0}$ is a map typically tending to zero for $d$ going to infinity, and $Z_{\beta(d)}=\sum_{n\in\Z}\exp(-n^2 \beta(d)/2)$ is the normalization constant (or state sum). 
We deduce the following result, see Theorem \ref{theo:type}.
\begin{letterthm}\label{theo:intro-type}
Consider the circle group $G=\fS$ and a family of heat-kernel states $m_{d}\in\Prob(\Z)$ parametrized by a function $\beta:\D\to\R_{>0}$.
Assume that $\beta$ is $p$-summable (i.e.~$\sum_{d\in\D} \beta(d)^{p} <\infty$) for some $0<p<1/2$.\\ 
Then, $(\scrM,\varpi)$ is a hyperfinite type III factor equipped with a normal faithful state.
Moreover, the gauge group acts in a state-preserving way on it.
\end{letterthm}
This is drastically different from Theorem \ref{theo:intro-I} where we only consider a single probability measure $m=m_d.$
In particular, we have a nontrivial modular time evolution which commutes with the action of the gauge group, i.e.~we consider the modular flow $\theta:\R\act \scrM, \theta_t(x):= \Delta^{it} x\Delta^{-it}$, where $\Delta$ is the modular operator associated with the GNS representation of $(\scrM,\varpi)$, see for instance \cite{Connes-94-Noncommutative-Geometry, Bratteli-Robinson-87-Operator-Algebras1, Bratteli-Robinson-97-Operator-Algebras2, Borchers-00-Modular-Theory} for details and mathematical as well as physical applications.\\[0.1cm]
We note that, as might be guessed from the use of an abelian group and its Pontryagin dual, the construction of the heat-kernel states can be dualized. In the particular case of the circle group $G=\fS$, this is achieved by defining measures as Fourier transforms of Bessel functions of the second kind, which solve the discrete heat equation on $\Z$. Thereby, we obtain $O(2)$-vector models at each level $t$. We provide further details on this duality and its relation to the strong- and weak-coupling limits of Yang-Mills theory in {\cite[Sections 3.3 and 5.1]{Brot-Stottmeister-Phys}}.\\[0.1cm]
In addition to the previous theorems, we also provide some geometrical choices of $\beta$ for which the restriction of the Jones' action $V\act \scrM_0$ to its rotation subgroup $\Rot$ of $V$ (i.e.~all the rotations of angles $\tfrac{k}{2^n}$) extends to an action on the von Neumann algebra $\scrM$, see Corollary \ref{cor:Rot}.
\begin{letterthm}
Adopting the hypothesis of Theorem \ref{theo:intro-type}, we consider for any dyadic rational $d\in\D$ the length of the largest s.d.i.~starting at $d$ and assume that $\beta(d)$ only depends on this length.
Then, the restriction of the Jones' action to the rotation subgroup $\Rot\act \scrM_0$ extends to an action on the von Neumann algebra $\scrM$.
\end{letterthm}
We observe that the action of the rotation subgroup $\Rot$ is not state-preserving. Moreover, we prove that under the given assumptions Thompson's group $F$ will fail to have an action on the von Neumann algebra $\scrM$, see Proposition \ref{prop:SingularF}.\\[0.25cm]
Let us conclude the introduction with the following general observations:\\[0.1cm]
This article is a first effort to use Jones' machinery in the context of lattice gauge theory and to provide a unification with ideas used in loop quantum gravity as well as operator-algebraic techniques from algebraic quantum field theory \cite{Haag-96-Local-Quantum-Physics}.
The degree of generality presented here already allows us to recover several results but there are many cases and generalization that are open today.
The ultimate goal being to replace the compact  group $G$ by a subfactor.
In future work we plan to settle the noncommutative compact group case and, more generally, to replace the group $G$ by a separable compact quantum group as it is a natural framework for our construction in the noncommutative situation. A first step in the treatment of the compact quantum group case can be found in our companion article {\cite[Remark 4.3]{Brot-Stottmeister-Phys}}.
This will potentially allow us to deal with richer kinds of tensor categories than representation categories of compact (abelian) groups.
Moreover, we plan to investigate higher dimensional models invoking a suitable replacement of Thompson's groups as well as necessary generalizations of our constructions from the point of view of renormalization group theory, see also {\cite[Section 5]{Brot-Stottmeister-Phys}}.\\[0.1cm]
{A different kind of potential generalization sticking to the 1+1-dimensional setting would involve the alteration of the states $\omega_{t}$ to account for more complicated interactions from a physical perspective. In view of the latter, our leading example, the heat-kernel state, contains only a ``kinetic'' term (given by the group Laplacian $\Delta_{G}$) which could be supplemented by a suitable ``potential'' term, for example, given by a multiplication operator corresponding to a function on $G$. One natural choice of such a function is a character, $\chi_{\pi} = \Tr_{V_{\pi}}(\pi(\ .\ ))$, w.r.t.~some representation $\pi:G\rightarrow\cU(V_{\pi})$. Up to an offset and a normalization, this choice directly reflects the ``potential'' (or magnetic) term of higher dimensional versions of the Kogut-Susskind Hamiltonian. Moreover, if we consider the spatial circle $\fS$ as embedded into two dimensions, we can interpret this choice as the restriction of the two-dimensional Kogut-Susskind Hamiltonian to a single plaquette. The latter interpretation is invoked in \cite{HuebschmannRudolphSchmidt09} where a detailed analysis of the gauge-invariant sector of the model including the exact solution of the spectral problem is given. Additionally, special emphasis is put on the singular nature of the classical reduced phase space and its counterpart on the quantum side. It is interesting to point out that the Gibbs states of this model a related to the Trotter-Kato product of the heat-kernel and the dual heat-kernel state on the trivial tree $t_{0}$ (one root and one leaf) in the weak-coupling limit, see \cite[Section 3]{Brot-Stottmeister-Phys}.}\\[0.1cm]
{Concerning gauge invariance and the generically singular reduced phase spaces of classical gauge theories, we point out the notion of Rieffel induction and its application to constrained quantum systems \cite{Landsman95}. As observed in \cite{Wren98}, Rieffel induction provides an operator-algebraic method to implement constraints in a quantum theoretical model that avoids potential ambiguities in the definition of domains of (unbounded) observables due to singularities in the reduced phase space. If the constraints are given in terms of an action by a compact group, as in the models discussed here, Rieffel induction will yield the gauge-invariant observables via restriction to an invariant subspace. But, if the constraints generate a non-compact group, the procedure will still apply making it a very interesting candidate to handle the gauge-invariant sector in the case of generalized symmetries.}\\[0.1cm]
An open problem, we have not investigated so far, is to determine which type III$_\lambda, \ 0\leq \lambda\leq 1$ factor we obtain in Theorem \ref{theo:intro-type}. 
But, we expect that there is a fairly precise correlation between $\lambda$ and the map $\beta$.
In this regard, we remark that the proof of $\scrM$ being a factor is rather indirect using a continuity argument, which works when $\beta$ is in $\ell^p(\D)$ for some $0<p<1/2.$
It appears very plausible that the correct assumption should be that $\beta$ is summable rather than $p$-summable.\\[0.1cm]
Finally, let us briefly comment on the relation of our construction to the one proposed by Jones in \cite{Jones18-Hamiltonian}, a more exhaustive discussion can be found in our companion article. In \cite{Jones18-Hamiltonian}, the concept of a scale-invariant respectively weakly scale-invariant transfer matrix given by collection of operators $T_t$ is used as additional data instead of that of a (projective) family of states $\omega_t$. Both concepts can be compared if the states $\omega_{t}$ are normal and, thus, are given by trace-class operators $T_{\omega_{t}}$: using the correspondence between the inverse temperature in statistical mechanics and Euclidean time, a formula for the Hamiltonian of the system can be deduced in both cases. The difference between (weak) scale invariance and projective consistency is further elucidated by renormalization group theory, see {\cite[Section 4]{Brot-Stottmeister-Phys}}. Let us also remark that (weak) scale invariance is known as projective consistency in the loop quantum gravity literature \cite{Thiemann-07-QuantumGR} and is exactly the property that implies the existence of the strong-coupling Hamiltonian $H$ referred to after theorem \ref{theo:intro-I}. 
Interestingly, weak scale invariance also appears as an essential ingredient in a recent analysis of Hamiltonian renormalization \cite{Lang-Liegener-Thiemann17}.\\[0.1cm] 
The article is organized in the following way:\\[0.1cm]
Beside the introduction this paper has three other sections.
In section \ref{sec:prelim}, we recall the correspondence between standard dyadic partitions, binary forests and dyadic rationals.
We explain Jones' machinery in our specific case giving a lattice theory and actions of Thompson's groups. 
In section \ref{sec:Calg}, we explain the general construction of lattice gauge theories in the C*-algebraic setting together with the action of Thompson's group $V$.
In section \ref{sec:VNA}, we introduce states and von Neumann algebras as GNS completions of C*-algebras.
After describing the general case, we specialize to abelian compact groups and work in a dual picture.
We obtain a precise structure theorem when we only deal with one measure.
We finally focus on the circle group case and a family of heat-kernel states.

\section{Preliminaries}\label{sec:prelim}

\subsection{The category of binary forests} 

A binary forest is an isotopy class of a disjoint union of finitely many planar rooted trees embedded in the strip $\R\times [0,1]$ of $\R^2$ with roots on $\N\times \{0\}$ and leaves on $\N\times \{1\}.$
Isotopies preserve the strip but can act on the boundary.
The trees are binary meaning that each vertex has zero or two descendants.
Note that there is a natural order on the roots and the leaves of a forest.
We will always count them from left to right.
%
\newcommand{\forest}{
\begin{tikzpicture}[baseline = .4cm]
\draw (0,0)--(0,1);
\draw (1,0)--(1,2/3);
\draw (1,2/3)--(2/3,1);
\draw (1,2/3)--(4/3,1);
\draw (2,0)--(2,1/3);
\draw (2,1/3)--(5/3,1);
\draw (2,1/3)--(7/3,1);
\draw (13/6,2/3)--(2,1);
\end{tikzpicture}
}
{Here is an example of a forest with three roots and six leaves:
$$f = \ \forest \ .$$
}
We confer the structure of a tensor category to these forests as follows:
 Let $\cF$ be the category with set of objects the natural numbers $\{1,2,3,\cdots\}$ and morphisms $\cF(n,m)$ from $n$ to $m$ equal to the set of binary forests with $n$ roots and $m$ leaves.
Consider two forests $g,f$ such that $f$ has $n$ leaves and $g$ has $n$ roots. 
Consider a representative of $f$ (resp. $g$) in the strip $\R \times [0,1]$ and assume that its $i$th leaf (resp. the $i$th root) is at the point $(i,1)$ (resp. $(i,0)$).
Shift the representative of $g$ by one unit vertically and consider the union of those two representatives that we contract vertically by a factor 2 to a subset of the strip $\R\times [0,1]$.
This gives us a new forest and its isotopy class is the composition $g\circ f$ sometimes denoted $gf$.
Roughly speaking, $gf$ is the vertical concatenation of $g$ on top of $f$ where the $i$th leaf of $f$ is lining up with the $i$th root of $g$. 
\newcommand{\funfun}{
\begin{tikzpicture}[baseline=.2cm, scale = .4]
\draw (0,0)--(0,-.5);
\draw (0,0)--(-1,2);
\draw (-.5,1)--(0,2);
\draw (0,0)--(1,2);
\end{tikzpicture}
}
\newcommand{\compo}{
\begin{tikzpicture}[baseline = -.2cm, scale = .6]
\draw (1,-2)--(1,-2.5);
\draw (1,-2)--(0,0);
\draw (.5,-1)--(1,0);
\draw (1,-2)--(2,0);
\draw (0,0)--(0,1);
\draw (1,0)--(1,2/3);
\draw (1,2/3)--(2/3,1);
\draw (1,2/3)--(4/3,1);
\draw (2,0)--(2,1/3);
\draw (2,1/3)--(5/3,1);
\draw (2,1/3)--(7/3,1);
\draw (13/6,2/3)--(2,1);
\end{tikzpicture}
}
{Consider the forest $f$ of the preceding figure and the tree
$$t = \ \funfun\ .$$
The composition $f\circ t$ is equal to:
$$\begin{small}\compo\ .\end{small}$$
}
\begin{notation}We denote by $I$ and $Y$ the trees with one and two leaves respectively, and we write $\fT$ for the set of all trees.
\end{notation}
We define a tensor product structure on this category:
The tensor product of objects is the addition of natural numbers, i.e.~$n\otimes m := n+m$, and the tensor product of forests (i.e.~morphisms) is the horizontal concatenation.
Meaning that if $f,g$ are forests with representatives inside $\R\times [0,1]$, we consider the union of those two representatives when the one of $f$ is placed on the left and say that $f\otimes g$ is the isotopy class of the union of those two representatives.

{For example, if $f$ and $t$ are the forest and tree of above, we obtain that 
$$t\otimes f = \funfun \quad \forest \ .$$}
Equipped with this tensor product we obtained a tensor category.
If $n\leq 1$ and $1\leq j\leq n$, we will write $f_{j,n}$ for the forest with $n$ roots and $n+1$ leaves whose $j$th tree is $Y$ and all the others are equal to $I$, hence $f_{j,n} = I^{\ot j-1} \ot Y \ot I^{n-j}.$
It is easy to see that any morphism is the composition of some $f_{j,n}$ and, thus, the tensor category $\cF$ is singly generated by $Y$ in the sense that any morphism is the composition of tensor products of $Y$ and the trivial morphism $I$.

\subsection{Dyadic rationals and partitions}\label{sec:Dyadic}
Let $t_\infty$ be the infinite binary rooted tree.
We decorate its vertices by intervals such that the root corresponds to the open interval $(0,1)$ and the successors of a vertex decorated by $(d,d')$ are decorated by $(d,\tfrac{d+d'}{2})$ to the left and $(\tfrac{d+d'}{2} , d' )$ to the right.
Note that the collection of all these decorations is given by the intervals of the form $(\tfrac{a}{2^n}, \tfrac{a+1}{2^n})$ with $a,n$ natural numbers and such that $\tfrac{a+1}{2^n}\leq 1.$
We call such an interval a standard dyadic interval (s.d.i.).
\newcommand{\tinf}{
\begin{tikzpicture}
\node at (0,-.25) {$(0,1)$};
\draw (0,0)--(-2,1);
\draw (0,0)--(2,1);
\node at (-2,1.25) {$(0,1/2)$};
\draw (-2,1.5)--(-3,2.5);
\draw (-2,1.5)--(-1,2.5);
\node at (2,1.25) {$(1/2,1)$};
\draw (2,1.5)--(1,2.5);
\draw (2,1.5)--(3,2.5);
\node at (-3,2.75) {$(0,1/4)$};
\node at (-1,2.75) {$(1/4,1/2)$};
\node at (1,2.75) {$(1/2,3/4)$};
\node at (3,2.75) {$(3/4,1)$};
\node at (-3,3.25) {$\cdots$};
\node at (-1,3.25) {$\cdots$};
\node at (1,3.25) {$\cdots$};
\node at (3,3.25) {$\cdots$};
\end{tikzpicture}
}
{Here is the beginning of this labelled tree:
$$\tinf\ .$$}
Let $\D$ be the set of all dyadic rationals inside the half open interval $[0,1)$ that is $$\D=\left\{ \dfrac{a}{2^n} :\ n\geq 1, 0\leq a\leq 2^n-1 \right\}.$$
We identify $[0,1)$ with the torus $\R/\Z$ and $\D$ with the set of dyadic rationals inside of it which explains why we remove the point $1$.

A standard dyadic partition (s.d.p.) of the unit interval is a finite sequence of dyadic rationals $0=d_1< d_2<\cdots<d_{n} <1=d_{n+1}$ such that 
$(d_j,d_{j+1})$ is a s.d.i.~for any $1\leq j\leq n.$
Note that this is strictly speaking not a partition since the union of those intervals does not contain the boundary points $d_j, 1\leq j\leq n+1$. 
But, with the aforesaid identification in mind, we can, for example, equally well consider half-open intervals with left boundaries included.

Consider a finite rooted binary tree $t\in\fT$ with $n$ leaves and view it as a rooted sub-tree of $t_\infty.$
We assign the decoration $(e_j,e_j')$ to the $j$th leaf of $t$ associated with the vertex in $t_\infty$ equal said leaf.
Note that we necessarily have that $e_j = e_{j-1}'=:d_j$ for $2\leq j\leq n+1$ and $e_1=0=:d_1$, $e_n'=1 =: d_{n+1}.$
Therefore, we obtain a sequence of dyadic rationals $d_2 < d_3 < \cdots < d_{n}$ which gives a partition of $(0,1)$ (minus finitely many points) given by $(d_1,d_2), (d_2, d_3),\cdots, (d_{n},d_{n+1})$ with each interval being a s.d.i.~and, thus, resulting in a s.d.p.
One can check that this process gives a bijection from $\fT$ to the set of s.d.p.~of the unit interval.
{For example, if 
$$t=\funfun,$$ 
then the associated s.d.p.~is 
$$\{ (0,1/4), (1/4,1/2), (1/2,1)\}.$$}

If it is convenient, we will consider the set of trees $\fT$ as a partially ordered set: $s\leq t$ iff there exists a forest $f$ such that $t=fs$.
This even makes $\fT$ a direct set since any two trees $s,t$ are smaller than the complete tree $t_m$ for sufficiently large $m$, where $t_m$ is the rooted binary tree with $2^m$ leaves all at distance $m$ from its root.
Note that this partial order corresponds to the binary relation between pairs of partitions where one is finer than the other.

\subsection{Thompson's groups}
Thompson's group $F$ is the group of piecewise linear homeomorphisms of the standard interval $[0,1]$ fixing the boundary points with integer powers of 2 as slopes and dyadic rationals for breakpoints.
We recall the description of $F$ as a fraction group, see \cite{Cannon-Floyd-Parry96} for this specific example and \cite{Belk04, Jones16-Thompson} -- we stick to the same formalism as the latter.
Note that similar constructions of groups in the formalism of cancellative semigroups were developed by Ore, see for instance \cite{Maltsev53}.
Consider the set of all couples $(t,s)$ of trees with the same number of leaves that we quotient by the equivalence relation $\sim$ generated by $(ft,fs)\sim (t,s)$ where $f$ is any forest having the same number of roots as the number of leaves of $s$ (and thus of $t$).
Denote by $\tfrac{t}{s}$ the equivalence class of $(t,s)$.
This quotient set, denoted by $G_\cF$, has a group structure with multiplication
$$\dfrac{t}{s} \dfrac{v}{u} = \dfrac{pt}{qu} \text{ where } ps = qv,$$
neutral element $e=\tfrac{s}{s}$ for any $s\in\fT$ and $\tfrac{t}{s}^{-1} = \tfrac{s}{t}.$
We say that $G_\cF$ is the group of fractions associated to the category $\cF$ with fixed object $1\in ob\cF.$
This group is isomorphic to Thompson's group $F$.
The reason is that an element of $F$ sends an adapted s.d.p.~onto another s.d.p.~as it is affine (w.r.t.~$\D$) on each of the s.d.i.~of the initial partition and will remain unchanged if we refine both partitions in the same way. 
Using the correspondence between s.d.p.~and trees we obtain a description of $F$ in terms of $G_\cF$.
\newcommand{\funfdeux}{
\begin{tikzpicture}[baseline = .4cm]
\draw (2,0)--(2,1/3);
\draw (2,1/3)--(5/3,1);
\draw (2,1/3)--(7/3,1);
\draw (13/6,2/3)--(2,1);
\end{tikzpicture}
}
{For example, the fraction $\tfrac{t}{s}$ where 
$$t=\funfun$$ and $$s=\funfdeux$$ is the unique element of $F$ sending $[0,1/2)$ onto $[0,1/4)$, $[1/2,3/4)$ onto $[1/4,1/2)$ and $[1/2,1)$ onto $[3/4,1)$ where the restriction of each of the three initial interval is linear with positive slope.}

Now, we consider the category of symmetric forests $\cSF$ with objects $\N$ and morphisms $\cSF(n,m)=\cF(n,m)\times S_m$ where $S_m$ is the symmetric group of $m$ elements. 
Graphically we interpret a morphism $(p,\tau)\in\cSF(n,m)$ as the concatenation of two diagrams.
On the bottom we have the diagram previously explained for the forest $p$ in the strip $\R\times [0,1]$. 
The diagram of $\tau$ is the union of $m$ segments $[x_i , x_{\tau(i)}+(0,1)], i=1,\cdots,m$, in $\R\times [1,2]$ where the $x_i$ label $m$ distinct points in $\R\times \{1\}$ such that $x_i$ is on the left of $x_{i+1}.$
The full diagram of $(p,\tau)$ is obtained by stacking the diagram of $\tau$ on top of the diagram of $p$ such that $x_i$ agrees with the $i$th leaf of $p$.
\newcommand{\perm}{
\begin{tikzpicture}[baseline = .4cm]
\draw (0,0)--(1,1);
\draw (1,0)--(2,1);
\draw (2,0)--(0,1);
\end{tikzpicture}
}
\newcommand{\taut}{
\begin{tikzpicture}[baseline = .4cm, scale=.5]
\draw (0,0)--(0,-.5);
\draw (0,0)--(-1,2);
\draw (-.5,1)--(0,2);
\draw (0,0)--(1,2);
\draw (-1,2)--(0,3);
\draw (0,2)--(1,3);
\draw (1,2)--(-1,3);
\end{tikzpicture}
}
{Here is one example. Consider the cyclic permutation $\tau$ of $\{1,2,3\}$ sending $1$ to $2$, $2$ to $3$ and $3$ to $1$.
We represent this permutation by the diagram:
$$\perm\ .$$
Let $t$ be the tree $\funfun$.
We have that $(t,\tau)$ is described by the diagram:
$$\taut \ .$$
}
Given two symmetric forests $(q,\tau)\in\cSF(n,m), (p, \sigma )\in \cSF(m,l)$, let $l_i$ be the number of leaves of the $i$th tree of $p$, then we define the composition of morphisms as follows:
$$(p,\sigma)\circ (q,\tau) := ( \tau(p) \circ q , \sigma S( p , \tau ) ),$$
where $\tau(p)$ is the forest obtained from $p$ by permuting its trees such that the $i$th tree of $\tau(p)$ is the $\tau(i)$th tree of $p$ and $S(p,\tau)$ is the permutation corresponding to the diagram obtained from $\tau$ where the $i$th segment $[x_i , x_{\tau(i)} + (0,1)]$ is replaced by $l_{\tau(i)}$ parallel segments.
\newcommand{\permtilde}{
\begin{tikzpicture}[baseline = .4cm]
\draw (0,0)--(1,1);
\draw (1,0)--(2,1);
\draw (2,0)--(3,1);
\draw (3,0)--(4,1);
\draw (4,0)--(5,1);
\draw (5,0)--(6,1);
\draw (6,0)--(0,1);
\end{tikzpicture}
}
\newcommand{\foresttilde}{
\begin{tikzpicture}[baseline = .4cm]
\draw (2,0)--(2,1);
\draw (0,0)--(0,2/3);
\draw (0,2/3)--(-1/3,1);
\draw (0,2/3)--(1/3,1);
\draw (1,0)--(1,1/3);
\draw (1,1/3)--(2/3,1);
\draw (1,1/3)--(4/3,1);
\draw (7/6,2/3)--(1,1);
\end{tikzpicture}
}
{For example, if we consider the forest 
$$f = \forest$$ and the permutation $$\tau=\perm\ ,$$ then 
$$(f,\id)\circ (I\ot I \ot I , \tau) = (\tau(f) , S(f,\tau))$$ 
where 
$$\tau(f) = \foresttilde$$ \text{ and } $$S(f,\tau) = \permtilde \ .$$
}

Thompson's group $V$ is isomorphic to the group of fractions of the category $\cSF$.
Hence, any element of $V$ is an equivalence class of a pair of \emph{symmetric} trees.
The equivalence relation being generated by $( (t,\tau) , (s,\sigma) ) \sim ((f,\phi)\circ (t,\tau) , (f,\phi)\circ (s,\sigma) ) )$ where $(f,\phi)$ is a symmetric forest.
Given an element $g=\tfrac{(t,\tau)}{(s,\sigma)} \in V$ and the s.d.p.'s $(I_1,\cdots,I_n)$ and $(J_1,\cdots,J_n)$ associated to $s$ and $t$ respectively, we have that $g$, viewed as a transformation of the unit interval, is the unique piecewise linear function with constant slope on each $I_k$ that maps $I_{\sigma^{-1}(i)}$ onto $J_{\tau^{-1}(i)}$ for any $1\leq i\leq n.$
Roughly speaking, $g$ is a  element of $F$ together with a permutation of the intervals.

If we consider the cyclic group $\Z/m\Z$ as a subgroup of the symmetric group $S_m$, we will have the subcategory $\cAF\subset \cSF$ of \emph{affine} forests for $\cAF(n,m)=\cF(n,m)\times \Z/m\Z$.
The group of fractions of $\cAF$ is isomorphic to Thompson's group $T$.
Recall that it is the the group of piecewise \text{affine-linear} homeomorphisms of the torus $\R/\Z$ with integer powers of 2 as slopes and dyadic rationals for breakpoints.
In particular, $T$ contains the subgroup generated by rotations about angles $2^{-n}$, $n\in\N$.
We typically treat $\cF$ and $\cAF$ as subcategories of $\cSF$ resulting in embeddings at the level of groups $F<T<V$.

Clearly, the classical action of $V$ on $[0,1]$ corresponds to an action on the set of dyadic rationals $\D$, which we freely make use.

\subsection{Actions of Thompson's groups}\label{sec:Actions}
If $(\cC,x)$ is a category together with a distinguishes object satisfying certain axioms, we can construct the group of fractions $G_\cC$.
If we additionally have a functor $\Psi:\cC\to \cD$, we can construct an action of the group of fractions $G_\cC$ depending on $\Psi$ following Jones \cite{Jones16-Thompson}.
Moreover, this action will inherit properties from the category $\cD$, e.g.~if $\cD=\Hilb$ is the category of Hilbert spaces with isometries as morphisms, the action will be conveyed by a unitary representation of $G_\cC$.\\[0.25cm]
We present this construction in detail for the category of forests $\cC=\cF$ with the distinguished object $1\in\N$ and, thus, $G_\cF$ is Thompson's group $F$.
Due to the focus of our article, we mainly work with the target category of Hilbert spaces (equipped with isometries) $\Hilb$ or of C*-algebras (equipped with injective *-morphisms) $\Calg$.

Consider a covariant functor $\Phi:\cF\to \cD$ where $\cD$ is a concrete category. 
Given a tree $t\in\fT$ with $n=\target(t)$ leaves, we define $X_t:=\{ (t,\xi):\ \xi \in \Phi(n)\}$ which is a copy of $\Phi(n)$ indexed by $t$.
We set 
$$X_\Phi:=\{(t,\xi):\ \xi\in \Phi(\target(t)) \}/\sim$$ where $(t,\xi)\sim (s,\eta)$ iff there exists forests $p,q$ such that $(pt,\Phi(p)\xi) = (qs,\Phi(q)\eta).$
It follows that $X_\Phi$ is the inductive limit $\varinjlim_{t\in\fT} X_t$ of the directed system of objects $(X_t,\ t\in\fT)$ over the directed set $(\fT,\leq)$ with connecting maps 
$$\iota_t^{ft}:X_t\to X_{ft}, (t,\xi)\mapsto (ft,\Phi(f)\xi).$$
We write $\tfrac{t}{\xi}$ or simply $(t,\xi)$ by slightly abusing notation for the equivalence class of $(t,\xi)$ inside the quotient space $X_\Phi$.
If the connecting maps are injective, we will identify $X_t$ with a subspace of $X_\Phi$ for any $t\in\fT$.
The set $X_\Phi$ admits an action $\alpha_\Phi:F\act X_\Phi$ that is given by the formula:
$$\alpha_\Phi\left( \dfrac{t}{s} \right) \dfrac{u}{\xi} := \dfrac{pt}{\Phi(q)\xi} \text{ where } {qu = ps},$$
for a pair of trees $(s,t)$ with the same number of leaves, $(u,\xi)\in X_{\Phi}$ and suitable forests $q,p$. 
We call $\alpha_{\Phi}$ the {\it Jones' action} (associated to $\Phi$).

\begin{remark}
Note that this construction can be interpreted as a certain Kan extension where the fraction group is viewed as a Quillen homotopy group, see {\cite[Appendix]{Brothier19WP}}. 
We thank Sergei Ivanov for pointing out this description of the Jones' action.\\[0.1cm]

A similar construction can be performed when the functor $\Phi:\cF\to \cD$ is contravariant instead of covariant. 
In that case, we obtain an inverse system of sets $X_t$ and an action of the fraction group (i.e.~Thompson's group $F$ in this case) on the inverse limit $\varprojlim_{t\in\fT} X_t$.
We use this construction for describing the space $\oscrA$ in Section \ref{sec:Commutative}.\end{remark}

In this article, we are only interested in tensor functors (or monoidal functors).
These have the advantage to be defined by very few data and that any Jones' action of $F$ will extend to an action of $V$ as explained below.\\[0.25cm]
Consider a concrete tensor category $\cD$ and let $\Phi:\cF\to \cD$ be a tensor functor.
We have that $\Phi(n)=\otimes_{k=1}^n \Phi(1)$ and thus $\Phi(n)$ is characterized by $\Phi(1).$
If $f_{j,n}$ is the forest with $n$ roots, $n+1$ leaves and with its $j$th tree having two leaves, we may deduce $\Phi(f_{j,n}) = \id^{\ot j-1}\ot \Phi(Y)\ot \id^{\ot n-j}$, where $\id$ is the identity of $\Phi(1).$
Since any morphism $f$ is a composition of such $f_{j,n}$, we find that $\Phi$ is characterized by the morphism $R:=\Phi(Y)\in \Hom_\cD(\Phi(1),  \Phi(1)\ot \Phi(1)).$
Conversely, an object $A$ in a concrete tensor category $\cD$ together with a morphism $R:A\to A\ot A$ defines a tensor functor from $\cF$ to $\cD.$
Given such a tensor functor, we obtain a Jones' action $\alpha_\Phi:F\act X_\Phi.$

In the context of tensor functors, we can always extend the Jones' action of $F$ to an action of the largest of Thompson's groups $V$ as follows.
Let $\theta:S_n\to \Aut(\Phi(n))$ be the action of the symmetric group defined by permutations of the tensors factors when $\Phi(n)$ is identified with $\ot_{i=1}^n \Phi(1)$.
More precisely, if $\eta=\eta_1\ot\cdots\ot\eta_n\in \ot_{j=1}^n\Phi(1) \simeq \Phi(n)$ and $\kappa\in S_n$, then $\theta( \kappa )(\eta) = \eta_{ \kappa^{-1} (1) } \ot\cdots\ot \eta_{ \kappa^{ -1 } ( n ) }.$
Consider $g:=\tfrac{(t , \tau)}{(s,\sigma)}\in V$ with trees $t,s$ having $n$ leaves and permutations $\tau,\sigma\in S_n$.
If $\tfrac{s}{\xi}$ is in $X_\Phi$, then we put
$$\alpha_{\Phi} \left( \dfrac{(t,\tau)}{(s,\sigma)} \right) \dfrac{s}{\xi} := \dfrac{t}{\theta(\tau^{-1}\sigma) \xi}.$$
This provides an action of Thompson's group $V$ on $X_\Phi$ which extends the Jones' action $\al_\Psi:F\act X_\Psi.$

\subsubsection{The category of Hilbert spaces}
Consider the category $\Hilb$ of (complex) Hilbert spaces with isometries as morphisms and a functor $\Phi:\cF\to \Hilb.$
Write $\fH=\Phi(1)$ and $\fH_t:=\{(t,\xi) : \ \xi\in\Phi(n)\}$ instead of $X_t.$
The inductive limit of (complex) vector spaces $\varinjlim_{t\in\fT} \fH_t$ has a pre-Hilbert structure given by the sesquilinear form
$$\langle \dfrac{t}{\xi} , \dfrac{s}{\eta} \rangle := \langle \Phi(p)\xi , \Phi(q) \eta \rangle \text{ where } {pt = qs}$$
with the inner product on the right hand side performed in the Hilbert space $\Phi(\target(p))$. 
This sesquilinear form is well-defined because the morphisms of $\Hilb$ are isometries.
We complete this pre-Hilbert space to a Hilbert space that we denote by $\scrH$, and we observe that the action of Thompson's group $F$ extends to a unitary representation $\pi_\Phi:F\to \cU(\scrH)$.
\\[0.1cm]
We may add a tensor product structure to $\Hilb$ with the classical notions of tensor products of Hilbert spaces and of operators.
Then, a Hilbert space $\fH$ together with an isometry $R:\fH\to \fH\ot \fH$ defines a tensor functor $\Phi_R:\cF\to\Hilb$ which gives an action of $V$ on $\varinjlim_{t\in\fT} \fH_t$ that extends to a unitary representation $\pi_{\Phi_R}:V\to\cU(\scrH)$.
See \cite{Brot-Jones18-1} for explicit examples of such functors and representations.
\\[0.25cm]
We can also define a different tensor product structure $\odot$ on $\Hilb$ by setting $H\odot K$ equal to the direct sum of the Hilbert spaces $H$ and $K$ and $A\odot B$ equal to the direct sum of the operators $A\in B(H), B\in B(K).$
Using this structure, a tensor functor $\Phi:\cF\to (\Hilb,\odot)$ is equivalent to the data consisting of a Hilbert space $\fH$ together with a pair of operators $A,B\in B(\fH)$ satisfying the Pythagorean equality:
$$A^*A+B^*B = \id_\fH.$$
Note that if $A=B^*$, then this is the CAR condition.
By defining tensor functors as follows: $\Phi(n)= \fH^{\oplus n}$ and $\Phi(Y)= A\oplus B$, this will induce a family of representations of Thompson's group $V$ called Pythagorean representations which were introduced and studied in \cite{Brot-Jones18-2}.

\subsubsection{The category of C*-algebras}
Consider the category of unital C*-algebras $\Calg$ with injective *-morphisms.
Let $\Psi:\cF\to \Calg$ be a functor and denote by $(B_t ,\ t\in\fT)$ the associated directed system with connecting maps $\iota_t^{ft}:B_t\to B_{ft}$ and (algebraic) inductive limit $B_\Psi$.
Since $\iota_t^{ft}$ is an injective *-morphism between C*-algebras it is automatically an isometry.
This implies that $B_\Psi$ is a *-algebra with a C*-norm $\| \cdot \|$, i.e.~$\| a^*a \| = \| aa^*\| = \|a\|^2$ for any $a\in B_\Psi$.
Therefore, the completion of $B_\Psi$ w.r.t.~this norm is a C*-algebra that we denote by $\mathscr B_0.$
Since the directed system of C*-algebras $(B_t,\ t\in\fT)$ admits a unique C*-limit, we may write $\varinjlim_{t\in\fT} B_t$ for the completion $\scrB_0$ and call it the direct limit of C*-algebras.
The general construction implies that we have an action $\al_\Psi:F\act B_\Psi$ by *-automorphism.
But, $\al_\Psi(g)$ is necessarily an isometry for $g\in F$ w.r.t.~the C*-norm and thus extends into an automorphism of the C*-algebra $\scrB_0.$
Therefore, the action $\al_\Psi$ extends in a unique way to an action $\al$ of $F$ on the C*-algebra $\mathscr B_0$ such that $\al(g)$ is a *-automorphism for any $g\in F$.

We equip the category $\Calg$ with the minimal tensor product $\ot_{\min}$ for the object and the associated tensor product of morphisms.  
A tensor functor $\Psi:\cF\to (\Calg,\ot_{\min})$ is defined by a C*-algebra $B$ and an injective *-morphism $R: B\to B\ot_{\min} B.$
By the above reasoning, this provides an automorphic action of Thompson's group $V$ on a C*-algebra $\scrB_0.$\\[0.25cm]
The following example will be useful later:\\[0.1cm]
Denote by $\otimes_{d\in\D}^{\min} B$ the unique C*-completion of the inductive limit of C*-algebras $\ot_{d\in E}^{\min} B$ where $E$ runs over the finite subsets of $\D$.

\begin{proposition}\label{prop:TensorProd}
Consider a unital C*-algebra $B$ and the map $$R:B\to B\ot_{\min} B, b\mapsto b\ot 1.$$
Let $\Psi:\cF\to (\Calg,\ot_{\min})$ and $\alpha_\Psi:V\to \Aut(\mathscr B_0)$ be the associated tensor functor and its Jones' action respectively.

There exists an isomorphism of C*-algebras 
$$J:\mathscr B_0 \to \otimes_{d\in\D}^{\min} B$$
such that $\al_\Psi$ is conjugated by $J$ to the generalized Bernoulli shift action associated with $V\act \D$, i.e.
$$\Ad(J) \alpha_\Psi(v) (\ot_{d\in\D} b_d) = \ot_{d\in\D} b_{v^{-1}d},$$
for any elementary tensor $\ot_{d\in\D} b_d$ and $v\in V$.
\end{proposition}

\begin{proof}
We prove this proposition in detail to illustrate the formalism.
Consider a tree $t$ with $n$ leaves and associated s.d.p.~$d_1=0<d_2<\cdots<d_n<d_{n+1}=1$.
Denote by $\D(t)$ the set $\{d_1,\cdots,d_n\}$ and $J_t:B_t\to \otimes_{d\in\D}^{\min} B$ the unital embedding induced by the inclusion $\D(t)\subset \D.$
Let $f$ be the forest with $n$ roots, $n+1$ leaves and its $j$th tree equal to $Y$.
Observe that $\iota_t^{ft}(x_1\ot \cdots \ot x_n) = x_1\ot \cdots x_{j}\ot 1 \ot x_{j+1}\ot\cdots\ot x_n$ for an elementary tensor $x=x_1\ot\cdots\ot x_n$.
Therefore, $J_{ft}\circ \iota_t^{ft}(x) = J_t(x)$ for any such $x$ which implies that $J_{ft}\circ \iota_t^{ft}=J_t$ by linearity and density.
Since any forest is a finite composition of such $f$, we infer that the family of maps $(J_t,\ t\in\fT)$ is compatible w.r.t.~the directed system $(B_t,\ t\in \fT)$ and, thus, gives rise to a densely defined map $J:\varinjlim_{t\in\fT} B_t\to  \otimes_{d\in\D}^{\min} B$.
As each $J_t$ is an injective C*-morphism and, thus, isometric, this implies that $J$ is isometric and extends to the C*-algebra $\scrB_0$ as an injective C*-morphism.
Now, any dyadic rational $d\in\D$ appears in the s.d.p.~of a certain tree.
Therefore, any elementary tensor of the algebraic tensor product of the $B$ over $\D$ appears in the range of $J$ implying that $J$ has dense range and, thus, is surjective because it is a morphism between C*-algebras.

Consider an element  $v=\tfrac{s}{t}$ of Thompson's group $F$ and $x=x_1\ot\cdots\ot x_n\in B_t\subset \scrB_0$.
Denote by $d'_1=1<d'_2<\cdots<d_n'<d'_{n+1}=1$ the s.d.p.~associated to $s$.
Then $J(\al_\Psi(v) \tfrac{t}{x}) = J(\tfrac{s}{x}) = \ot_{d\in\D} \hat x(d)$ where 
$$\tilde x(d)= \begin{cases} x_k \text{ if } d=d'_k \text{ for some } 1\leq k\leq n\\ 1 \text{ otherwise }\end{cases}.$$
In particular, the vector $\tfrac{t}{x}\in \scrB_0$ and $v= \tfrac{s}{t}$ satisfies the formula of the proposition.
By density and directedness of $\fT$, this formula is valid for any $v\in F.$
To prove it for the largest of Thompson's groups $V$, it is sufficient to consider a permutation of a fixed s.d.p.~and an elementary tensor that can be decomposed w.r.t.~this s.d.p.:\\
For a permutation $\sigma\in S_n$, the unique element $v\in V$ sending $(d_j,d_{j+1})$ onto $(d_{\sigma(j)}, d_{\sigma(j)+1})$ for any $1\leq j\leq n$ in an affine way, and an elementary tensor $x=\ot_{j=1}^n x_j \in B^{\ot n}$, we have that $J(\al_\Psi(v) \tfrac{t}{x}) = \ot_{d\in\D} y(d)$ with $y(d) = \begin{cases} x_k \text{ if } d= d_{\sigma(k)} \text{ for some } 1\leq k\leq n\\ 1 \text{ otherwise }\end{cases}.$
Clearly, the formula of the proposition is satisfied which finishes the proof.
\end{proof}

\subsubsection{States and von Neumann algebras}\label{sec:state-prelim}
Assuming that we have a functor $\Psi:\cF\to \Calg$, we consider the directed system $(B_t,\ t\in\fT)$ with inclusion maps $\iota_t^{ft}:B_t\to B_{ft}$ and the C*-inductive limit $\mathscr B_0$.
We intend to have a state on $\mathscr B_0$ and to consider the von Neumann algebra induced from it.
To this end, we consider a family of states $\omega_t:B_t\to \C, t\in\fT$ such that $\iota_t^{ft}$ is state-preserving for any $t$ and $f$.
Then, there is a unique state $\varpi$ on $\mathscr B_{0}$ such that $\varpi(b) = \omega_t(b)$ if $b\in B_t$ because of uniform boundedness of the family $\omega_t$, where we identify $B_t$ with a C*-subalgebra of $\scrB_0$.
Performing the GNS construction for $(\mathscr B_0,\varpi)$ results in a triple $(\pi,\fH,\Omega)$, where $(\pi,\fH)$ is a *-representation of $\mathscr B_0$ and $\Omega$ is a cyclic vector (the vacuum vector) satisfying that $\langle \pi(b)\Omega , \Omega \rangle = \varpi(b)$ for any $b\in\mathscr B_0.$
Let $\mathscr B$ be the von Neumann algebra given by the weak completion $ \pi(\mathscr B_0)''$. 
We continue to denote by $\varpi$ the unique normal extension of the limit state to $\scrB$ via $\Omega$.\\[0.25cm]
Two important facts to keep in mind are:\\[0.1cm]
First, assume that each state $\omega_t$ is faithful.
This implies that the GNS representation $\pi$ of $\scrB_0$ is faithful.
Indeed, the associated GNS representation $\pi_t$ of $B_t$ is faithful and thus isometric since $B_t$ is a C*-algebra.
Since the GNS representation $\pi$ of $\scrB_0$ restricted to $B_t$ contains $\pi_t$ we obtain that $\pi$ is isometric when restricted to $B_t$ and is thus isometric by density on the whole algebra $\scrB_0$ and hence faithful, cp.~\cite[Proposition II.8.2.4]{Blackadar-06-Operator-Algebras}.
However, we don't have in general that $\varpi$ is faithful on the von Neumann algebra $\scrB.$
We will encounter such an example in Section \ref{sec:Gabelian}.\\[0.1cm]
Second, we always have a Jones' action of $F$ (and, therefore, $V$ in case of a tensor functor) on the C*-algebra $\mathscr B_0$.
However, this action might not be (asymptotically) state-preserving and in general does not extend to the von Neumann algebra $\scrB$. 
In analogy with classical probability, we need a condition on the action of $F$ resembling the quasi-invariance of measures.
This situation happens in Section \ref{sec:Zleaves} where we have that, under certain conditions, the subgroup of dyadic rotations inside Thompson's group $T$ acts on the von Neumann algebra, while Thompson's group $F$ does not.

\section{Gauge theory: construction of a C*-algebra}\label{sec:Calg}

\subsection{Directed systems of C*-algebras and actions of Thompson's groups}

We fix a separable compact group $G$ with Haar measure $m_G.$
Let $B(L^2(G))$ be the von Neumann algebra of all bounded linear operators on the Hilbert space $L^2(G)$ of complex-valued square integrable functions on $G$ {w.r.t.~the Haar measure $m_G$.}
{We define a C*-algebra $M$ associated to $G$ which will be the basic building block of our construction of a kinematic model.}\\
By $\la:G\act L^2(G)$, we denote the left regular unitary representation of $G$, and write $N:=\overline{\Span(\la_g: g\in G ) }$ for the C*-subalgebra of $B(L^2(G))$ generated by the $\la_g$ that we may write $\overline{\C G}.$
{Hence, $N$ is the norm completion in $B(L^2(G))$ of the set of operators of the form $\sum_{g\in G} a_g\la_g$ where $a_g\in\C$ is equal to zero for all but finitely many $g\in G$.}
{The unit is given by $\la_e$ where $e$ is the neutral element of $G$. This operator belongs to the reduced C*-algebra of $G$ precisely when $G$ is discrete.}
We define $Q:=\scrC(G)$, i.e.~$Q$ is the commutative C*-algebra of complex-valued continuous functions on $G$, which we identify with the C*-subalgebra of $B(L^2(G))$ that acts by pointwise multiplication.
Observe that the von Neumann algebra generated by $N$ and $Q$ is equal to $B(L^2(G))$.
Moreover, the unitary operator $\la_g$ normalizes $Q$ for any $g\in G.$
Let $M$ be the C*-subalgebra of $B(L^2(G))$ generated by $\overline{\C G}$ and $\scrC(G)$.
We observe that the latter {is the norm completion inside $B(L^2(G))$ of} the algebraic crossed product 
$$\scrC(G)\rtimes_{\alg} G:= \{\sum_{g\in G} a_g \la_g \vert \ a_g \in \scrC(G), \supp(a) \text{ finite } \},$$ 
where $\supp(a)$ denotes the {the support of $a$} that is the set of $g\in G$ for which $a_g\neq 0.$
{Note that} this crossed product is in general not the classical C*-algebraic reduced crossed product unless $G$ is discrete.
As pointed out by Siegfried Echterhoff, we have the following short exact sequence:
$$0\to \overline{\C G}\to M \to \scrC(G) \to 0.$$
{In the definition of $M$ the group $G$ is involved twice: On the one hand, it underlies the definition of the space of continuous functions $\scrC(G)$ and, on the other hand, it is the basis of the group ring or convolution algebra $\C G$. The first invokes the topology of $G$ while the second only uses the algebraic structure of $G$. For this reason, we use the notation $G_d$ in the second case to express the fact that we forget the topology of $G$ and consider it as a discrete group. To this end, we denote our algebra $M$ as the crossed-product $\scrC(G)\rtimes G_d$.\\
Observe that this is only a notation and in general $M$ is not the reduced crossed-product of the discrete group $G_d$ acting on the C*-algebra $\scrC(G)$. Indeed, the later will not be separable if $G$ is uncountable as the first one $M$ is always separable since it is a C*-subalgebra of $B(L^2(G))$ where $G$ is separable.
Similar constructions and notations were previously considered, see \cite{Grundling96}.}

{Note that in those cases in which we deal with (normal) states and von Neumann algebras we may choose any weakly dense *-subalgebra of $B(L^2(G))$, see Remark \ref{rem:subalg}. In this sense, our choice of C*-algebra $M$ is minimal to include all unitaries $\lambda_{g}$, $g\in G$ and multiplication operators $a\in C(G)$.}

From now on, all tensor products of C*-algebras are considered to be minimal tensor products. 
But, as we will mainly be dealing with nuclear C*-algebras, the choice of tensor product will mostly not matter.\\[0.25cm]
We define the morphism 
\begin{equation}\label{equa:R}R:M\to M\ot M, R(b)=\Ad(u)(b\ot \id) \text{ where } u\xi(g,h):=\xi(gh,h) ,\end{equation}
for all $\xi\in L^2(G\times G)$. 
This defines a tensor functor 
$$\Psi:\cF\to (\Calg,\ot_{\min}), \Psi(1):=M, \Psi(Y)=R.$$
By the previous section, we obtain the direct limit of C*-algebras 
$$\scrM_0:=\varinjlim_{t\in\fT} M_t$$ 
together with an action of Thompson's group: $\al_{\Psi}:V\act \scrM_0.$

{The motivation for choosing $M$ as the norm completion of the algebraic crossed-product is threefold: First,
$M$ is a unital C*-algebra. Second, the map $R:M\to M\ot M$ defined in \eqref{equa:R} has a simple description when an operator $a$ is written as a discrete sum $\sum_{g\in G} a_g\la_g$ with $a_g\in\scrC(G)$. Third, the inductive limit algebra $\scrM_0$ admits a very explicit description in terms of a discrete crossed product such as $M$, see Section \ref{sec:M_0}.}

\begin{lemma}\label{lem:R}
We have the following equalities:
\begin{equation}
\label{eq:ymrefine}
R(\sum_{g\in G} a_g \la_g) = \sum_{g\in G} ( a_g\circ \mu_G) \la_{g,e},
\end{equation}
where $\mu_G$ is the group multiplication of $G$ and $g\mapsto a_g$ has finite support with values in $\scrC(G).$

In particular, we have that $R(\scrC(G))\subset \scrC(G)\ot \scrC(G), \ R(\overline{\C G})\subset \overline{\C G}\ot \overline{\C G}$ and $R(a)=a\circ\mu_G, R(b) = b\ot \id$ for any $a\in\scrC(G), b\in \overline{\C G}.$
\end{lemma}

\begin{proof}
This is a routine computation.
\end{proof}

\begin{remark}
The identities above are still valid if $a\in L^\infty(G)$ is in the von Neumann algebra of bounded measurable maps on $G$ and $b\in LG$ is in the group von Neumann algebra of $G$. 
Clearly, $M\ot M$ embeds densely (for the weak topology) inside $B(L^2(G\times G))$. 
We also have the following isomorphisms at our disposal: 
$$\otimes_{k=1}^n \scrC(G)\simeq \scrC(G^n) \text{ and } \ot_{k=1}^n \overline{\C G}\simeq \overline{ \C [G^n]  }.$$ 
The unitary $u$ in \eqref{eq:ymrefine} is an analogue of the structure operator in \cite{TakesakiII}, and it appears naturally as multiplicative unitary in the context of quantum groups \cite{Timmermann-08-Quantum-Groups}. \\[0.1cm]
At this point, it is crucial to work with unital C*-algebras and unital maps in order to have $R(Q)\subset Q\ot Q$ and $R(N)\subset N\ot N$. But, as explained in the introduction, unitality is also natural from the point of view of the renormalization group, cp.~{\cite[Section 4]{Brot-Stottmeister-Phys}}.
\end{remark}

The preceding lemma tells us that we can define the tensor functors $\Psi_Q,\Psi_N:\cF\to\Calg$ with the same $R$ but such that $\Psi_Q(n)$ is the $n$th tensor power of $Q:=\scrC(G)$ and $\Psi_N(n)$ the $n$th tensor power of $N:=\overline{\C G}$.
From these two functors we obtain two directed systems of C*-algebras $(Q_t,  \ t\in\fT)$ and $(N_t,  \ t\in\fT)$ with limits $\scrQ_0$ and $\scrN_0$ that we identify with C*-subalgebras of $\scrM_0$.
Moreover, these C*-subalgebras are stable under the action of $V.$
$\scrN_0$ and $\scrQ_0$ form two models that are interesting in their own right, which we refer to as the convolution and the commutative part respectively.

\begin{notation}
If $t$ is a tree with $n$ leaves, we will write $G_t$ for a copy of $G^n$ indexed by leaves of $t$.
Moreover, we will write $Q_t=\scrC(G_t), N_t:=\overline{\C [G_t]},$ and $M_t:=\scrC(G_t)\rtimes G_{t,d}$ {where the label $d$ refers fact that we think of $G_t$ in a discrete way.}
\end{notation}

\subsection{Description of the commutative part $\scrQ_0$.}\label{sec:Commutative}
We start by describing the limit algebra $\scrQ_0$ at the group level.\\[0.1cm]
To this end, we consider the contravariant tensor functor $\Xi:\cF\to \Set$ defined as $\Xi(1)=G$ and $\Xi(Y)=\mu_G$ that is the group multiplication.
Here, $\Set$ denotes the category of sets equipped with the tensor product structure given by cartesian products.
This functor provides an inverse system $(G_t, t\in\fT)$ with maps $p_s^{t}:G_{t}\to G_s$ if $s\leq t$ (i.e.~$t=fs$ for some forest $f$) defined by the group multiplication.
The inverse limit is defined as:
$$\varprojlim_{t\in\fT} G_t := \{ x=(x_t)_{t\in \fT} \in \prod_{t\in \fT} G_t : \ \Xi(f)(x_{ft}) = x_t ,\ t\in\fT, f\in\cF\}.$$
We have the following Jones' action 
$$\alpha_\Xi:V\act \varprojlim_{t\in\fT} G_t, \al_\Xi(v)(x) =y,$$
where $y$ is the unique element satisfying $$y_t(I) = x_t(v^{-1} I) \text{ if } I\in\cI(t) \text{ and } v^{-1}|_I \text{ is affine} ,$$
and $\cI(t)$ is the s.d.p.~associated to the tree $t\in\fT.$
Another way to describe the action for the group $F$ is:
$$(\alpha_\Xi(\dfrac{t}{s})x)_r = \Xi(p)(x_{fs}) \text{ where } ft = pr.$$
Next, we introduce a slightly different description of the inverse limit which we mostly use.
Consider the product space $\prod_{I \ \textup{s.d.i.}} G$, i.e.~the product of copies of $G$ indexed by all s.d.i.~intervals, and let $\oscrA$ be the subset of this product equal to all $x$ satisfying $x(I) = x(I_1)x(I_2)$ for any s.d.i.~$I$ where $I_1$ and $I_2$ are its first and second half respectively.
Then, we have the following bijection 
$$j:\oscrA\to \varprojlim_{t\in\fT} G_t, j(x)_t := x|_{\cI(t)}$$
with inverse map
$$j^{-1}(y)(I) = y_t(I) \text{ if } I\in \cI(t).$$
The Jones' action becomes $$\al_\oscrA:V\act \oscrA, \al_\oscrA(v) x (I) = x(v^{-1} I)$$
if $v^{-1}|_I$ is affine.
If $v^{-1}$ is not affine on $I$, then we break $I$ into smaller s.d.i.~$I_1,\cdots,I_k$ on which $v^{-1}$ is affine and then set $\al_\oscrA(v) x (I) := x(v^{-1} I_1)\cdots x(v^{-1} I_k)$.\\[0.1cm]
Recall that the space $\oscrA$ corresponds to a space of generalized connections of a principal $G$-bundle over the torus $\R/\Z$ in view of classical, geometric models of gauge theory, see, for example, \cite{Velhinho-05-Functorial-Aspects, Thiemann-07-QuantumGR}.\\[0.25cm]
It is possible to introduce additional structures on the inverse system of $G_t$, cp.~\cite{Ashtekar-Lewandowski-94-Representation-Theory, Ashtekar-Lewandowski-95-Projective, Ashtekar-Lewandowski-95-Projective-Differential}.
We consider $G_t$ as a compact topological space (with its usual topology) and we equipped it with its Borel $\sigma$-algebra and its Haar measure $m_{G_t}$.
Since the group multiplication $\mu_G:G\times G\to G$ is continuous and p.m.p., we assert that all maps $\Xi(f)$ are continuous and p.m.p.
Thus, the inverse limit $\varprojlim_{t\in\fT} G_t$ inherits a topology that makes it a compact topological space.
Moreover, let $p_s:\varprojlim_{t\in\fT} G_t\to G_s$ be the projection built from all $p_s^{r}$ with $r \geq s$.
Then, there exists, by the Riesz-Markov-Kakutani theorem, a unique probability measure $\overline m$ on $\varprojlim_{t\in\fT} G_t$ satisfying that each projection $p_s:(\varprojlim_{t\in\fT} G_t , \overline m) \to (G_s,m_{G_s})$ is p.m.p.~because each map $p_s^r, r\geq s$ is p.m.p.
By transferring the topology and the probability measure to $\oscrA$ via the map $j$, we obtain the restriction of the product topology of $\prod_{I\ \textup{s.d.i.}} G$ to $\oscrA$ and a probability measure that we denote by $m_\oscrA$.
In particular, we obtain that Thompson's group $V$ acts by p.m.p.~homeomorphisms on $\oscrA$. 
We now describe the limit and the Jones' action at the C*-algebra level, cp.~\cite{Ashtekar-Lewandowski-95-Projective-Differential} for the gauge theory context as well as \cite{Blackadar-06-Operator-Algebras} for general statements.

\begin{proposition}\label{prop:VQaction}
We have an isomorphism of C*-algebras $J:\scrQ_0\to \scrC(\oscrA)$ such that $J(\scrC(G_t)) = {p_t}_*(\scrC(\oscrA))$ that we continue to identify with $\scrC(G_t).$
Moreover, the Jones' action $\al:V \act \scrM_0$ restricts to $\scrC(\oscrA)$ is given by the following formula:
$$v\cdot b(x) = b(v^{-1} x),\ v\in V, b\in\scrC(\oscrA), x\in\oscrA.$$
\end{proposition}

\begin{proof}
The functor $\Psi_Q:\cF\to \Calg$ is obtained by composition of the functor $\Xi$ and the contravariant functor $\scrC:X\to \scrC(X)$ that sends a compact topological space $X$ to the C*-algebra of continuous function $\scrC(X)$.
This implies that the limit $\scrQ_0$ is isomorphic to $\scrC(\oscrA)$ via the map 
$$J:\scrQ_0\to \scrC(\oscrA), \dfrac{t}{b}\mapsto b\circ p_t.$$
The formula for the Jones' action on $\scrC(\oscrA)$ is found from the expression for the Jones' action on $\oscrA$ and the functor $\scrC$.
\end{proof}

\begin{remark}
It is important to note that $\oscrA$ does not have a natural group structure unless the group $G$ is abelian. 
Clearly, this is due to the fact that the group multiplication $\mu_G$ is a group morphism if and only if $G$ is abelian and, thus, the pointwise multiplication will only provide a group structure on $\oscrA$ in that latter case.
\end{remark}

\subsection{Description of the convolution part $\scrN_0$}\label{sec:Convolution}
This time, we consider the group $\prod_\D G$ of all maps from $\D$ to $G$ with the group law given by the pointwise multiplication. We denote by $\supp(g)$ the support of $g$, i.e.~the set of $d\in\D$ such that $g(d)\neq e$.
Let $\oplus_\D G$ be the subgroup of $\prod_\D G$ of finitely supported elements and $\C[ \oplus_\D G]$ its group ring. 
The action $V\act \D$ gives an action $V\act \oplus_\D G$ via generalized Bernoulli shifts.\\[0.1cm]
For a tree $t\in\fT$ with associated s.d.p.~$d_1=0<d_2<\cdots<d_n<d_{n+1}=1$, we have a embedding of groups $\iota_t :G_t\to \oplus_\D G$ given by the inclusion $\D(t):=\{d_1,\cdots,d_n\}\subset \D$ such that $G_t$ becomes the subgroup of maps supported in $\D(t)$.

\begin{proposition}
There is an isomorphism $J:\scrN_0\to \otimes_\D \overline{\C G}$ where $\otimes_\D \overline{\C G}$ is the infinite minimal tensor product of the C*-algebra $\overline{\C G}$ over the set $\D.$
Moreover, the Jones' action $\Ad(J)\circ\al:V\act \otimes_\D \overline{\C G}$ is given by generalized Bernoulli shifts: 
$$\Ad(J)\circ\al(v)(\ot_{d\in\D} b_d) = \ot_{d\in\D} b_{v^{-1}d} , \text{ for any elementary tensor } \ot_{d\in\D} b_d \text{ and }  v\in V.$$
\end{proposition}

\begin{proof}
Lemma \ref{lem:R} shows that $R(a)=a\ot \id$ for any $a\in N:=\overline{\C G}.$
Therefore, Proposition \ref{prop:TensorProd} implies that we have an isomorphism $J:\scrN_0\to \otimes_\D \overline{\C G}$ such that the conjugate of the Jones' action is the desired action by generalized Bernoulli shifts.
\end{proof}

The group $\oplus_\D G$ together with the action $V\act \oplus_\D G$ can be interpreted as the result of Jones' construction for the tensor functor: 
$$\Upsilon:\cF\to \Gr \text{ such that } \Psi(1)=G \text{ and } \Upsilon(Y)(g) = (g,e), g\in G,$$
where $\Gr$ is the category of groups whose tensor product structure is given by the direct products of groups and group morphisms. 
This way, the classical topology on $\oplus_\D G$ is the final one corresponding to the restriction of the box topology of $\prod_\D G$.
By composing this functor with the functor $H\mapsto \C[H]$ that associates with a group its group ring *-algebra, we obtain a functor similar to $\Psi_N$. The difference being that we do not take a completion in this scheme. 

The group ring $\C [\Gd]$ is a dense *-subalgebra of $\otimes_\D \overline{\C G}$ and we have that $J(\la_{G_t}(g)) = u(\iota_t(g)), g\in G_t, t\in\fT$ where $u:\Gd\to \otimes_\D \overline{\C G}$ is the canonical embedding.
We have the following identity for the action of $V$ on the group ring that is another way to define the Jones' action on $\scrN_0$:
$$[\Ad(J)\al(v)](u_g) = u_{vg}, g\in\Gd, v\in V.$$ 

\subsection{Description of the full inductive limit $\scrM_0$.}\label{sec:M_0}

Following the description of $\scrQ_0$ and $\scrN_0$, we consider $\scrM_0$ and the interplay between these two subalgebras.
We freely identify $\scrQ_0$ with $\scrC(\oscrA)$ and $\scrN_0$ with $\bigotimes_{d\in\D}^{\min} \overline{\C G}$. 
As before, we denote by $u:\Gd\to \scrN_0$, $g\mapsto u_g$ the embedding of the group $\Gd$ inside the unitary group of $\scrN_0.$
We observe that the group $\oplus_\D G$ admits an action on $\oscrA$ by:
\begin{equation}\label{equa:GA}g\cdot x(d,d') = g(d) x(d,d'), \ (d,d') \text{ a s.d.i.},\ g\in \oplus_\D G,\ x\in \oscrA.\end{equation}
This action does not restricts to an action of $\oplus_\D G$ on the spaces $(G_t,m_{G_t})$, $t\in\fT$. Nevertheless, for any $g\in\oplus_\D G$, we find $t_{g}\in\fT$ s.t.~the map $g\cdot:(G_{s},m_{s})\rightarrow (G_{s},m_{s})$ is well-defined for $s\geq t_{g}$ as well as continuous and p.m.p.
This way, we have a family of continuous actions on $\oscrA$ by the subgroups $\oplus_{\D(t)} G$, where $\D(t)=\{d_1,\cdots,d_n\}$ is the set of boundary points associated to the s.d.p.~of the tree $t$, compatible with the inverse system of the $(G_t,m_{G_t})$ given by the contravariant functor $\Xi$ defined in Section \ref{sec:Commutative}.
This implies that the group $\oplus_\D G$ acts by p.m.p.~homeomorphism on $(\oscrA,m_\oscrA)$.

Here, we are only interested in knowing that each element $g\in\oplus_\D G$ acts continuously and p.m.p. on $\oscrA.$ We do not require any statement regarding the continuity of the morphism $\oplus_\D G\to \Aut(\oscrA)$. Therefore, we have so far not defined any topology on $\Aut(\oscrA)$.
We denote by $\scrC(\oscrA)\rtimes_{\alg} \Gd$ the *-algebra equal to the algebraic crossed-product for this action generated by the group ring $\C[\Gd]$ and $\scrC(\oscrA)$.

\begin{theo}\label{theo:MzeroNC}
The *-algebra $\scrC(\oscrA)\rtimes_{\alg} \Gd$ embeds as a dense *-subalgebra of $\scrM_0.$

There exists a unique faithful *-representation $$\pi:\scrM_0\to B(L^2(\oscrA,m_\oscrA))$$ satisfying 
$$\pi(bu_g)\xi(x) = b(x) \xi(g^{-1} x),\ b\in\scrC(\oscrA), g\in\Gd, \xi \in L^2(\oscrA,m_\oscrA), x\in \oscrA.$$

The Jones' action $\al:V\act\scrM_0$ implemented by the functor $\Psi$ satisfies the formula:
$$\al(v) (b u_g) = b(v^{-1} \cdot ) u_{vg},\ v\in V, b\in \scrC(\oscrA), g\in \Gd$$
and is implemented by the unitaries $U_v\in U(L^2(\oscrA,m_\oscrA)),v\in V$ defined by the formula: 
$$U_v\xi(x):= \xi(v^{-1} x),\ \xi \in L^2(\oscrA,m_\oscrA), x\in \oscrA.$$
\end{theo}

\begin{proof}
Consider a tree $t$ with associated s.d.p.~$\dn$, let $\D(t)=\{d_1,\cdots,d_n\}$ be the set of endpoints and let $\cI(t)$ be the set of s.d.i.~$(d_j,d_{j+1}), 1\leq j\leq n$.
Recall that $M_t=\scrC(G_t)\rtimes G_{t,d}$ and identify the first and second copy of $G_t$ with $\oplus_{\cI(t)} G$ and $\oplus_{\D(t)}G$ respectively, thus $M_t\simeq \scrC(\oplus_{\cI(t)} G)\rtimes_{\alg} \oplus_{\D(t)}G.$
Observe that $\oplus_{\D(t)}G$ acts on $\oplus_{\cI(t)} G$ via the formula \eqref{equa:GA}.
In this way, we can identify the dense *-subalgebra $\scrC(\oplus_{\cI(t)} G)\rtimes_{\alg} \oplus_{\D(t)}G$ as a *-subalgebra of $\scrC(\oscrA)\rtimes_{\alg} \Gd$.
By the preceding results, this collection of embeddings inside $\scrC(\oscrA)\rtimes_{\alg} \Gd$ respects both directed systems which implies that $\scrC(\oscrA)\rtimes_{\alg} \Gd$ embeds inside $\scrM_0$.
Moreover, $\scrC(\oscrA)\rtimes_{\alg} \Gd$ contains the union of the $\scrC(\oplus_{\cI(t)} G)\rtimes_{\alg} \oplus_{\D(t)}G$ for $t\in\fT$ which is a dense subalgebra of $\scrM_0$ and, thus, $\scrC(\oscrA)\rtimes_{\alg} \Gd$ is dense inside $\scrM_0$.\\[0.1cm]
Next, we show that we have a directed system of faithful representations.
Consider $\pi_t:M_t\to B(L^2(G_t))$ the obvious representation coming from the definition of $M_t$, which is faithful by definition.
Observe that we have a directed system of Hilbert spaces $(L^2(G_t,m_{G_t}), t\in\fT)$ with inclusion maps $w_t^{ft}:L^2(G_t,m_{G_t})\to L^2(G_{ft},m_{G_{ft}})$ given by $w_t^{ft}(\xi)(x):= \xi(p_t^{ft}(x)), \xi \in L^2(G_t,m_{G_t}), x\in G_{ft}$.
Moreover, the maps $w_t^{ft}$ are isometries since the projections $p_t^{ft}:(G_{ft},m_{G_{ft}})\to (G_t, m_{G_t})$ are p.m.p.
The limit of this system is $L^2(\oscrA,m_\oscrA)$.
Recall that $\iota_t^{ft}:M_t\to M_{ft}$ is the embedding given by the functor $\Psi:\cF\to \Calg.$
We have the following compatibility condition:
$$w_{t}^{ft} ( \pi_t (a_t) \xi_t ) = \pi_{ft} ( \iota_t^{ft} (a_t) ) w_t^{ft}(\xi_t),$$
for any $a_t\in M_t$ and vector $\xi_t\in L^2(G_t)$.
Therefore, we can define a *-representation on the direct limit:
$\pi:\varinjlim_{t\in\fT} M_t\to B(L^2(\oscrA,m_\oscrA))$ satisfying $\pi_t(a_t)\xi_t = \pi(a_t)\xi_t$ for any $a_t\in M_t$ and $\xi_t\in L^2(G_t).$
Since each $\pi_t$ is injective, it is an isometry and, thus, $\pi$ extends to a faithful representation on the C*-algebra completion $\scrM_0.$
For any $t\in\fT$, $b\in\scrC(\oplus_{\cI(t)} G), g\in \oplus_{\D(t)} G$ and $\xi \in L^2(G_t)$, we have that $\pi(bu_g)\xi = b \xi(g^{-1} \cdot)$ and by density the formula of the theorem.

The properties of the Jones' action are a consequence of the last two subsections. The unitary implementation is checked by a routine computation.
\end{proof}

The construction of $L^2(\oscrA,m_\oscrA)$ and the representation $v\in V\mapsto U_v \in \cU(L^2(\oscrA,m_\oscrA)$ can be interpreted as a Jones' representation. 
Indeed, consider the contravariant tensor functor $\tilde\Xi:\cF\to \Prob$, where $\Prob$ is the category of probability measure spaces with direct product for tensor product satisfying $\tilde\Xi(1) = (G,m_G)$ and $\tilde\Xi(Y)$ being the group multiplication. Compose it with the contravariant tensor functor $L^2:\Prob \to \Hilb$ that sends a probability measure space $(X,\nu)$ to $L^2(X,\nu)$.
Then, $v\in V\mapsto U_v$ is the Jones' action induced by the covariant tensor functor $L^2\circ \tilde\Xi:\cF\to \Hilb.$

\begin{remark}
\label{rem:gausslaw}
The representation $\pi$ of $\scrM_{0}$ is typically called irregular because the time-zero gauge fields are only defined in exponentiated form via $\scrC(\oscrA)$. This irregularity is a consequence of the fact that we work with a Hilbert space representation of an algebra of gauge fields. If we wanted to work with regular representation, we would presumably need to invoke indefinite inner product spaces, see \cite{Loeffelholz-03-Temporal-Gauge}.
\end{remark}

\subsection{Action of the gauge group}\label{sec:gauge}
Let $\Ga:=\prod_\D G$ be the infinite product of $G$ over the set $\D$ equipped with the Tychonoff topology, i.e.~the group of all maps $g:\D\to G$ with the pointwise convergence topology.
It is a separable compact group.

We consider a tree $t$ with $n$ leaves and its associated s.d.p.~$\dn$.
We identify the group $G_t$ with $\oplus_{\cI(t)} G$ where $\cI(t)=\{(d_i,d_{i+1}),\ 1\leq i \leq n \}$ and define the gauge action at level $t$:
$$Z_t:\Ga\act G_t, Z_t(s)(g)(d_i,d_{i+1}) = s(d_i) g(d_i , d_{i+1}) s(d_{i+1})^{-1}.$$
As we work with the torus $\R/\Z$, we have the periodicity condition $s(d_{n+1}) = s(1) = s(0) = s(d_1)$. 
This action is p.m.p.~w.r.t.~the Haar measure $m_{G_t}$ because $G_t$ is unimodular and, thus, provides a unitary representation
$$W_t:\Ga\to \cU(L^2(G_t)), W_t(s)\xi(x):= \xi(Z_t( s^{-1} ) x), s\in \Ga, x\in G_t, \xi \in L^2(G_t).$$
We denote by $\ga_t:\Ga\act B(L^2(G_t))$ the adjoint action $\ga_t(s)(a):= W_t(s) a W_t(s)^*, s\in \Ga, a\in B(L^2(G_t))$.\\[0.1cm]
The following proposition follows from a routine computation but requires several identifications. We provide the proof for the convenience of the reader.
\begin{proposition}\label{prop:gaugeaction}
The C*-subalgebra $M_t=\scrC(G_t)\rtimes G_{t,d} \subset B(L^2(G_t))$ is stable under the action $\ga_t:\Ga\act B(L^2(G_t))$ and satisfies the formula:
$$\ga_t : \Ga\act M_t, \ga_t(s)\left(\sum_{g\in G_t} a_g \la_g\right)\mapsto \sum_{g\in G_t} a_g(Z_t(s)^{-1}\cdot ) \la_{sgs^{-1}},$$
for $s\in \Ga, a_g\in\scrC(G_t)$ and $g\mapsto a_g$ with finite support.
\end{proposition}

\begin{proof}
Consider a tree $t$ associated with s.d.p.~data $\dn,\D(t),\cI(t)$ as above.
Identify $M_t$ with the norm completion of $\scrC(\oplus_{\cI(t)} G)\rtimes_{\alg} \oplus_{\D(t)}G$ and $L^2(G_t)$ with $L^2(\oplus_{\cI(t)} G).$ 
Observe that under those identifications the left regular action of $G_t$ on itself becomes 
$$\oplus_{\D(t)}G\act \oplus_{\cI(t)} G, (g\cdot x) (d,d'):= g(d) x(d,d') \text{ for } (d,d')\in \cI(t).$$
Consider $s\in\Ga, a\in\scrC(\oplus_{\cI(t)} G), g\in \oplus_{\D(t)}G$ and $\xi\in L^2(G_t).$
Observe that 
\begin{align*}
\ga_t(s)a\xi(x) & = W_t(s) a W_t(s)^* \xi(x) = a W_t(s)^* \xi(Z_t(s^{-1} ) x)\\
& = a ( Z_t( s^{-1} ) x) W_t(s)^* \xi (Z_t(s^{-1} ) x) = a ( Z_t( s^{-1} ) x) \xi ( x)
\end{align*}
Therefore, $\ga_t(s)(a) = a( Z_t( s^{-1} )  \cdot ).$

Consider $(d,d')\in\cI(t).$
We have that 
\begin{align*}
Z_t(s) g^{-1} Z_t(s^{-1}) x (d,d') & = s(d) (g^{-1} Z_t(s^{-1}) x ) (d,d') s(d')^{-1}\\
& = s(d)g(d)^{-1} (Z_t(s^{-1}) x ) (d,d') s(d')^{-1}\\
& = s(d)g(d)^{-1}s(d)^{-1} x (d,d') s(d') s(d')^{-1}\\
& = (sgs^{-1})^{-1}(d) x (d,d'). 
\end{align*}

Now, we infer the result for the convolution operator $\la_g$:
\begin{align*}
\ga_t(s)(\la_g)\xi(x) & = \la_g W_t(s)^* \xi (Z_t(s^{-1}) x) = W_t(s)^* \xi (g^{-1}Z_t(s^{-1}) x)\\
& =  \xi (Z_t(s) g^{-1}Z_t(s^{-1}) x) = \xi ( (sgs^{-1})^{-1}  x )\\
& = \la_{sgs^{-1}} \xi (x).
\end{align*}
We obtain that $\ga_t(s)(\la_g) = \la_{sgs^{-1}}.$
\end{proof}

Consider the C*-algebraic limit $\scrM_0:=\varinjlim_{t\in\fT} M_t$ that is isomorphic to the norm completion of $\scrC(\oscrA)\rtimes_{\alg} \Gd$ inside $B(L^2(\oscrA,m_\oscrA))$ by Theorem \ref{theo:MzeroNC}.
The family of actions $Z_t:\Ga\act(G_t,m_{G_t})$ is compatible with the inverse system of probability measures spaces $(G_t,m_{G_t}, t\in\fT)$ obtained from the contravariant functor $\Xi:\cF\to \Set$ introduced in Section \ref{sec:Commutative}.
Moreover, all maps involved are p.m.p.~and continuous.
This implies that there exists a unique action 
$Z:\Ga\act (\oscrA,m_\oscrA)$ such that if we project $\oscrA$ onto $G_t$, we will obtain the action $Z_t$.
Moreover, if we equip $\oscrA$ with its topology inherited from the inverse system of topological space, $\Ga$ will act by p.m.p.~homeomorphisms.
The action $Z$ is defined by the formula:
$$Z(s)(x)(d,d') := s(d) x(d,d') s(d')^{-1},$$
for $s\in \Ga:=\prod_\D G, x\in \oscrA$ and $(d,d')$ is a s.d.i.
The corresponding unitary representation is given by
$$W:\Ga\to \cU(L^2(\oscrA,m_\oscrA)), W(s) \xi(x):= \xi(Z(s)^{-1} x),$$
$s\in  \Ga, x\in \oscrA, \xi \in L^2(\oscrA,m_\oscrA)$, and the associated adjoint action by
$$\ga:\Ga\act B(L^2(\oscrA,m_\oscrA)), \ga(s)(a):= W(s) a W(s)^*,$$
$s\in \Ga, a \in B(L^2(\oscrA,m_\oscrA))$.\\[0.1cm]
The following proposition results from a routine computation and the formula for the Jones' action of Theorem \ref{theo:MzeroNC}.

\begin{proposition}\label{prop:gaugeCstar}
The C*-subalgebra $\scrM_0\subset B(L^2(\oscrA,m_\oscrA))$ is stable under the action $\ga:\Ga\act B(L^2(\oscrA,m_\oscrA)).$
We continue to denote by $\ga$ the restricted action that we call the gauge group action and which satisfies the formula:
$$\ga : \Ga\act \scrM_0, \ga(s)\left(\sum_{g\in \oplus_\D G} a_g u_g\right)\mapsto \sum_{g\in \oplus_\D G} a_g(Z(s)^{-1}\cdot ) u_{sgs^{-1}},$$
for $s\in \Ga, a_g\in\scrC(\oscrA)$ and $g\mapsto a_g$ with finite support.

Let $\al:V\act \scrM_0$ be the Jones' action.
Then, we have the compatibility condition:
$$\al(v) \ga(s) \al(v^{-1}) = \ga( s(v^{-1}\cdot) ),\text{ for any } s\in \Ga, v\in V.$$
\end{proposition}

Now, we show that the gauge group action is basically the left regular action of the group $\Ga$.
To see this, we introduce the holonomy map $H$, see below, that is defined w.r.t.~the specific point $1$ on the torus in the sense of a boundary condition.\\[0.1cm]
Let us consider $x\in \oscrA$ and $d\in\D$.
There exists a tree $t$ with s.d.p.~$d_1=0<d_2<\cdots <d_n<d_{n+1}=1$ such that $d=d_j$ for a certain $1\leq j\leq n$, and we put 
$$H(x)(d) := x(d,d_{j+1} ) x(d_{j+1},d_{j+2})\cdots x(d_n,1).$$
This formula does not depend on the choice of the tree $t$ and defines a map 
$$H:\oscrA\to \prod_\D G.$$

\begin{proposition}\label{prop:holonomy}
The map $H$ is an homeomorphism from $\oscrA$ onto $\prod_\D G$ called the holonomy map.
Its inverse is defined by the formula:
$$H^{-1}(g)(d,d') = \begin{cases} g(d) g(d')^{-1} \text{ if } d'\neq 1 \\ g(d) \text{ otherwise } \end{cases} , \ \text{ for } g\in\prod_\D G \text{ and } (d,d') \text{ a } s.d.i.$$
Moreover, it is p.m.p.~when $\oscrA$ and $\prod_\D G$ are equipped with $m_\oscrA$ and the tensor product of measures $\ot_\D m_G$ respectively.

We transfer the gauge group action $Z:\Ga\act (\oscrA,m_\oscrA)$ to $\prod_\D G$ via the conjugation with the map $H$ giving us $Z':\Ga\act (\prod_\D G, \ot_\D m_G).$
Identifying the space $\prod_\D G$ with the gauge group $\Ga$, we have that $$Z'(s)g(d) = s(d) g(d) s(0)^{-1},\ \text{ for any } s\in \Ga, g\in \prod_\D G, d\in \D.$$
In particular, the associated unitary representation is strongly continuous.
\end{proposition}

This last proposition is easy to check and is left to the reader.

\subsection{Local algebras of fields}\label{sec:fieldalg}
Consider a connected open subset $O$ of the torus.
We intend to define the algebra of (time-zero) fields in this region of space.
Let $t\in\fT$ be a tree and let $\cI(t)$ be the associated collection of {\it open} s.d.i.
Recall that 
$$M_t \simeq \ot_{I\in \cI(t)} M \simeq \ot_{I\in \cI(t)} (\scrC(G)\rtimes G_d ) \simeq \scrC(G_t)\rtimes G_{t,d}.$$
We write $\cI(t,O)$ for the subset of s.d.i.~$I\in \cI(t)$ with $I\subset O$ and define the C*-subalgebra $M_t(O) = \ot_{I\in \cI(t,O)} M \ot \ot_{J\notin \cI(t,O)} \C$.
Moreover, we define the coordinate projection 
$$p_{t,O}: G_t \to G_{t,O},\ (g_I)_{I\in \cI(t)} \mapsto (g_I)_{I\in \cI(t,O)}$$
and the group embedding 
$$\iota_{t,O}: G_{t,O} \to G_{t},\ (g_I)_{I\in \cI(t,O)} \mapsto (\tilde g_I)_{I\in \cI(t)} \text{ such that } \tilde g_I = \begin{cases} g_I \text{ if } I \in \cI(t,O)\\ e \text{ otherwise } \end{cases}.$$
The inclusion $M_t(O)\subset M_t$ corresponds to the embedding 
\begin{align*}
j_{t,O} & :\scrC(G_{t,O})\rtimes G_{t,O,d} \to \scrC(G_{t})\rtimes G_{t,d},\\
& \sum_{g\in G_{t,O}} a_g \la_g\mapsto \sum_{g\in G_{t,O}} (a_g\circ p_{t,O} )\la_{\iota_{t,O}(g)},
\end{align*}
with $a_g\in\scrC(G_{t,O})$ and $g\mapsto a_g$ finitely supported.
Recall that $\oscrA$ is the compact space of maps $x$ from the set of s.d.i.~to $G$ satisfying $x(I) = x(I_1)x(I_2)$ whenever $I_1,I_2$ are the first and second half of the s.d.i $I$ equipped with the product subspace topology.
Therefore, we define the space 
$$\oscrA(O):= \oscrA \cap \prod_{I s.d.i.~: I\subset O} G$$
and its associated coordinate projection $p_O:\oscrA\to \oscrA(O)$. This space can be interpreted as inverse limit of the following system:
consider the collection of spaces $(G_{t,O},\ t\in\fT)$ and the projections $p_t^{ft}(O): G_{ft,O}\to G_{t,O}, x\mapsto p_{t,O}\circ p_t^{ft}\circ\iota_{ft,O}(x)$, where $p_t^{ft}:G_{ft}\to G_t$ is the projection constructed via the group multiplication, see Section \ref{sec:Commutative}.\\[0.1cm]
We denote by $\D(O)$ the set of $d\in \D$ such that $d\in O$ or $d$ is the left boundary point of the connected open set $O.$
We define the inclusion 
$$\iota_O:\oplus_{\D(O)} G\to \oplus_\D G, \iota_O(g)(d) = \begin{cases} g(d) \text{ if } d \in \D(O)\\ e \text{ otherwise } \end{cases}.$$
We know that the gauge group $\Ga:=\prod_\D G$ acts on $\scrM_0$ via the action $\ga$ defined in Section \ref{sec:gauge}.
We define the localized gauge group $\Ga(O):=\prod_{\D\cap \overline O} G$, where the product runs over all dyadic rationals that are in the \textit{closure} of $O$.
We denote by $q_O:\Ga\to \Ga(O)$ the associated coordinate projection, and we define the action 
$$Z_O:\Ga(O)\act \oscrA(O), Z_O(s)(x)(d,d'):= s(d) x(d,d') s(d')^{-1},$$ 
for $s\in\Ga(O), x\in\oscrA(O)$, $(d,d')$ a s.d.i.~contained inside $O$.
Here, we consider the intersection of $\D$ with the closure of $O$ for the definition of $\Ga(O)$ to include both boundary points of $O$.
This is necessary because the formula for $Z_O$ above requires to use the values $s(d)$ and $s(d')^{-1}$ also when $d$ and $d'$ are the left and right boundary of $O$ respectively.\\[0.25cm]

We deduce the following properties.

\begin{proposition}\label{prop:observables}
Consider some connected open subsets $O,O_1,O_2$ of the torus and a tree $t\in \fT$.
\begin{enumerate}
\item If $O_1\subset O_2$, then $M_t(O_1)\subset M_t(O_2).$
\item If $O_1\cap O_2 = \emptyset$, then $M_t(O_1)$ and $M_t(O_2)$ mutually commute.
\item If $t\leq s$, then $M_t(O)$ is a unital C*-subalgebra of $M_s(O)$ when $M_t$ is identified with a subalgebra of $M_s$ via the functor $\Psi.$
\item The norm closure of the union of the C*-algebras $(M_s(O),\ s\in \fT)$ is a unital C*-subalgebra $\scrM_0(O)\subset \scrM_0.$
The algebraic crossed-product $\scrC(\oscrA(O))\rtimes_{\alg} \oplus_{\D(O)} G$ embeds inside $\scrM_0(O)$ and is a dense *-subalgebra.
\item The inclusion $\scrM_0(O)\subset \scrM_0$ restricts to the following embedding
\begin{align*}
j_O & : \scrC(\oscrA(O))\rtimes_{\alg} \oplus_{\D(O)} G\to \scrC(\oscrA)\rtimes_{\alg} \oplus_\D G\\
& \sum_{g\in \oplus_{\D(O)} G} a_g u_g\mapsto \sum_{g\in \oplus_\D G} (a_g\circ p_O) u_{\iota_O(g)},
\end{align*}
where $g\mapsto a_g\in\scrC(\oscrA(O))$ is finitely supported.
\item If $O_1\subset O_2$, then $\scrM_{0}(O_1)\subset\scrM_{0}(O_2)$.
\item If $O_1\cap O_2 = \emptyset$, then $\scrM_{0}(O_1)$ and $\scrM_{0}(O_2)$ mutually commute.
\item If $v$ is an element of Thompson's group $T$, then $\al(v)(\scrM_0(O)) = \scrM_0(vO)$ where $\al$ is the Jones' action.
\item Let $\ga_O:\Ga(O) \act \scrM_0(O)$ be the action satisfying
$$\ga_O(s)\left(\sum_{g\in \oplus_{\D(O)} G} a_g u_g \right) :=  \sum_{g\in \oplus_{\D( O)} G} a_g(Z_O(s ^{-1} )\cdot ) u_{sgs^{-1}}.$$
This corresponds to the restriction of the gauge group action on $\scrM_O$ that we call the localized gauge group action.
Moreover, the family of localized gauge group actions is equivariant w.r.t.~the Jones' action, i.e.~$\Ad(\alpha(v)) \circ \ga_O = \ga_{vO}$ for $v\in T$.
\end{enumerate}
\end{proposition}
\begin{proof}
(1) and (2) are obvious from the definition of $M_t(O)$ in term of tensor products over the set $\cI(t).$

(3) comes from the fact that the directed system $(M_t,,\iota_t^s , t\leq s \in\fT)$ has unital *-morphism embeddings for maps.

The first part of (4) is obvious since by (3) we have that $(M_t(O), t\in \fT)$ is nested and $M_t(O)\subset M_t$ by definition.

The second parts of (4) and (5) come from the identification of $M_t\simeq\scrC(G_t)\rtimes_{\dis} G_t$ and its embedding inside $\scrM_0$ that is identified with the completion of $\scrC(\oscrA)\rtimes_{\alg} \oplus_\D G$.

(6) and (7) are easy consequences of (5).

(8) and (9) are obvious consequences of the formulae given in Theorem \ref{theo:MzeroNC} and Proposition \ref{prop:gaugeCstar}.
Note that we restrict to Thompson's group $T$ in order to have that $vO$ is still connected.
\end{proof}

\begin{remark}In subsequent sections, we introduce a state $\varpi$ on $\scrM_0$ and perform the GNS construction giving a von Neumann algebra equipped with a cyclic vector $(\scrM,\Omega_{\varpi})$.
Then, we can define the localized von Neumann algebra $\scrM(O)$ as the weak completion of $\scrM_0(O)$ inside $\scrM$ resulting in local, isotonous net of von Neumann algebras.
If the action of Thompson's group $T$ extends to $\scrM$, we obtain a statement similar to the previous proposition, where the space $O$ is send to $vO$ for $v\in T$. Moreover, we have a similar statement for the action of the localized gauge groups.
At this point, we will not develop further the von Neumann algebra version of the net of field algebras because statements easily follow from the C*-algebraic case.\end{remark}

\section{Gauge theory: construction of a von Neumann algebra with a state}\label{sec:VNA}

\subsection{General construction}
We now want to find a suitable state on $\scrM_0$ and consider its associated von Neumann algebra obtained via the GNS construction.
We start by considering a state on the C*-algebra $M:=\scrC(G)\rtimes_{\dis} G$.

\subsubsection{Construction of a state on $M$.}
Let $\wh G$ be the unitary dual of $G$ that is the set of equivalence classes of irreducible unitary representations of $G$.
Recall that since $G$ is compact and separable, this set is countable and all the representations are finite dimensional.
We call them $(\pi,H_\pi)$ with $d_\pi$ being the dimension of $H_\pi$. We write $\chi_\pi(g):= \tr_{H_\pi}(\pi(g)), g\in G$ for the character associated to $\pi$ where $\tr_{H_{\pi}}$ is the non-normalized trace.

\begin{remark}
\label{rem:heatkernel}
As stated in the introduction, the choice of the state is driven by the heat kernel of a compact Lie group $G$. The bi-invariant Laplacian $\Delta_{G}$ of $G$ determines the strong-coupling limit of the Kogut-Susskind Hamiltonian of lattice gauge theory. In fact, it corresponds to the full Hamiltonian in the 1+1-dimensional setting. Now, the heat kernel is the fundamental solution $\rho_{\beta}$ (at the identity of $G$) to the heat equation,
$$\tfrac{\textup{d}}{\textup{d}\beta}\rho_{\beta} = \tfrac{1}{2}\Delta_{G}\rho_{\beta},$$
and determines a family of associated Gibbs or KMS states on $M$ via the trace-class operators $\lambda(\rho_{\beta})$. As is well-known, $\rho_{\beta}$ has an expansion into characters of $G$,
$$\rho_{\beta} = \sum_{\pi\in\wh G}d_{\pi}e^{-\frac{\beta}{2}c_{\pi}}\chi_{\pi},$$
where $c_{\pi}$ is the negative of the eigenvalue of $\Delta_{G}$ w.r.t.~$(\pi, H_{\pi})$. Thus, $\rho_{\beta}$ determines a probability measure on $\wh G$ via $h_{0,\beta}(\pi) = \rho_{\beta}(e)^{-1}d^{2}_{\pi}e^{-\frac{\beta}{2}c_{\pi}}$.
\end{remark}

Following the preceding remark, we choose a strictly positive probability measure $m\in \Prob(\wh G)$ on $\wh G$ and set $h_0:\wh G\to [0,1]$ such that $h_0(\pi):=m(\{ \pi \}).$
Note that such a $m$ exists if and only if $G$ is separable.
Set $$h:=\sum_{\pi \in \wh G} d_\pi^{-1} h_0(\pi) \chi_\pi$$ that is continuous and integrable on $(G,m_G).$
Consider the functional $$\omega(b):=\Tr(b\la(h)),b\in B(L^2(G)).$$

It is important to observe that each character $\chi$ defines a central projection $\la(\chi)$ of the group von Neumann algebra $LG$ such that the family $\{\la(\chi)\}_{\chi}$ is a partition of the identity.
In particular, we have that $\la(h)$ is in the center of $LG.$

{\bf Notation:} We have defined a probability measure $m$ on the dual group $\wh G$ using a function $h$. Recall that $m_G$ is the Haar measure on $G$ with total mass one which should not be confused with $m$. We write $L^p(G)$ for the space $L^p(G,m_G)$ w.r.t.~the Haar measure $m_G$.

\begin{lemma}\label{lem:omega}
The functional $\omega$ is a faithful normal state on $B(L^2(G))$ and satisfies
$$\omega \left( \sum_{g\in G} b_g \la_g \right) = \sum_{g\in G} \left( \int_G b_g\ dm_G \right) h(g^{-1}),$$
where $b_g\in L^\infty(G)$ and is equal to zero for all but finitely many $g\in G.$ 
\end{lemma}

\begin{proof}
Consider the function $\pi_{n,m}(g) := \langle \pi(g) \delta_m^\pi , \delta_n^\pi\rangle, g\in G$ where $\pi\in\wh G$ and $(\delta_k^\pi,\ 1\leq k\leq d_\pi)$ is an orthonormal basis of $H_\pi.$
Peter-Weyl theorem states that the set $\{ \sqrt{ d_\pi}\ \pi_{n,m} ,\ \pi\in \wh G, 1\leq n,m\leq d_\pi\}$ is an orthonormal basis of $L^2(G)$ and observe that $$\chi_\pi * \pi'_{n,m} = \begin{cases} \pi_{n,n} \text{ if } \pi=\pi' \text{ and } n=m\\ 0 \text{ otherwise } \end{cases} .$$
Consider some scalars $a_\pi, \pi\in\wh G$ satisfying that $\sum_{\pi\in\wh G} | a_\pi | d_\pi <\infty$.
This implies that the series of functions $\sum_{\pi\in\wh G} a_\pi \chi_\pi$ converges uniformly on $G$ to a continuous function $a=\sum_{\pi\in\wh G} a_\pi \chi_\pi$ that we identify with a convolution operator of $LG.$
The formula of above implies that $\Tr(\la(\chi_\pi)) = d_\pi$ for any $\pi\in\wh G$ and thus $\Tr(\la(a)) = a(e) = \sum_{\pi\in\wh G} a_\pi d_\pi$.
Moreover, if $b=\pi_{n,m}$ and is identified with the associated operator of pointwise multiplication $b\in L^\infty(G)$, then $\Tr(b \la(a)) =0$ unless $\pi$ is the trivial representation for any $a$ as above.
This implies that for any $b\in L^\infty(G)$ we have that 
$$\Tr(b \la(a)) = \langle b, 1\rangle \Tr(\la(a)) = \left( \int_G b \ dm_G\right) \Tr(\la(a)).$$

Observe that $h(e) = \sum_{\pi\in\wh G} d_\pi^{-1} h_0(\pi) d_\pi = \sum_{\pi\in\wh G} h_0(\pi) = 1$. 
Moreover, $\la_g\la(h) = \la(h^g)$ where $h^g(x) = h(g^{-1} x)$ and thus $\Tr(\la_g \la(h)) = h^g(e) = h(g^{-1})$.
This implies the formula of the lemma.

The set of characters forms a set of orthogonal projections inside $LG$ with sum equal to the identity.
This implies that $\la(h)$ is a positive operator with strictly positive spectrum (since $h_0(\pi)\neq 0$ for any $\pi\in\wh G$) implying that $\omega:=\Tr(\cdot \la(h))$ is faithful.
\end{proof}

We now define our setting and our family of states.
\subsubsection{A coherent family of states}
We consider a family of measures given by any map $m:\D\to\Prob(\wh G), d\mapsto m_d$ such that $m_d$ is strictly positive for any $d\in\D.$
Define $h_d:=\sum_{\pi\in\wh G} \tfrac{h_{d,0}(\pi)}{d_\pi} \chi_\pi$ for $d\in\D$ where $h_{d,0}$ is the unique positive map satisfying $m_d(A) = \sum_{\pi\in A} h_{d,0}(\pi)$ for any $A\subset\wh G.$
The element $h_d$ belongs to the center of $LG$ and defines a normal faithful state 
$$\omega_d:b\in B(L^2(G))\mapsto \Tr(b h_d)$$
that we restrict to $M$.

Consider a tree $t\in\fT$ with $n$ leaves and associated s.d.p.~$d_1=0<d_2<\cdots<d_n<d_{n+1}=1.$
We have $M_t=\scrC(G_t) \rtimes_{\dis} G_t $ living inside $B(L^2(G_t))$.
We identify $B(L^2(G_t))$ with the von Neumann $n$th tensor power of $B(L^2(G))$ denoted $B(L^2(G))^{\ovt n}.$
Define 
\begin{itemize}
\item the element $h_t:=h_{d_1}\ot\cdots\ot h_{d_n}$ that is in the center of $LG_t$;
\item the associated probability measure $m_t:=m_{d_1}\ot\cdots\ot m_{d_n}$ on $G_t$ and 
\item the state $\omega_t:M_t\to \C,\ \omega_t(b) = \Tr(b h_t).$
\end{itemize}

\begin{proposition}\label{prop:NCstate}
Fix a tree $t$ with associated s.d.p.~$d_1=0<d_2<\cdots<d_n<d_{n+1}=1$ and identify $M_t$ with the $n$th tensor power of $M$.
The state $\omega_t$ of $M_t$ is equal to the tensor product of states $\omega_{d_1}\ot\cdots\ot\omega_{d_n}$ that is the restriction of a normal faithful state.
It satisfies the equality
\begin{equation}\label{equa:stateNC}\omega_t\left( \sum_{g\in G_t} a_g \la_g \right) = \sum_{g\in G_t} \left( \int_{G_t} a_g \ d m_t \right) \prod_{j=1}^n h_{d_j}(g_j^{-1}),\end{equation}
where $g\mapsto a_g\in \scrC(G_t)$ has finite support. 
Moreover, the embedding $\Psi(f):M_t\to M_{ft}$ is state-preserving (i.e.~$\omega_{ft} \circ \Psi(f) = \omega_{t}$) for any forest $f$ with $n$ roots.
\end{proposition}

\begin{proof}
By identifying $h_t$ with $h_{d_1}\ot\cdots h_{d_n}$ we obtain that $\Tr(\cdot h_t) = \ot_{j=1}^n \Tr(\cdot h_{d_j})$ and thus $\omega_t$ is identified with $\omega_{d_1}\ot\omega_{d_2}\ot\cdots\ot\omega_{d_n}$.
Since any $\omega_d,d\in\D$ is a normal faithful state so does $\omega_t$.

Consider an elementary tensor of functions $a=a_1\ot\cdots a_n\in\scrC(G_t)$ and $g=(g_1,\cdots,g_n)\in G_t.$
By applying $\omega_t$ to $a$ and $\la_g$ we obtain:
$$\omega_t(a) = \prod_{j=1}^n \omega_{d_j}(a_j) = \prod_{j=1}^n \int_G a_j\ dm_{d_j} = \int_{G_t} a\ dm_t$$
and 
$$\omega_t(\la_g) = \prod_{j=1}^n \omega_{d_j}(\la_{g_j}) = \prod_{j=1}^n \Tr( \la_{g_j} h_{d_j}) = \prod_{j=1}^n h_{d_j}(g_j^{-1}).$$
Observe that since $h_t$ is in the center of $LG_t$ (and is trace class with trace equal to one) we have that $\Tr( a \la_g h_t) = \Tr(a h_t) \Tr(\la_g h_t)$. 
This proves formula \eqref{equa:stateNC}

Consider $f:= f_{j,n}$ the forest with $n$ roots, $n+1$ leaves whose $j$th tree has two leaves and write $e_j:= \tfrac{d_j + d_{j+1}}{2}.$ 
Since any forest is a composition of such elementary one it is sufficient to show $\omega_{ft} \circ \Psi(f)=\omega_t.$
Moreover, by density it is sufficient to check the equality $\omega_{ft} \circ \Psi(f)(x)= \omega_t(x)$ for an elementary tensor $x=x_1\ot\cdots\ot x_n$ .
Observe that 
\begin{align*}
\omega_{ft} \circ \Psi(f)(x) & = \omega_{ft} (x_1\ot\cdots\ot x_{j-1} \ot R(x_j) \ot x_{j+1}\cdots\ot x_n)\\
& = \omega_{d_1}(x_1) \cdots \omega_{d_{j-1}} (x_{j-1}) [\omega_{d_j} \ot \omega_{e_j}]( R(x_j)) 
\ot \omega_{d_{j+1}}(x_{j+1}) \cdots \omega_{d_n}(x_n)\\
& =  [\omega_{d_j} \ot \omega_{e_j}]( u (x_j\ot \id) u^*) \prod_{i\neq j} \omega_{d_i}(x_i)\\
& =  \Tr ( u (x_j\ot \id) u^* (h_{d_j} \ot h_{e_j}) ) \prod_{i\neq j} \omega_{d_i}(x_i)\\
\end{align*}
Observe that the operator $u$ belongs to the group von Neumann algebra $R(G\times G)$ acting to the right which commutes with $L(G\times G).$
Since $h_d$ is in $LG$ for any $d\in\D$ we obtain that $u$ commutes with $h_{d_j} \ot h_{e_j}$ and thus
$$\Tr ( u (x_j\ot \id) u^* (h_{d_j} \ot h_{e_j}) ) = \Tr ( (x_j\ot \id)  (h_{d_j} \ot h_{e_j}) ) = \omega_{d_j}(x_j)$$
implying that $\omega_{ft}\circ \Psi(f) = \omega_t.$
\end{proof}

\subsubsection{The limit state and the GNS completion}
By the last proposition we can define  $\varpi$ to be the unique state on $\scrM_0$ satisfying that $\varpi(b)=\omega_t(b)$ for any $b\in M_t$ and any $t\in\fT.$
Let $\scrM$ be the GNS completion of $\scrM$ w.r.t.~$\varpi$ and continue to write $\varpi$ the normal extension of this state on $\scrM$, see Section \ref{sec:Actions} for more details.

Recall that the group ring $\C[\Gd]$ embeds as a dense *-subalgebra inside $\scrN_0$ and that we have an action $\Gd\act\oscrA$ given by the formula
$gx(I):= g(d) x(I)$ where $I$ is a s.d.i.~starting at $d$.
Let $\scrC(\oscrA)\rtimes_{\alg} \Gd$ be the *-algebra 
$$\{\sum_{g\in \Gd} b_g u_g : \ b_g\in\scrC(\oscrA), |\supp(g\mapsto b_g)|<\infty\}$$
such that $\Ad(u_g) b(x) = b(g^{-1}x)$ for any $g\in\Gd, b\in\scrC(\oscrA), x\in \oscrA.$
The next theorem gives a description of $\scrM$ as a GNS-completion of a crossed-product but when the state is nontrivial on the group acting part.

\begin{theo}\label{omegaNC}
There is an injective *-morphism 
$$J:\scrC(\oscrA)\rtimes_{\alg} \Gd\to \scrM$$ 
with weakly dense image such that 
\begin{equation}\label{equa:varpi}\varpi\circ J(\sum_{g\in \Gd} b_g u_g) = \sum_{g\in \Gd} \left(\int_\oscrA b_g(t) d m_\oscrA(t) \prod_{d\in\supp(g)} h_d(g(d)^{-1}) \right).\end{equation}

The gauge group action $\ga:\Ga:=\prod_\D G\act \scrM_0$ is state-preserving and thus extends to an action on the completion $\scrM$.

If all measures $m_d$ are equal to a single one, the Jones' action $\al_\Psi:V\act \scrM_0$ will be state-preserving and extends to an action on $\scrM$.
\end{theo}

\begin{proof}
Theorem \ref{theo:MzeroNC} implies that the range of $J$ is weakly dense.
Since the GNS representation associated to $\varpi$ is faithful (see Section \ref{sec:state-prelim}), then the map $J$ is necessarily injective.
The formula of the state is obvious at any tree level $t\in\fT$ giving us the formula for the algebraic crossed-product $\scrC(\oscrA)\rtimes_{\alg} \oplus_\D G$.

Consider $s\in\Ga$ and an element of the algebraic crossed-product $\sum_{g\in \Gd} b_g u_g$.
Note that $\ga(s) \left( \sum_{g\in \Gd} b_g u_g \right) = \sum_{g\in \Gd} b_g(Z(s^{-1}) \cdot )  u_{sgs^{-1}}$ that is send by the state to
$$\sum_{g\in \Gd} \left(\int_\oscrA b_g(Z(s^{-1}) t) d m_\oscrA(t) \prod_{d\in\supp(g)} h_d(s(d)g(d)^{-1} s(d)^{-1}) \right).$$
The measure $m_\oscrA$ is invariant under the transformation $Z(s)$ by Proposition \ref{prop:holonomy}.
Moreover, the function $h_d$ is in the span of the character of $G$ implying that 
$$h_d(s(d)g(d)^{-1} s(d)^{-1}) = h_d(g(d)^{-1}).$$
We obtain that $\ga(s)$ preserves the state on a weakly dense subalgebra and thus is a state-preserving automorphism of $(\scrM,\varpi).$

Recall that $V\act (\oscrA,m_\oscrA)$ is p.m.p.~which implies, using \eqref{equa:varpi}, that $\varpi$ restricted to $\scrC(\oscrA)$ is invariant under the action $\al:V\act \scrC(\oscrA)$.
The *-algebra generated by $\Gd$ is closed under the action of $V$ and the state $\varpi$ restricted on it splits as an infinite tensor product of states indexed by $\D$.
If $m_d$ does not depend on $d\in\D$, then $\varpi|_{ \C[\Gd]}$ is equal to an infinite tensor product of the same state. 
Since the Jones' action acts by shifting indices in $\D$, it will unchange the restricted state $\varpi|_{\C[\Gd]}$.
Finally observe that $\varpi$ satisfies that $\varpi(bc) = \varpi(b)\varpi(c)$ if $b\in \scrC(\oscrA)$ and $c\in \C[\Gd]$ and that any element of $\scrC(\oscrA)\rtimes_{\alg} \Gd$ can be written as a sum of $bc$ with $b\in\scrC(\oscrA)$ and $c\in \C[\Gd].$ 
This implies that the action $\al:V\act \scrC(\oscrA)\rtimes_{\alg} \Gd$ is state-preserving when $m_d$ does not depend on $d\in\D$.
In that case, the Jones' action defines an action by automorphism of $V$ on the von Neumann algebra $\scrM$.
\end{proof}

Note that in general the action $\al:V\act \scrM_0$ does not extend to an action on $\scrM$.
The problem being that the state $\varpi$ is not always invariant under this action.
In Section \ref{sec:Zleaves} we will present a model where the rotation subgroup of Thompson's group $T$ still acts on $\scrM$ but in a non-state preserving way.

{\begin{remark}\label{rem:subalg}
The pair $(\scrM,\varpi)$ does not depend on the choice, $M$, of a weakly dense subalgebra of $B(L^2(G))$.
More precisely, note that the map $R$ can be extended to an algebra morphism from $B(L^2(G))$ to $B(L^2(G))\ot B(L^2(G))$ and the state $\omega$ extends to a normal state on $B(L^2(G))$ and so does $\omega_t$ on $B(L^2(G^n))$ where $n$ is the number of leaves of $t$.
If we consider any weakly dense *-subalgebra $\tilde M\subset B(L^2(G))$ satisfying that $R(\tilde M)\subset \tilde M\ot \tilde M$, we can define a directed system $((\tilde M_t,\omega_t\vert_{\tilde M_t}),\ t\in\fT)$ of *-algebras equipped with states and state-preserving isometric maps exactly as we did with $\{M,R,\omega\}$.
We will obtain a limit *-algebra with a state $(\tilde\scrM_0,\tilde\varpi)$ that we complete via the GNS construction into a von Neumann algebra $(\tilde\scrM,\tilde\varpi)$.
A similar map as the one given in Theorem \ref{omegaNC} will provide a state-preserving isomorphism from $(\tilde\scrM,\tilde\varpi)$ onto $(\scrM,\varpi)$.
\end{remark}}

Next, we intend to analyze in more depth the couple $(\scrM,\varpi)$ and the action of Thompson's group $V$ on it.
In order to achieve this, we will further assume that the group $G$ is abelian.

\subsection{The abelian group case together with a single measure}
\label{sec:Gabelian}
In this section, we further assume that the compact separable group $G$ is abelian.
The unitary dual $\wh G$ is then given by the Pontryagin dual of $G$ that is a countable discrete abelian group. We also assume in this section that the map $m:\D\to \Prob(\wh G)$ is constant.\\[0.1cm]
We start by defining the dual version of our analysis using the Fourier transform.
Recall that the Fourier transform is the unitary transformation 
$$U_F: L^2(G)\to \ell^2(\wh G), \chi\mapsto \delta_\chi,$$
where $\chi$ is a character of $G$ and $\delta_\chi$ is the corresponding delta function of $\ell^2(\wh G).$
Denote by $\la:\wh G\act \ell^2(\wh G)$ the left regular representation and identify $ \ell^\infty(\wh G)$ with the operator of pointwise multiplication action on $\ell^2(\wh G).$
Observe that $\Ad(U_F)(\scrC(G)) =C_{red}^*(\wh G)$ is the reduced group C*-algebra.
Since $\wh G$ is abelian the full and reduced group C*-algebras coincide and thus we drop the subscript 'red'.
We have that $\Ad(U_F)(\overline{\C G})$ is the C*-algebra generated by the maps $\chi\mapsto \chi(g)$ for fixed $g\in G$, i.e.~the characters of $\wh G$.
It is a C*-subalgebra of the pointwise multiplication operators that we denote by $\Ch(\wh G).$
Identify $B(L^2(G))$ with $B(\ell^2(\wh G))$ via the Fourier transform.
Under the identification we obtain that $N=C^*(\wh G), Q=\Ch(\wh G)$ and $M$ is the crossed-product C*-algebra $\Ch(\wh G)\rtimes \wh G$ living inside $B(\ell^2(\wh G))$ generated by $N$ and $Q$ and where $\wh G$ acts on $\wh G$ via the left regular representation that lifts to an action on $\Ch(\wh G).$
Doing the same identification at a tree level we obtain $N_t=C^*(\wh G_t), Q_t=\Ch(\wh G_t)$ and $M_t= \Ch(\wh G_t)\rtimes \wh G_t$.
By our previous work we have that $N_t,Q_t,M_t$ are isomorphic to the $n$th (minimal) tensor product of $N,Q,M$ if $t$ has $n$ leaves respectively.
Note that the Fourier transform swaps the convolution and pointwise multiplication parts.

Our inclusion map $R:B(L^2(G))\to B(L^2(G))\ot B(L^2( G)), b\mapsto \Ad(u)(b\ot \id)$ is then replaced by 
$$\wh R:= \Ad(U_F\ot U_F) \circ R\circ \Ad(U_F): B(\ell^2(\wh G)) \to B(\ell^2(\wh G))\ovt B(\ell^2(\wh G))$$
and the state $\omega(b):=\Tr(b \la(h)), b\in B(L^2(G))$  by $\wh \omega:= \omega\circ\Ad(U_F).$
Observe that since $G$ is abelian the dimension of a unitary representation $\pi$ is $d_\pi=1$ and we only have characters $\chi$ instead of matrix coefficients $\pi_{n,m}.$
Therefore, our function is $h=\sum_{\chi\in\wh G} h_0(\chi) \chi$ with associated measure $m\in \Prob(\wh G)$ satisfying $m(A)=\sum_{\chi\in A} h_0(\chi)$ for $A\subset\wh G.$

\begin{lemma}\label{lem:whR}
We have the equalities:
$$
\wh R( \sum_{\chi\in\wh G} b_\chi \la_\chi)  = \sum_{\chi\in\wh G} b_\chi \la_\chi \ot \la_\chi,$$
where $b_\chi\in \Ch(\wh G)$ and is equal to zero except for finitely many $\chi\in\wh G.$

In particular, $\wh R(Q)\subset  Q\ot  Q$ and $\wh R( N)\subset  N \ot  N.$

The state $\wh \omega$ satisfies that $$\wh\omega( \sum_{\chi\in\wh G} b_\chi \la_\chi) = \sum_{g\in \wh G} b_e(g) h_0(g)= \int_{\wh G} b_e(g) d m(g),$$
where $e$ is the neutral element of $\wh G$.
\end{lemma}

\begin{proof}
Those are direct consequences of Lemma \ref{lem:omega}.
\end{proof}
We obtain two systems of C*-algebras $(Q_t,t\in \fT)$ and $(N_t,t\in\fT)$ with C*-completions of inductive limits written $\scrQ_0$ and $\scrN_0$ respectively. 
In order to keep notations simples we write $\omega_t$ the state on $M_t$ without hat and limit state $\varpi$ also without hat.
Before describing $\scrQ_0$ and $\scrN_0$ we define some limits of groups.

\begin{definition}
Consider the group $\wh G_{fr}$ of all maps $g:\D\to \wh G$ such that there exists a s.d.p.~$\dn$ satisfying that $g$ is constant on each half-open interval $[d_j,d_{j+1})\cap \D$ for $1\leq j\leq n.$

Denote by ${\wh G}^\D$ the infinite product of groups $\prod_{d\in\D} \wh G$ endowed with the product topology where each copy $\wh G$ is equipped with the discrete topology.
\end{definition}

Consider the following action $\wh G_{fr}\act{\wh G}^\D$ given by the formula
\begin{equation}\label{equa:GfrGD} (g\cdot x) (d) = g(d)x(d),\ \forall  g\in \wh G_{fr}, x\in {\wh G}^\D, d\in \D.\end{equation}

Let $\ot_\D \Ch(\wh G)$ be the infinite tensor product over the set $\D$ of the C*-algebra $\Ch(\wh G)$. 
Note that it is isomorphic to the unique C*-completion of the inductive limit of C*-algebras $\Ch({\wh G}^E)$ where $E$ runs over every finite subsets of $\D.$ 
The action \eqref{equa:GfrGD} provides an action of $\wh G_{fr}$ over the tensor product $\ot_\D \Ch( \wh G) $ when we consider it as an algebra of functions overs ${\wh G}^\D.$

Denote by $u_\chi, \chi\in \wh G_{fr}$ the classical embedding of $\wh G_{fr}$ inside $C^*(\wh G_{fr})$.
Put $m^\D$ the infinite tensor product of the measure $m$ on $\wh G.$
We are now able to describe a weakly dense *-subalgebra of $\scrM$.

\begin{proposition}\label{prop:AlgCP}
Consider the algebraic crossed-product $\ot_\D \Ch(\wh G)\rtimes_{\alg} \wh G_{fr}$ for the action described above.
This *-algebra embeds as a weakly dense *-subalgebra of $\scrM$.
Moreover, $$\varpi(\sum_{\chi\in \wh G_{fr}} b_\chi u_\chi ) = \int_{{\wh G}^\D} b_e (t) d m^\D(t),$$
where $b_\chi\in \ot_\D \Ch(\wh G)$ is equal to zero for all but finitely many $\chi\in\wh G_{fr}.$

The Jones' action $V\act \scrM_0$ extends to a state-preserving action $V\act (\scrM,\varpi).$
Consider the classical action $V\act \D$ providing the action $V\act {\wh G}^\D$ which restricts to $V\act \wh G_{fr}.$
We obtain the action 
$$V\act \ot_\D \Ch(\wh G)\rtimes_{\alg} \wh G_{fr},\  v\cdot (\sum_{\chi\in \wh G_{fr}} b_\chi u_\chi ) := \sum_{\chi\in \wh G_{fr}} b_{\chi}(v^{-1} \ \cdot ) u_{v\chi}.$$
that is the restriction of the Jones' action.
\end{proposition}

\begin{proof}
Let $\Psi_Q,\Psi_N:\cF\to \Calg$ be the tensor functors induced by $\wh R$ such that $\Psi_Q(1)=Q$ and $\Psi_N(1)=N.$
Lemma \ref{lem:whR} implies that those functors are defined by the map satisfying $R_Q(\la_\chi)=\la_\chi\ot \la_\chi,\chi\in\wh G$ and $R_N(a)=a\ot \id, a\in N.$
Proposition \ref{prop:TensorProd} implies that $\scrN_0=\ot_\D N = \ot_\D \Ch(\wh G)$ the infinite (minimal) tensor product of $N$.
 
Consider the tensor functor $\Upsilon:\cF\to \Gr$ defined as $\Upsilon(1)= \wh G$ and $\Upsilon(Y)(\chi) = (\chi,\chi)\in \wh G\times \wh G$ where $\Gr$ is the category of countable discrete groups with tensor structure given by direct products.
This defines a directed system of groups $(\wh G_t,\ t\in\fT)$. 
Fix a tree $t\in\fT$ and consider its s.d.p.~$\dn$.
Embed $\wh G_t$ into $\wh G_{fr}$ via the map $j_t(g)(d) = g_i$ if $d_i\leq d < d_{i+1}.$
Note that the family of maps $(j_t, t\in\fT)$ is compatible w.r.t.~the directed system of groups and defines an embedding of $\varinjlim_{t\in\fT}\wh G_t$ into $\wh G_{fr}.$
By definition of $\wh G_{fr}$ this map is clearly surjective.

Consider the functor $C^*_{red}:\Gr\to \Calg$ that associates to a countable discrete group $\Gamma$ its reduced group C*-algebra $C^*_{red}(\Gamma).$
One can see that the composition of functors $C^*_{red}\circ \Upsilon$ is equal to the functor $\Psi_Q$.
Therefore, $\scrQ_0$ is isomorphic to the reduced C*-algebra $C^*(\wh G_{fr}).$

Let $\iota_t: \Ch(\wh G_t)\to \ot_\D \Ch(\wh G)$ and $j_t:\wh G_t\to \wh G_{fr}$ be the embeddings given by the directed systems of groups and C*-algebras.
Observe that they induce an embedding $\kappa_t$ of $\Ch(\wh G_t)\rtimes_{\alg} \wh G_t$ inside $\ot_\D \Ch(\wh G)\rtimes_{\alg} \wh G_{fr}$.
Moreover, the family of embedding $\kappa_t$ is compatible with the directed system of the $M_t$ and thus we obtain an embedding $\kappa$ at the limit of the union of the $\Ch(\wh G_t)\rtimes_{\alg} \wh G_t$ inside the crossed-product $\ot_\D \Ch(\wh G)\rtimes_{\alg} \wh G_{fr}$.
The range of this embedding is dense since the algebras $\Ch(\wh G_t)$ with $t\in\fT$ generate the tensor product $\ot_\D \Ch(\wh G)$ and the C*-algebras $C^*(\wh G_t)$ generate $C^*(\wh G_{fr}).$
Taking its inverse and extending it we obtain the desirable embedding.

Consider a tree $t$ with $n$ leaves and put $m_t$ the $n$th tensor power of the probability measure $m\in \Prob(\wh G).$
Check that if $\chi\in \wh G_t$ and $b_\chi\in  \Ch(\wh G_t)$, then $\omega_t( b_\chi \la_\chi) = \delta_{\chi,e} \int_{\wh G_t} b_e(\chi) dm_t(\chi).$
This implies that $\varpi(\sum_{\chi\in \wh G_{fr}} b_\chi u_\chi ) = \varpi(b_e)$ for $\sum_{\chi\in \wh G_{fr}} b_\chi u_\chi \in\scrM_0$ with $b_\chi=0$ for all but finitely many $\chi\in\wh G_{fr}.$
We conclude by observing that $\int_{{\wh G}^\D} \iota_t(b)\ d m^\D = \int_{G_t} b \ dm_t$ for any $b\in \ell^\infty(G_t).$

Since the map $\D\to \Prob(\wh G)$ is constant we have by Theorem \ref{omegaNC} that the Jones' action is state-preserving and thus extends to an action on the von Neumann algebra $\scrM.$
By Proposition \ref{prop:TensorProd} we have that the Jones' action restricted to $\ot_\D \Ch(\wh G)$ is the generalized Bernoulli shift given by $V\act \D$.
The group $\wh G_{fr}=\varinjlim_{t\in\fT} \wh G_t$ is constructed via the tensor functor $\Upsilon:\cF\to\Gr$ and thus carries a Jones' action $V\act \wh G_{fr}.$
Consider $v\in V$, there exists $n\geq 1$, a permutation $\tau$ and two s.d.p.: $d_1=0<d_2<\cdots<d_n<d_{n+1}=1$ and $d'_1=0<d'_1<\cdots<d'_n<d'_{n+1}=1$ such that $v$ sends the interval $[d_j,d_{j+1})$ onto the interval $[d'_{\tau(j)} , d'_{\tau(j+1)} ).$
Let $t$ be the tree associated to the first s.d.p.~and observe that an element $g\in \wh G_t$ is a $n$-tuple $(g_1,\cdots,g_n)$ such that $g(d) = g_j$ if $d\in [d_j,d_{j+1})$ as an element of $\wh G_{fr}.$
If we apply $v$ to $g$ we obtain the map $vg:\D\to \wh G$ satisfying $vg(d')=g_j$ if $d'\in [d'_{\sigma(j)} , d'_{\sigma(j+1)} ) = v ([d_j,d_{j+1}))$ and thus $vg(d') = g(v^{-1}d').$
This implies that $V$ acts by shifting the $\D$-indices on $\wh G_{fr}.$
All together we obtain the desired formula for the Jones' action.
\end{proof}

Recall that $\Ga:=\prod_\D G$ is the gauge group that acts on $(\scrM,\varpi)$ in a state-preserving way, see Section \ref{sec:gauge} for details.
We describe the gauge group action $\ga:\Ga\act \scrM$ in our current dual picture.
Consider a tree $t$ with $n$ leaves with associated s.d.p.~$\dn.$
Given a character $\chi \in\wh G_t$ there exists $\chi_1,\cdots,\chi_n \in\wh G$ such that $\chi = \chi_1\ot \cdots \ot \chi_n.$
If $s\in \Ga$, we define the quantity 
\begin{equation}\label{equa:gaugedual}s[\chi]:= \prod_{j=1}^n \chi_j(s(d_j)^{-1} s(d_{j+1}) ) , \end{equation} and recall that $s(1)=s(0)$ since we work in the torus.
Consider the directed system of groups $(\wh G_t,t\in\fT)$ obtained from the tensor functor $\Upsilon:\cF\to\Gr$ described by $\Upsilon(1)=\wh G$ and $\Upsilon(Y)(\chi)=(\chi,\chi)$ that was introduced in the proof of Proposition \ref{prop:AlgCP}.
An easy computation shows that the map $(s,\chi)\mapsto s[\chi]$ is compatible with the directed system of groups $(\wh G_t,\ t\in\fT)$ associated to the functor $\Upsilon$, i.e.~$s[\chi] = s[\Upsilon(f)(\chi)]$ for any $s\in\Ga, t\in\fT, \chi \in \wh G_t$ and $f\in\cF.$ We continue to denote by $s[\chi]$ the value of above when $\chi$ is in the limit group $\wh G_{fr}.$

\begin{proposition}
The gauge group action $\ga:\Ga\act \scrM$ satisfies the formula:
$$\ga(s)\left(\sum_{\chi\in\wh G_{fr}} a_\chi u_\chi \right) = \sum_{\chi\in\wh G_{fr}} s[\chi] a_\chi u_\chi,$$
where $a_\chi\in\ot_\D \Ch(\wh G)$ and $\chi\mapsto a_\chi$ has finite support.
\end{proposition}
\begin{proof}
By density, it is sufficient to check the formula at a tree level $t\in\fT$ for an element of the algebraic subalgebra $ \Ch(\wh G_t)\rtimes_{\alg} \wh G_t.$
Since $G$ is abelian the formula of Proposition \ref{prop:gaugeCstar} implies that the gauge group acts trivially on $\overline{\C[G_t]}$ which corresponds, in the dual picture, to $ \Ch(\wh G_t)$.
It is then sufficient to check that $\gamma(s)(u_\chi) = s[\chi] u_\chi$ for $s\in \Ga$ and $\chi\in \wh G_t$.
Let $d_1=0<d_2<\cdots< d_n<d_{n+1} =1$ be the s.d.p.~associated to $t$ and write $\chi=\chi_1\ot\cdots\ot \chi_n$ with $\chi_j\in\wh G.$
The unitary operator $u_\chi$ is the conjugation of the continuous map $\chi=\chi_1\ot\cdots\ot \chi_n\in C(G_t)$ by the Fourier transform.
To avoid notational confusions we write $\ga'$ the gauge group action in the classical picture of Section \ref{sec:gauge} by opposition to $\ga$ that is the gauge group action in the dual picture. 
The formula of Proposition \ref{prop:gaugeCstar} gives that $\ga'(s)(\chi)(g) = \chi(Z(s^{-1} ) g) = \prod_{j=1}^n \chi_j(s(d_j)^{-1} g_j s(d_{j+1}) ) = s[\chi] \chi(g)$ and thus $\ga'(s)(\chi) = s[\chi]\chi.$
This implies the desired formula once we conjugate this expression with the Fourier transform.
\end{proof}

Let $(\pi,\fH,\Omega)$ be the triple associated to the GNS construction of $(\scrM_0,\varpi).$
In order to fully understand the weak completion $\scrM$ we will describe the action of the *-algebra $\ot_\D \Ch(\wh G)\rtimes_{\alg} \wh G_{fr}$ on $\fH$.
Given $g\in \wh G_{fr}$, we denote by $g_*m^\D$ the pushforward measure defined by the formula
$$g_*m^\D(C) = m^\D(g^{-1} \cdot C) = \prod_{d\in\D} m^\D(g(d)^{-1} C(d)),$$
for any measurable cylinder $C=\prod_{d\in \D} C(d).$

\begin{lemma}\label{lem:GNSmodule}
There exists a unitary transformation 
$$W:\fH\to \bigoplus_{k\in \wh G_{fr}} L^2( {\wh G}^\D ,k_* m^\D )$$
satisfying that 
$$\Ad(W) \pi( a) ( \xi_k )_{k\in \wh G_{fr}} =( a \xi_k )_{k\in \wh G_{fr}} \text{ and } \Ad(W) \pi( u_g ) (\xi_k )_{k\in \wh G_{fr}} = (\xi_{g^{-1}k }^g )_{k\in \wh G_{fr}}$$
for all $a\in\ot_\D \Ch(\wh G), g\in\wh G_{fr}$ and vector $\xi = ( \xi_k )_{ k\in \wh G_{fr} } $ where $\xi_{ g^{-1} k }^g (x) := \xi_{ g^{-1} k } ( g^{-1} x )$.
\end{lemma}

\begin{proof}
Set $B:=\ot_\D \Ch(\wh G)$ that we consider as a space of maps from ${\wh G}^\D$ to $\C$.
By definition of the GNS construction and by Proposition \ref{prop:AlgCP} we have that $\fH$ is the closure of $B\rtimes_{\alg} \wh G_{fr}$ w.r.t.~the inner product $\langle \xi,\eta\rangle := \varpi(\eta^*\xi)$.
Decompose the vector space $B\rtimes_{\alg} \wh G_{fr}$ as the direct sum $\oplus_{k\in\wh G_{fr}} B u_k$.
Observe that if $\xi,\eta\in B$ and $g,k\in \wh G_{fr}$, then 
$$\langle \xi u_g , \eta u_k \rangle = \begin{cases} \int_{{\wh G}^\D } \ \overline\eta \xi \ d( g_* m^\D) \text{ if } g=k\\ 0 \text{ otherwise } \end{cases} .$$
Therefore, $\fH\simeq\oplus_{k\in \wh G_{fr}} \fH_k$ where $\fH_k=L^2({\wh G}^\D, k_*m^\D)$ and write $W$ this unitary transformation.

The other equalities are routine computations.
\end{proof}

Recall that two measures on a measure space are singular if their supports are disjoint and are equivalent (i.e.~are in the same measure class)  if they have the same null sets.
We recall the fundamental theorem of Kakutani for infinite products of measures that will be useful for our study. 

\begin{theo}\label{theo:Kakutani}\cite{Kakutani48}
Let $X$ be a measurable space with two infinite family of probability measures $(\mu_d, \nu_d, d\in\D).$
Assume that $\mu_d$ is equivalent to $\nu_d$ for any $d\in\D$ and put $\mu:=\ot_{d\in\D} \mu_d$ and $\nu:=\ot_{d\in\D} \nu_d$ the infinite tensor products of measures defined on the product $\sigma$-algebra of the product space $X^\D.$
Denote by $\rho(\mu_d,\rho_d)$ the quantity $\sum_{\chi\in\wh G} \sqrt{\mu_d(\chi) \nu_d(\chi)}$ that is in $(0,1].$
Then the measure $\mu$ and $\nu$ are either equivalent or singular with each other.
Moreover, they are equivalent if and only if 
$$ - \sum_{d\in\D} \log(\rho(\mu_d,\nu_d)) <+\infty.$$
\end{theo}

We deduce the following lemma.
\begin{lemma}\label{lem:Kakutani}
The measures $g_* m^\D$ and $g'_* m^\D$ are mutually singular if and only if there exists $d\in \D$ such that $g(d)_* m \neq g'(d)_* m$ where $g,g'\in\wh G_{fr}.$
Otherwise the measures are equal.
If the group $\wh G$ is torsion free, then the family of measures $(g_*m^\D, g\in\wh G_{fr})$ are mutually singular.
\end{lemma}

\begin{proof}
Consider $g,g'\in \wh G_{fr}$ and assume that there exists $d_0\in\D$ satisfying $g(d_0)_*m\neq g'(d_0)_*m$.
Since $g$ and $g'$ are locally constant we can find an infinite subset $E\subset\D$ such that $g(d ) =g(d_0)$ and $g'(d ) =g'(d_0)$ for any $d\in E.$
We have that 
\begin{align*}
 - \sum_{d\in\D} \log(\rho( g(d)_* m , g'(d)_*m)) & \geq - \sum_{d\in E} \log(\rho( g(d)_* m , g'(d)_* m))\\
 & = - | E | \log(\rho( g(d_0)_*m , g'(d_0)_*m ) )  = +\infty.
 \end{align*}
Kakutani's theorem implies that the measures $g_*m^\D$ and $g'_*m^\D$ are mutually singular.
If such a $d$ does not exists, then by definition the two measures $g_*m^\D$ and $g'_*m^\D$ are equal.

Assume that $\wh G$ is torsion free and let $\chi\in\wh G$ satisfying that $\chi_*m=m.$
If $\chi\neq e$, then $m(\wh G)\geq m(\chi^\Z) = | \Z | m(\{\chi\}) = +\infty$, a contradiction.
Hence, $\chi=e$.
Therefore, if the two measures $g_*m^\D$ and $g'_* m^\D$ are equivalent, then $g(d)g'(d)^{-1}=e$ for any $d\in\D$ implying that $g=g'.$
\end{proof}

Define the subgroup $N<\wh G$ of $g$ such that $g_*m=m$ and set 
$$N_{fr}:=\{g\in \wh G_{fr},\ g(d)\in N, \ \forall d\in\D\}.$$
Let $\sigma:\wh G_{fr}/N_{fr}\to \wh G_{fr}$ be a section and consider the cocycle 
$$\kappa: \wh G_{fr}\times \wh G_{fr}/N_{fr}\to N_{fr}, \ (g,\ga)\mapsto \sigma(g\ga)^{-1} g \sigma(\ga).$$
Define the group action 
\begin{equation}\label{equa:GrX} \wh G_{fr}\act ({\wh G}^\D \times \wh G_{fr}/N_{fr}), \ g\cdot (z, \ga) := (\kappa(g,\ga) z , g \ga).\end{equation}
Equipped the space ${\wh G}^\D \times \wh G_{fr}/N_{fr}$ with the measure $m^\D\ot \mu_c$ where $\mu_c$ is the counting measure.
Since $N_{fr}$ acts in a p.m.p.~way on $({\wh G}^\D,m^\D)$ we obtain that the action $\wh G_{fr}\act ({\wh G}^\D \times \wh G_{fr}/N_{fr} , m^\D\ot \mu_c)$ is measure preserving.
We write $$\mathscr B:= L^\infty({\wh G}^\D\times \wh G_{fr}/N_{fr}, m^\D\ot \mu_c)\rtimes \wh G_{fr}$$ the crossed-product von Neumann algebra. 
Observe that since the measure is invariant under the action we have that $\scrB$ does not have any type III component.
Consider the unique normal state $\varpi_\mathscr B$ on $\mathscr B$ satisfying that 
$$\varpi_\scrB (\sum_{g\in \wh G_{fr}} A_g v_g) = \int_{ {\wh G}^\D } A_e(z, \bar e) dm^\D(z),$$
where $A_g\in L^\infty({\wh G}^\D\times \wh G_{fr}/N_{fr}, m^\D\ot \mu_c)$ for $g\in\wh G_{fr}$ and where $v_s, s\in \wh G_{fr}$ denotes the unitaries of $\mathscr B$ implementing the action.
There is a natural action of Thompson's group $V\act \D$ induced by its action on the unit interval.
This provides an action on $\wh G^\D$ consisting in shifting indices that restricts to the subgroup $\wh G_{fr}$ and passes to the quotient $\wh G_{fr}/ N_{fr}.$
Write $\al_\mathscr B:V\act \mathscr B$ the action induced by those three (i.e.~$\al_\scrB(v)\left(\sum_{g\in\wh G_{fr}} A_g u_g\right) = \sum_{g\in\wh G_{fr}} A_g(v^{-1} \cdot ) u_{vg}$ for $v\in V, A_g\in L^\infty({\wh G}^\D\times \wh G_{fr}/N_{fr}, m^\D\ot \mu_c)$ ) which clearly acts by automorphisms letting the state $\varpi_\mathscr B$ invariant.

We are now able to prove the main theorem of this section.

\begin{theo}\label{theo:singlemes}
Let $G$ be a compact abelian separable group and $m\in\Prob(\wh G)$ a strictly positive probability measure.
We have a state-preserving isomorphism of von Neumann algebras 
$$\psi:(\scrM,\varpi)\to (\mathscr B,\varpi_\mathscr B)$$ and in particular $\scrM$ does not have any type III component.
Moreover, the Jones' action $\al:V\act \scrM$ preserves the state $\varpi$ and is conjugate to the action $\alpha_\scrB$, i.e.~$\Ad(\psi)\circ\alpha=\alpha_\scrB.$

If $\wh G$ is torsion free (e.g.~$G$ is the circle group), then $(\scrM,\varpi)$ is isomorphic to $$(L^\infty({\wh G}^\D, m^\D)\ovt B(\ell^2(\wh G_{fr})),m^\D\ot \langle \cdot \de_e,\de_e\rangle),$$
which is a type I$_\infty$ von Neumann algebra with a diffuse center and equipped with a non-faithful state.

If $G$ is a finite group and $m$ is $\wh G$-invariant (i.e.~$m$ is the Haar measure of $\wh G$), then $(\scrM,\varpi)$ is isomorphic to the hyperfinite type II$_1$ factor equipped with its trace.
\end{theo}

\begin{proof}
Put $(X,\mu):=({\wh G}^\D\times \wh G_{fr}/N_{fr}, m^\D\ot \mu_c).$
Consider the group action of \eqref{equa:GrX}: 
$$\wh G_{fr}\act (X,\mu), \  g\cdot (z,\ga) = (\sigma(g\ga)^{-1} g \sigma(\ga) z, g\ga), \ g\in \wh G_{fr}, z \in {\wh G}^\D, \ga\in \wh G_{fr}/N_{fr}$$ and recall that it is measure preserving.
If $\xi:X\to \C$ is any map and $g\in \wh G_{fr}$, we write $\xi^g$ the map $\xi^g(x):=\xi(g^{-1} x), x\in X$ for this action.

Write elements of $\scrB$ as formal sums $\sum_{g\in\wh G_{fr}} A_g v_g$ with $A_g\in L^\infty(X,\mu)$.
Assume that $\scrB$ is represented on the Hilbert space $K:=L^2(X,\mu)\ot \ell^2(\wh G_{fr})$ and acts in the following classical way: 
$$v_g(\xi\ot \delta_k) := \xi^{g}\ot \delta_{gk}$$
for $g,k\in\wh G_{fr}$, $\xi \in L^2(X,\mu)$ and 
$$A (\xi\ot\delta_k) = (A\xi)\ot \delta_k$$
where $(A\xi)(y):= A(y)\xi(y) \text{ and } A\in L^\infty(X,\mu).$

Given $a\in \ot_\D \Ch(\wh G)$ we define the map 
$$\psi(a)\in L^\infty(X,\mu) \text{ such that } \psi(a)(z,\ga) := a(\sigma(\ga)z),$$
for $z\in {\wh G}^\D, \ga\in \wh G_{fr}/N_{fr}$.
We easily check that $\overline{\psi(a)} = \psi(\overline a)$ and that $\psi(a^g) = \psi(a)^g$ for any $a\in \ot_\D \Ch(\wh G), g\in \wh G_{fr}$ where $a^g(y):= a(g^{-1} y)$ for $y\in{\wh G}^\D$.
This defines the *-morphism
$$\psi  :\ot_\D \Ch(\wh G)\rtimes_{\alg} \wh G_{fr}\to \scrB, \ 
\sum_{g\in\wh G_{fr}} a_g u_g \mapsto \sum_{g\in\wh G_{fr}} \psi(a_g) v_g.$$
Moreover, 
\begin{align*}
\varpi_\scrB\circ\psi\left(\sum_g a_g u_g \right) & = \varpi_\scrB\left(\sum_g \psi(a_g) v_g \right)\\
& = \int_{{\wh G}^\D} \psi(a_e)(z,\bar e) d m^\D(z)
 = \int_{{\wh G}^\D} a_e(\sigma(\bar e)z) d m^\D(z) \\
& = \int_{{\wh G}^\D} a_e(z) d m^\D(z) \text{ since $\sigma(\bar e)\in N_{fr}$ implies $\sigma(\bar e)_*m^\D=m^\D$}\\
& = \varpi\left(\sum_g a_g u_g \right).
\end{align*}
Therefore, $\psi$ is a densely defined *-morphism from $(\scrM,\varpi)$ to $(\scrB,\varpi_\scrB)$ that is state-preserving. 

Let us prove that $\psi$ extends to a normal isomorphism from $\scrM$ onto $\scrB.$
Identify $K$ with the space $L^2({\wh G}^\D,m^\D)\ot \ell^2(\wh G_{fr}/N_{fr})\ot \ell^2(\wh G_{fr})$.
Under this identification we have that 
$$v_g(\eta\ot \delta_\ga\ot \delta_k) = \eta^{\kappa(g,\ga)} \ot \delta_{g\ga}\ot \delta_{gk}$$
and 
$$A(\eta\ot \delta_\ga\ot \delta_k) = (A(\cdot, \ga)\eta )\ot \delta_\ga\ot \delta_k$$
for $g,k\in \wh G_{fr}, \eta\in L^2({\wh G}^\D,m^\D), \ga\in \wh G_{fr}/N_{fr}$ and $A\in L^\infty(X,\mu).$
In particular, if $a\in \ot_\D \Ch(\wh G)$, then 
\begin{equation}\label{equa:A-module}\psi(a)(\eta\ot \delta_\ga\ot \delta_k) = (a(\sigma(\ga)\cdot )\eta )\ot \delta_\ga\ot \delta_k.\end{equation}
For any $\ga\in \wh G_{fr}/N_{fr}$ consider the subspace 
$$K_\ga:=  \{ \sum_{k\in \wh G_{fr}} \eta_k \ot \delta_{k\ga} \ot \delta_k , \ \eta_k\in L^2({\wh G}^\D,m^\D) , \sum_k  \| \eta_k\|^2<\infty \} \subset K $$
and observe that is is closed under the action of $\scrB$ making it a $\scrB$-module.
We have the following decomposition into $\scrB$-modules:
$$K=\oplus_{\ga\in \wh G_{fr}/N_{fr}} K_\ga.$$
Consider the transformation 
$$U_\ga: K_\ga \to \oplus_{k\in\wh G_{fr}} L^2({\wh G}^\D, k_*m^\D), \ \sum_{k\in\wh G_{fr}}\eta_k\ot\delta_{k\ga}\ot\delta_k \mapsto (\eta_{k} ^{\sigma(k\ga)})_{k\in\wh G_{fr}}$$
and check that it is a unitary transformation.
Moreover, it satisfies that 
$$U_\ga \psi(au_g) \xi =\pi( a u_g) U_\ga \xi$$ for any $\xi\in K_\ga$ and $a\in \ot_\D \Ch(\wh G), g\in\wh G_{fr}$ where $(\pi,\fH)$ is the GNS representation associated to $(\ot_\D\Ch(\wh G)\rtimes_{\alg}\wh G_{fr},\varpi)$ described in Lemma \ref{lem:GNSmodule}.

Therefore, $U_\ga$ provides an isomorphism between the $\ot_\D \Ch(\wh G)\rtimes_{\alg} \wh G_{fr}$-modules $K_\ga$ and the GNS module $\fH$ given by the state $\varpi$.
We obtain that $\ot_\D\Ch(\wh G)\rtimes_{\alg}\wh G_{fr}$-module $K$ is isomorphic to a direct sum of copies $\fH$ indexed by $\wh G_{fr}/N_{fr}$. 
This implies that $\psi$ extends to an isomorphism from $\scrM$ onto the weak closure of the range of $\psi.$

Let us show that the range of $\psi$ is weakly dense inside $\scrB$.
The set of unitary operators $(v_g,g\in\wh G_{fr})$ is by definition in the range of $\psi$.
Therefore, it is sufficient to prove that the weak closure of $\psi(\ot_\D \Ch(\wh G))$ is equal to $L^\infty(X,\mu).$
For $g\in \wh G_{fr}$ we define the representation $\pi_g:\ot_\D \Ch(\wh G)\to B(L^2({\wh G}^\D,g_*m^\D))$ where the maps act by pointwise multiplications.
Observe that the range of $\pi_g$ is weakly dense inside $L^\infty({\wh G}^\D,g_*m^\D)$ since $\Ch(\wh G)$ is weakly dense inside $\ell^\infty(\wh G)$ and since $L^\infty({\wh G}^\D,g_*m^\D)$ is obtained by performing the infinite tensor product $\ot_{d\in\D} (\ell^\infty(\wh G), g(d)_* m)$.
By Lemma \ref{lem:Kakutani} the family of measures $(\sigma(\ga)_* m^\D, \ga\in \wh G_{fr}/N_{fr})$ are mutually singular implying that the representations $(\pi_{\sigma(\ga)} ,\ga \in \wh G_{fr} /N_{fr} )$ are mutually disjoint.
A classical argument (see for instance \cite[Chapter 2]{Arveson76}) states that the weak completion of a sum of disjoint representations is the sum of the weak completions implying that the weak completion of $\oplus_{\ga \in \wh G_{fr} /N_{fr} } \pi_{\sigma(\ga)} (\ot_\D \Ch(\wh G))$ is equal to $\oplus_{\ga \in \wh G_{fr} /N_{fr} } L^\infty({\wh G}^\D,\sigma(\ga)_*m^\D)$. 
Identify $L^\infty(X,\mu)$ with the direct sum $\oplus_{\ga\in\wh G_{fr}/N_{fr}} L^2({\wh G}^\D, \sigma(\ga)_* m^\D)$ and observe that $K$, as a $L^\infty(X,\mu)$-module, is isomorphic to a direct sum of the regular module $\oplus_{\ga\in\wh G_{fr}/N_{fr}} L^2({\wh G}^\D, \sigma(\ga)_* m^\D)$.
Equation \ref{equa:A-module} gives that the range $\psi(\ot_\D \Ch(\wh G))$ is under this identification the diagonal subalgebra $\oplus_{\ga \in   \wh G_{fr}/N_{fr} } \pi_{\sigma(\ga)}(\ot_\D \Ch(\wh G)).$
Therefore, the weak closure of $\psi(\ot_\D \Ch(\wh G))$ is indeed equal to $L^\infty(X,\mu)$ implying that the range of $\psi$ is weakly dense inside $\scrB.$

Since $m^\D\ot \mu_c$ is invariant for the action of $\wh G_{fr}$, we obtain that the von Neumann crossed-product $\scrB$ has direct summands of type I or II.

Let $\al:V\act \scrM$ be the Jones' action that is state-preserving by Theorem \ref{omegaNC}.
Proposition \ref{prop:AlgCP} implies that the actions $\Ad(\psi)\circ \al$ and $\al_\scrB$ coincide on the subalgebra $\psi(\ot_\D  \Ch(\wh G)\rtimes \wh G_{fr})$ that is weakly dense and are thus equal.

If $\wh G$ is torsion free, then Lemma \ref{lem:Kakutani} implies that $N_{fr}$ is trivial.
This implies that the cocycle $\kappa$ is trivial and thus the action $\wh G_{fr}\act ({\wh G}^\D\times \wh G_{fr})$ is given by the formula $g\cdot (x,k):= (x,gk).$
Therefore, $\scrB = L^\infty({\wh G}^\D,m^\D)\ovt  \ell^\infty(\wh G_{fr})\rtimes \wh G_{fr} \simeq L^\infty({\wh G}^\D,m^\D)\ovt B(\ell^2(\wh G_{fr})),$ where $\ovt$ refers to the von Neumann crossed-product.
The statement about the state is obvious.

Assume that $\wh G$ is a finite group and that $m$ is $\wh G$-invariant.
Therefore, $N=\wh G$ and thus $N_{fr}=\wh G_{fr}.$
Hence, $\scrB = L^\infty({\wh G}^\D, m^\D)\rtimes \wh G_{fr}$ for the action $g\cdot x =gx$.
We have $\varpi_\scrB(\sum_{g\in\wh G_{fr}} a_g v_g) = m^\D(a_e).$ 
This state is clearly faithful and is tracial implying that $\scrB$ is of finite type.
Moreover, $\scrB$ is generated by a net of finite factors $B(\ell^2(G_t)), t\in\fT$ implying that $\scrB$ is a factor by \cite[Chapt. XIV, 2, Lem. 2.13]{TakesakiIII} and is manifestly hyperfinite since $B(\ell^2(G_t))$ is finite dimensional for any $t\in\fT.$
Since $\scrB$ is infinite dimensional we obtain that $(\scrB,\varpi_\scrB)$ is necessarily the hyperfinite II$_1$ factor equipped with its trace.
\end{proof}

\begin{remark}
Observe that at a tree-level $M_t$ our state $\omega_t$ is faithful and in general nontracial.
Using Tomita-Takesaki theory, this gives us a nontrivial time evolution for our system.
Unfortunately, at the limit, the state $\varpi$ is no longer faithful unless it is tracial.
One classical method to palliate the non-faithfulness is to consider a maximal projection $p\in\scrM$ such that $\varpi$ is faithful on the compression $p\scrM p$ and to consider  $(p\scrM p, \varpi(p)^{-1}\varpi(p\cdot p))$. 
Unfortunately, in our case, we will always end up with a tracial state and thus with a trivial time evolution.
For example, if $\wh G$ is torsion free, we will have that $(p\scrM p, \varpi(p)^{-1}\varpi(p\cdot p))$ is the commutative von Neumann algebra $L^\infty({\wh G}^\D, m^\D)$ together with the measure $m^\D$ for a state. 
\end{remark}

\begin{remark}
One can easily recover the algebras of fields $\scrM(O)$ for $O$ a connected open subset of the torus.
Indeed, in the definition of $\scrB$ replace each $\D$ by the subset $\D(O)$ and use the appropriate embeddings/projections for having a description of $\scrM(O)\subset \scrM$ in term of crossed-products von Neumann algebras.
\end{remark}

\subsection{The circle group case together with a family of heat-kernel states}\label{sec:Zleaves}
From now one consider the circle group $G=\mathbf S$ and its Pontryagin dual $\wh G$ that we identify with the infinite cyclic group $\Z.$
For any $b>0$, we define the heat-kernel states, cp.~Remark \ref{rem:heatkernel},
$$h_b:\Z\to \R, n\mapsto \dfrac{e^{-n^2 b/2}}{Z_b} \text{ where } Z_b = \sum_{n\in\Z} e^{-n^2 b/2}.$$ 
This defines a probability measure $m_b\in\Prob(\Z)$ via the formula
$$m_b(A) = \sum_{n\in A} h_b(n), \ A\subset\Z.$$

Next, we consider a map $\beta:\D\to \R_{>0}$ and consider the family of functions, measures and states indexed by $d\in\D$:
\begin{align*} 
h_{\beta(d)}(n) & = \dfrac{e^{-n^2 \beta(d)/2}}{Z_{\beta(d)}},\ n\in\Z;\\
m_{\beta(d)} (A) & = \sum_{n\in A} \dfrac{e^{-n^2 \beta(d)/2}}{Z_{\beta(d)}}, \ A\subset \Z;\\
\omega_{\beta(d)}(a) & =\Tr(a \la(h_{\beta(d)} ) ),\ a\in B(\ell^2(\Z)).
\end{align*}

Equipped $\Z^\D$ with the product $\sigma$-algebra $\Sigma^\D$ generated by cylinders and define the product probability measure $m_\beta^\D:=\ot_{d\in\D} m_{\beta(d)}$ on $(\Z^\D , \Sigma^\D).$
Define the countable discrete abelian group $\Z_{fr}$ of maps $g:\D\to \Z$ satisfying that there exists a s.d.p.~$d_1=0< d_2 <\cdots<d_n<d_{n+1} = 1$ such that $g$ is constant on each half-open interval $[d_j,d_{j+1}), 1\leq j\leq n.$
Consider the action 
\begin{equation}\label{equa:thetaZ}\theta:\Z_{fr}\act \Z^\D, (g\cdot x)(d) = g(d) + x(d).\end{equation}
Using Theorem \ref{theo:Kakutani}, we obtain the following criteria.

\begin{proposition}\label{prop:nonsingular}
The action $\theta:\Z_{fr}\act \Z^\D$ is nonsingular w.r.t.~the measure $m^\D_\beta$ if and only if $\beta\in\ell^{1}(\D)$, i.e.~$\sum_{d\in\D}\beta(d)<\infty.$
\end{proposition}
\begin{proof}
Consider an integer $k\in\Z$, a real number $b>0$ and the associated heat-kernel measure $m_b.$ 
Observe that 
\begin{align*}
\rho(k_*m_b, m_b) &:= \sum_{n\in\Z} \sqrt{ m_b(n-k) m_b(n)}\\
& = \sum_{n\in\Z} Z_b^{-1} \exp(-(n-k)^2b/4 - n^2 b/4)\\
& = Z_b^{-1} \exp(-k^2 b/8) \sum_{n\in\Z} \exp( - (n- k/2)^2 b/2).
\end{align*}

We treat two cases.
Assume that $k$ is even.
Then 
$$\sum_{n\in\Z} \exp( - (n-k/2)^2 b/2)=\sum_{n\in\Z} \exp( - n^2 b/2)=Z_b$$ 
and thus $\rho(k_*m_b, m_b)=\exp(-k^2 b/8).$

Assume that $k$ is odd.
Let us show that $\rho(k_*m_b, m_b)\sim_{b\to 0} \exp(-k^2 b/8).$
Observe that 
\begin{align*}
\rho(k_*m_b, m_b) & = Z_b^{-1} \exp(-k^2 b/8) \sum_{n\in\Z} \exp( - (2n-k)^2 b/8)\\
& = Z_b^{-1} \exp(-k^2 b/8) \sum_{n\in\Z} \exp( - (2n-1)^2 b/8)\\
& = Z_b^{-1} \exp(-k^2 b/8) [ \sum_{p\in\Z} \exp( - p^2 b/8) - \sum_{q\in\Z} \exp( - (2q)^2 b/8) ] \\
& = Z_b^{-1} \exp(-k^2 b/8) [ Z_{b/4} - Z_b]\\
& = \exp(-k^2 b/8)\dfrac{ Z_{b/4} - Z_b }{ Z_b }.\\
\end{align*}
We now approximate $Z_b$.
Consider the theta function $\theta(z)=\sum_{k\in\Z} e^{i\pi k^2 z}$ defined on the upper half-plane and recall the Jacobi-Poisson inversion formula that is 
\begin{equation}\label{equa:theta}\theta(-z^{-1})^2 = -iz \theta(z)^2.\end{equation}
In particular, $$Z_b = \sqrt{ \dfrac{2\pi}{b} } \sum_{n\in\Z} e^{-2\pi^2 n^2/b} = \sqrt{ \dfrac{2\pi}{b} } (1+ 2\sum_{n=1}^\infty (e^{-2\pi^2/b})^{n^2} ).$$
Observe that 
$$\sum_{n=1}^\infty (e^{-2\pi^2/b})^{n^2} )\leq \sum_{n=1}^\infty (e^{-2\pi^2/b})^{n} ) = \dfrac{e^{-2\pi^2/b}}{1-e^{-2\pi^2/b}}\to_{b\to 0} 0.$$ 
This implies that 
$$\rho(k_*m_b, m_b) = \exp(-k^2 b/8)\dfrac{ Z_{b/4} - Z_b }{ Z_b } \sim_{b\to 0} \exp(-k^2 b/8).$$
Note that this approximation is uniform w.r.t.~the variable $k$.

We obtain that 
\begin{equation}\label{equa:logrho}
-\log(\rho(k_*m_b, m_b)) \sim_{b\to 0} \dfrac{bk^2}{8} \text{ for any } k\in\Z.\end{equation}

Consider $g\in\Z_{fr}$ that is nonzero.
Then $g(d)\neq0$ for infinitely many values and implying that the series $\sum_{d\in\D} -\log(\rho(g(d)_*m_{\beta(d)}, m_{\beta(d)}))$ diverges if $\beta$ does not tend to zero.
We can thus assume that $\beta$ tends to zero at infinity and since $g$ is bounded we obtain that $-\log(\rho(g(d)_*m_{\beta(d)}, m_{\beta(d)}))$ is equivalent to $\tfrac{\beta(d)g(d)^2}{8}$ at infinity in $d\in\D$.
By Kakutani's theorem, we conclude that $g_*m_\beta^\D \simeq m_\beta^\D$ if and only if $\beta$ is summable finishing the proof of the proposition.
\end{proof}

Note that the proof of the last proposition implies the following statement: 
for any map $g:\D\to \Z$ we have that $g_*m_\beta^\D\simeq \mbD$ if and only if $d\in\D\mapsto g(d)^2 \beta(d)$ is summable.

Consider the setting described in Section \ref{sec:Gabelian} applied to $\wh G = \Z$ with $M_t=\Ch(\Z_t)\rtimes \Z_t$ the C*-subalgebra of $B(\ell^2(\Z_t))$ generated by $\Ch(\Z_t)$ acting by pointwise multiplication and the group ring $\C[\Z_t]$ acting by convolution.
They are equipped with the state $\omega_t$ satisfying
$$\omega_t\left( \sum_{g\in\Z_t} b_g u_g \right) = \int_{\Z_t} b_e \ dm_t,$$
where $m_t = m_{\beta(d_1)}\ot\cdots \ot m_{\beta(d_n)}$ and $d_1=0<d_2<\cdots< d_n<d_{n+1} = 1$ is the s.d.p.~associated to the tree $t\in\fT$.
Proposition \ref{prop:AlgCP} applied to this specific example gives that the inductive limit of C*-algebras $\scrM_0$ contains densely the *-algebra $\ot_\D\Ch(\Z)\rtimes_{\alg} \Z_{fr}$ and satisfies the following formula for the Jones' action:  
$$\al:V\act \scrM_0, \ \al(v)(\sum_{g\in\Z_{fr}} b_g u_g) = \sum_{g\in\Z_{fr}} b_g(v^{-1}\cdot ) u_{vg},$$
with $b_g\in \ot_\D\Ch(\Z)$ and $g\mapsto b_g$ finitely supported.
Moreover, the state satisfies the following formula
$$\varpi\left( \sum_{g\in\Z_{fr}} b_g u_g \right) = \int_{\Z^\D} b_e \ d \mbD.$$
We let the reader verify these formulae which are all easy consequences of the preceding sections.\\[0.1cm]
Let $(\scrM,\varpi)$ be the weak completion of $\scrM_0$ w.r.t.~$\varpi$ and the normal extension of this state.
The next theorem relates our action $\theta$ (see \eqref{equa:thetaZ}) and the von Neumann algebra $\scrM$.

\begin{theo}
If $\beta\in\ell^{1}(\D)$, then there exists an isomorphism of von Neumann algebras 
$$J:\scrM\to L^\infty(\Z^\D, m^\D_\beta)\rtimes \Z_{fr}$$
satisfying
$$ \varpi \circ J^{-1}(\sum_{g\in\Z_{fr}} b_g u_g ) = \int_{ { \Z}^\D } b_e(x) d m_\beta^\D(x),$$
where $(u_g,g\in\Z_{fr})$ implements the unitary action of $\Z_{fr}$ on $L^\infty(\Z^\D,\mbD).$
\end{theo}
\begin{proof}
If $\beta$ is in $\ell^{1}(\D)$, then the action $\theta: \Z_{fr}\act (\Z^\D, \mbD)$ is nonsingular by Proposition \ref{prop:nonsingular}.
The description of $\varpi$ and the definition of crossed-product von Neumann algebras for nonsingular actions implies the rest of the theorem.
\end{proof}

Next, we prove that the action is ergodic for a large class of choices of $\beta.$
The proof is rather indirect and uses the fact that the action of the group $\oplus_\D \Z\act (\Z^\D,\mbD)$ is ergodic.
Let $\Aut(\Z^\D, \mbD)$ be the group of all nonsingular transformations of the probability measure space $(\Z^\D, \mbD).$
We equipped it with the weak topology also called coarse topology, see \cite{Choksi-Kakutani79}.
It confers to $\Aut(\Z^\D, \mbD)$ a structure of Polish topological group.
A sequence $g_n$ converges to $g$ for the weak topology if the sequence of associated Radon-Nikodym derivatives $\tfrac{d{g_n}_* \mbD}{ d \mbD}$ converges to $\tfrac{d{g}_* \mbD}{ d \mbD}$ inside $L^1(\Z^\D, \mbD)$ and if for any measurable subset $A\subset \Z^\D$ we have that $\lim_{n\to\infty}\mbD( g^{-1}A\Delta g_n^{-1} A)=0$.
This topology makes strongly continuous the action of $\Aut(\Z^\D, \mbD)$ on $L^2(\Z^\D, \mbD)$.
We will show that for certain choices of $\beta$ we have that $\oplus_\D\Z$ is in the closure of $\Z_{fr}$ inside $\Aut(\Z^\D, \mbD)$.

\begin{prop}\label{prop:closurZfr}
Consider $\beta:\D\to \R_{>0}$ such that there exists $0< p < 1/2$ such that $\beta$ is in $\ell^p(\D).$
Then the group $\oplus_\D\Z$ is in the closure of $\Z_{fr}$ inside the topological group $\Aut(\Z^\D, \mbD)$.
\end{prop}
\begin{proof}
Consider $\beta$ and $0<p<1/2$ satisfying that $\beta\in\ell^p(\D).$
Let $\tilde g_n\in\Z_{fr}$ be the element defined as $\tilde g_n(d) = 1$ if $0\leq d < 1/2^n$ and $0$ otherwise.
We want to show that $\tilde g_n$ converges weakly to $g\in\oplus_\D\Z$ that is the function $g(d) =1$ if $d=0$ and $0$ otherwise.
Note that $\oplus_\D \Z$ is obviously a subgroup of $\Aut(\Z^\D, \mbD)$ since any of its element will only change finitely many of the measures $m_{\beta(d)}$ into an equivalent one.
Note that $\tilde g_n$ tends to $g$ if and only if $g_n$ tends to the identity where $g_n(d)=1$ if $1<d<1/2^n$ and $0$ otherwise.
Denote by $\varphi_\beta$ the Radon-Nikodym derivative $\tfrac{d{g_n}_* \mbD}{ d \mbD}$ and note that 
$$\varphi_\beta=\prod_{d\in\D: 1<d< 1/2^n} f_{\beta(d)}, \text{ where } f_{\beta(d)} =  \dfrac{d1_* m_{\beta(d)}}{ d m_{\beta(d)}}.$$

Observe that $f_b(n) = e^{(2n-1)b/2}$ for any $b>0$ and $n\in\Z$ and define the subset
$$X_d:=\{ n\in\Z :\ | f_{\beta(d)}(n) -1 | \leq \beta(d)^p\}.$$
We notice that, up to the endpoints, $X_d$ is equal to the interval 
$$( 1/2 + \dfrac{\log(1 - \beta(d)^p)}{\beta(d)} , 1/2 + \dfrac{\log(1 +\beta(d)^p)}{\beta(d)} )$$ that is roughly equal to $(-\beta(d)^{p-1}, \beta(d)^{p-1})$ when $d\to\infty$ (for the Fr\'echet filter on $\D$).
Classical formula on Poisson summations provides that $m_b(-N,N)\sim \sqrt{1-e^{-N^2b/2}}$ when $b\to 0, N\to \infty.$
This implies that $$-\log( m_{\beta(d)} ( X_d ) ) \sim_{d\to \infty} \dfrac{ e^{-\beta(d)^{2p-1}/2 } }{2}.$$
Write $\D(n)$ the set of dyadic rational inside $(0,1/2^n)$ and consider the cylinder 
$$X^{(n)}=\prod_{ d\in\D(n)} X_d \times \Z^{\D\setminus \D(n)}.$$
We have that $-\log(X^{(n)}) = \sum_{d\in\D(n)} -\log( m_{\beta(d)} ( X_d ) ) $ whose coefficient is equivalent at infinity to $\tfrac{ e^{-\beta(d)^{2p-1}/2 } }{2}$ and is thus convergent since $\beta(d)^{1-2p}$ is in some $\ell^q(\D)$ spaces.
This implies that $\lim_{n\to\infty} -\log(\mbD(X^{(n)})) =0$ and thus $\lim_{n\to\infty} \mbD(X^{(n)})=1.$
If we shift $X^{(n)}$ by one unit we also obtain that its measure tends to one.

Observe that 
\begin{align*}
\Vert \varphi_{g_n} - 1 \Vert_1 & \leq \int_{X^{(n)}} |\varphi_{g_n} - 1| dm_\beta^\D + \int_{\Z^\D\setminus X^{(n)} } |\varphi_{g_n} - 1| dm_\beta^\D \\
& \leq \int_{X^{(n)}} |\varphi_{g_n} - 1| dm_\beta^\D  + {g_n}_* m_\beta^\D(\Z^\D\setminus X^{(n)}) + m_\beta^\D(\Z^\D\setminus X^{(n)}) \\
& \leq \sup_{ X^{ ( n ) } } |\varphi_{g_n} - 1|  + {g_n}_* m_\beta^\D(\Z^\D\setminus X^{(n)}) + m_\beta^\D(\Z^\D\setminus X^{(n)}). \\
\end{align*}
By our previous discussion we have that the two right terms tend to zero.
By definition we have that if $x\in X^{(n)}$, then 
$\log(\varphi_{g_n}(x)) = \sum_{d\in\D(n)} \log(f_{\beta(d)}(x(d))) \leq  \sum_{d\in\D(n)} \log( 1 + \beta(d)^p).$
We have a similar lower bound that is $\sum_{d\in\D(n)} \log( 1 - \beta(d)^p)$.
Note that at infinity the principal term of the series is equivalent to $\beta(d)^p$ that is summable implying that $\log(\varphi_{g_n}(x))$ converges uniformly to $0$ when $n$ tends to infinity.
Therefore, $\lim_{n\to\infty} \sup_{ X^{ ( n ) } } |\varphi_{g_n} - 1| =0$ implying that $\lim_{n\to\infty}\Vert \varphi_{g_n} - 1 \Vert_1=0.$

Consider a measurable subset $A\subset \Z^\D$ and $\varep>0.$
Recall that an elementary cylinder is of the form $\prod_{d\in\D} C_d$ with $C_d=\Z$ for all but finitely many $d$ and that the $\sigma$-algebra of $\Z^\D$ is generated by them.
There exists an elementary cylinder $C$ satisfying that $\mbD(A\Delta C)<\varep.$
Observe that for $n$ large enough we have that $g_nC=C$ since the support of $g_n$ will eventually be contained in the set of $d$ satisfying $C_d=\Z$.
Observe that 
$$\mbD(A\Delta g_n^{-1}A)\leq \mbD(A\Delta C)+\mbD(C\Delta g_n^{-1}C)+\mbD(g^{-1}_nC\Delta g_n^{-1}A)\leq \varep + \mbD(g^{-1}_n(C\Delta A))$$
for $n$ large enough.
Moreover, $\mbD(g^{-1}_n(C\Delta A)) = \int_{C\Delta A} \varphi_{g_n} d \mbD$ and this quantity tends to $\mbD(C\Delta A)$ since the Radon-Nikodym derivative $\varphi_{g_n}$ tends to one in $L^1(\Z^\D,\mbD)$ implying that $\mbD(g^{-1}_n(C\Delta A))\leq \varep+\mbD(C\Delta A)\leq 2\varep$ for $n$ large enough.
We obtain that $\lim_{n\to\infty} \mbD(A\Delta g_n^{-1}A)=0$ and thus $g_n$ tends to the identity $e$ inside $\Aut(\Z^\D, \mbD)$.
Since any $\tilde g_n$ is in $\Z_{fr}$ we get that $g$ belongs to the closure $\overline{\Z_{fr}}.$
Given any other Dirac element $g_{d_0}$ at another dyadic rational $d_0$ (i.e.~$g_{d_0}(d_0)=1$ and $g_{d_0}(d)=0,\ \forall d\neq d_0$) we can adapt the proof of above to obtain that $g_{d_0}\in\overline{\Z_{fr}}.$
Since $\overline{\Z_{fr}}$ is a group and since $\oplus_\D\Z$ is generated by those Dirac elements we obtain that $\oplus_\D\Z$ is contained inside the closure $\overline{\Z_{fr}}$.
\end{proof}

We deduce the ergodicity of our action.

\begin{corollary}\label{coro:ergodicZfr}
If $\beta:\D\to \R_{>0}$ is $p$-summable for some $0< p < 1/2$, then the action $\theta:\Z_{fr}\act (\Z^\D,\mbD)$ is ergodic.
\end{corollary}

\begin{proof}
We start by showing that the action of $\oplus_\D \Z$ is ergodic.
Consider the crossed-product von Neumann algebra $N_\beta:= L^\infty(\Z^\D, \mbD)\rtimes \oplus_\D\Z$ and observe that it is isomorphic to the infinite tensor product $$\bigotimes_{d\in\D} (L^\infty(\Z , m_{\beta(d)})\rtimes \Z , \varphi_d) = \bigotimes_{d\in\D} (B(\ell^2(\Z)) , \varphi_d),$$
where $\varphi_d$ is the state defined as 
$$\varphi_d(x) = \sum_{k\in\Z} \langle x \delta_k,\delta_k\rangle m_{\beta(d)}(\{k\}), \ x\in B(\ell^2(\Z)).$$
By \cite[XIV, Corollary 1.10]{TakesakiIII} we have that $N_\beta$ is a factor since it is the GNS completion of a tensor product of factors w.r.t.~to a tensor product state.
This implies that the action $\oplus_\D\Z\act (\Z^\D,\mbD)$ is necessarily ergodic.
 
Consider $X\subset \Z^\D$ a measurable subset that is invariant under the action of $\Z_{fr}$.
By continuity, it is also invariant for the closure $\overline{\Z_{fr}}$ inside $\Aut(\Z^\D, \mbD)$ and thus, by Proposition \ref{prop:closurZfr}, is invariant under the action of $\oplus_\D \Z$.
Since this group acts in an ergodic way we obtain that $X$ has measure one or zero implying that the action of $\Z_{fr}$ is ergodic.
\end{proof}

The next theorem described the structure of $\scrM$ that is radically different from the case where $\beta$ is constant, see Theorem \ref{theo:singlemes}.

\begin{theo}\label{theo:type}
Consider $\beta:\D\to\R_{>0}$ that is in $\ell^{p}(\D)$ for some $0<p<1/2$ and the action $\theta:\Z_{fr}\act (\Z^\D , \mbD).$
Then $\theta$ is nonsingular, free, ergodic and of type III.

In particular, the von Neumann algebra $\scrM$ is a hyperfinite type III factor.
\end{theo}
\begin{proof}
The action is nonsingular by Proposition \ref{prop:nonsingular}, {ergodic by Corollary \ref{coro:ergodicZfr}} and is obviously free by definition.

Assume that the action is of type I or II.
Since the action is ergodic, there exists a measure $\mu$ in the measure class of $\mbD$ that is $\Z_{fr}$-invariant.
By continuity, this measure is invariant for the action of the closure $\overline{\Z_{fr}}$ and in particular is $\oplus_\D\Z$-invariant by Proposition \ref{prop:closurZfr}.
As in the proof of Corollary \ref{coro:ergodicZfr} we consider the factor $N_\beta:=L^\infty(\Z^\D, \mbD)\rtimes \oplus_\D\Z$ that is isomorphic to an infinite tensor product of type I factor equipped with states $\varphi_d$.
This factor is semi-finite since $\mbD$ is equivalent to a $\oplus_\D \Z$-invariant measure.
Denote by $\Tr(x):= \sum_{k\in\Z}  \langle x \delta_k,\delta_k\rangle$ the normal faithful tracial weight of $B(\ell^2(\Z))$ and observe that $\varphi_d(x) = \Tr( x h_d)$ where $h_d = \sum_{k\in\Z} m_{\beta(d)}(\{k\}) e_{k,k}$ and where $(e_{k,l},\ k,l\in\Z)$ is the classical system of matrix units of $B(\ell^2(\Z))$.
Moreover, by \cite[XIV, Theorem 1.14]{TakesakiIII}, this factor is semi-finite if and only if 
$$S(t):=\sum_{d\in\D} ( 1 - | \Tr( h_d^{1+it}) | ) <\infty \text{ for any } t\in\R.$$
We will show this series diverges for $t=1$. 
Observe that $$\Tr(h_d^{1+i}) = \dfrac{\sum_{k\in\Z}e^{-k^2 \beta(d) (1+i)/2}  }{Z_{\beta(d)}^{1+i}} = \dfrac{\sum_{k\in\Z}e^{-k^2 \beta(d) (1+i)/2}  }{(\sum_{\ell\in\Z}e^{-\ell^2 \beta(d)/2})^{1+i}}.$$
The Jacobi-Poisson formula \eqref{equa:theta} implies that 
$$\left( \sum_{k\in\Z} e^{-k^2 z } \right)^2 \sim_{z\to 0} \dfrac{\pi}{z}$$
for a complex number $z$ with strictly positive real part.
We obtain that when $d\in\D$ goes to infinity (and thus $\beta(d)$ tends to zero) we have that
$$|\Tr(h_d^{1+i})|=|\dfrac{\sum_{k\in\Z}e^{-k^2 \beta(d) (1+i)/2}  }{\sum_{\ell\in\Z}e^{-\ell^2 \beta(d)/2} }| \to_{d\to\infty} |\dfrac{1}{1+i}|^{1/2} \neq 1$$ implying that the series $S(1)$ diverges.
Therefore, $N_\beta$ is not semi-finite, a contradiction.
Therefore, $\mbD$ does not admit any equivalent $\Z_{fr}$-invariant measure and thus the action $\theta$ is of type III.
Since $\scrM$ is isomorphic to the crossed-product $L^\infty(\Z^\D,\mbD)\rtimes \Z_{fr}$ we obtain that $\scrM$ is a type III factor.
Moreover, it is hyperfinite since $\Z_{fr}$ is abelian.
\end{proof}

We provide a class of geometric examples of $\beta$ satisfying the hypothesis of the last theorem.

\begin{example}\label{ex:beta}
Consider $\beta:=\beta_\tau(d)= (d'-d)^\tau$ where $\tau>2$ is fixed and $(d,d')$ is the largest s.d.i.~starting at $d.$
We claim that this choice of $\beta$ provides a nonsingular free ergodic action $\alpha_\tau:\Z_{fr}\act (\Z^\D,m_\beta^\D)$ of type III.
To prove it it is sufficient to show that $\beta$ is in $\ell^{p}(\D)$ for some $0<p<1/2$ by Theorem \ref{theo:type}.
Write $\D_n:=\{ d\in \D: \beta_{1}(d)=2^{-n}\}$ and observe that $| D_n | = 2^{n-1}$ if $n\geq 1.$
Since $\tau>2$, there exists $0<p<1/2$ such that $p \tau>1.$
We obtain 
\begin{align*}
\sum_{d\in\D} \beta(d)^p & = \sum_{n=0}^\infty \sum_{d\in \D_n} \beta(d)^p \\
& = 1 +  \sum_{n=1}^\infty |D_n| \left(\dfrac{1}{2^n}\right)^{ p \tau}\\
& = 1+ \sum_{n=1}^\infty \dfrac{ 2^{ n-1 } }{ 2^{ n p \tau } } = 1 + 1/2 \sum_{n=1}^\infty (2^n)^{1- p \tau}\\
\end{align*}
that is summable since $p \tau>1$.
\end{example}

\begin{remark}
Theorem \ref{theo:type} provides a type III factor $(\scrM,\varpi)$ equipped with a faithful state.
By Tomita-Takesaki theory we obtain a nontrivial modular action $\R\act \scrM$ that we interpret as a time-evolution on our continuum limit algebra $\scrM$.
\end{remark}

The next proposition shows that the choice of $\beta_\tau$ above gives a measure $\mbD$ that is not quasi-invariant under the action of Thompson's groups.

\begin{prop}\label{prop:SingularF}
Consider $\beta:=\beta_\tau(d)= (d'-d)^\tau$ where $\tau>1$ is fixed and $(d,d')$ is the largest s.d.i.~starting at $d.$
Let $\mbD:=\ot_{d\in\D} m_{\beta(d)}$ be the associated measure on $\Z^\D.$
Then the generalized Bernoulli action of Thompson's group $F$
$$\kappa: F\act \Z^\D, \kappa_v (x)(d) = x(v^{-1} d),\ v\in V, x\in\Z^\D, d\in \D$$
is singular w.r.t.~the measure $\mbD.$

In particular, the Jones' action $\alpha:F\act \scrM_0$ does not extends to an action on the von Neumann algebra $\scrM$.
\end{prop}

\begin{proof}
Consider $v\in F$ and observe that ${\kappa_v^{-1}}_*\mbD = \ot_{d\in\D} m_{\beta(vd)}.$
By Kakutani's theorem (see \cite{Kakutani48}) we have that ${\kappa_v^{-1}}_*\mbD$ and $\mbD$ are equivalent (i.e.~are in the same measure class) if and only if $$\sigma^D({\kappa_v^{-1}}_*\mbD , \mbD) := \sum_{d\in \D} \sigma( m_{\beta(vd)} , m_{\beta(d)} ) <\infty,$$
where $$\sigma( m_{\beta(vd)} , m_{\beta(d)} ) = -\log (\rho( m_{\beta(vd)} , m_{\beta(d)} )) = -\log(\sum_{n\in\Z} \sqrt{ m_{ \beta(vd) } ( \{ n \} ) m_{ \beta(vd) } ( \{ n \} ) }).$$
Recall that $m_b\in\Prob(\Z)$ is defined as $m_b(\{n\}) = \tfrac{e^{-n^2 b/2}}{Z_b}$ where $Z_b = \sum_{k\in\Z} e^{-k^2 b/2} $ and that $Z_b\sim_{b\to 0} \sqrt{\tfrac{2\pi}{b}}.$
Moreover, if $a,b\in (0,1)$, then we have that 
$$\rho(m_a,m_b) = \dfrac{Z_{\frac{a+b}{2}}}{\sqrt{Z_aZ_b}} \sim_{a,b\to 0} \left(\dfrac{2\sqrt{ab}}{a+b} \right)^{1/2}.$$
Since $\beta(d)$ tends to zero when $d$ tends to infinity we obtain that 
$$\sigma( m_{\beta(vd)} , m_{\beta(d)} ) \sim_{d\to\infty} -\dfrac{1}{2}\log\left(\dfrac{2\sqrt{\beta(d)\beta(vd)}}{\beta(d)+\beta(vd) }\right).$$

Consider an element $v\in F$ satisfying $v(x)=x/2$ for any $x\in (0,1/2)$ and observe that $\beta(\tfrac{1}{2^n}) = \tfrac{1}{2^{\tau n}}$ for any $n\geq 1.$
Therefore, 
\begin{align*}
\sigma^D({\kappa_v^{-1}}_*\mbD , \mbD) & \geq \sum_{n=2}^\infty  -\log (\rho( m_{\beta(\frac{1}{2^{n+1} } )} , m_{\beta(\frac{1}{2^{n} })} )) \\
& = \sum_{n=2}^\infty  -\log (\rho( m_{\frac{1}{2^{(n+1)\tau} } } , m_{\frac{1}{2^{n \tau} }} ))\\
\end{align*}
Moreover, the term of this series is equivalent when $n$ tends to infinity to 
$$-\dfrac{1}{2}\log\left(\dfrac{2\sqrt{ 2^{-(2n+1)\tau}    } }{ 2^{-n\tau} + 2^{-(n+1)\tau} }\right) = -\dfrac{1}{2} \log\left( \dfrac{2}{2^{\tau/2} +2^{-\tau/2}}\right)>0.$$
This implies that the series diverges and thus the two measures ${\kappa_v^{-1}}_*\mbD$ and $\mbD$ are not equivalent.
Since the state $\varpi$ restricts to $\mbD$ on the subalgebra $L^\infty(\Z^\D,\mbD)$, we obtain that the map $\alpha(v)$ defines as an automorphism of the C*-algebra $\scrM_0$ is not continuous w.r.t.~weak topology of $\scrM$ and thus does not extend in a normal way on it.
\end{proof}

A similar proof shows that any element $v\in V$ for which there exists an interval $I$ such that $\Leb(vI)\neq \Leb(I)$ defines a singular transformation $\kappa_v$ w.r.t.~the measure $m_\beta$ with $\beta$ as above.
Therefore, any elements that do not only permute intervals of same length act in a singular way and fail to provide an automorphism of $\scrM$.

\subsubsection{Action of the rotations}\label{sec:rotact}
Consider the rooted binary complete tree $t_n$ with $2^n$ leaves all at distance $n$ from the root.
Let $r_n=\tfrac{(t_n,1)}{(t_n,0)}$ be the element of $T$ that permutes by one the leaves of the tree $t_n$.
It corresponds to the rotation by angle $2^{-n}$ when $T$ acts on the torus $\R/\Z.$
Write $\Rot_n$ the subgroup of $T$ generated by $r_n$ that is isomorphic to the cyclic group $\Z/ 2^n\Z$ and note that $r_n^2 =r_{n-1}$ implying that $\Rot_{n-1}$ is a subgroup of $\Rot_n.$
Let $\Rot:=\bigcup_{n\geq 1}\Rot_n$ be the union of all those groups that is a subgroup of $T$ and which is isomorphic to $\varinjlim_{n\in\N} \Z/ 2^n\Z$ for the system of inclusions $k+2^n\Z \mapsto 2k+2^{n+1}\Z$.

We wonder if the action of $\Rot$ on $\Z^\D$ is nonsingular for some choice of measure $\mbD$ with $\beta:\D\to\R_{>0}$.  This would provide actions of $\Rot$ on the von Neumann algebra $\scrM$.
We will show that it is the case for a large class of $\beta$.

Consider the function $\ell:\D\to \R, d\mapsto -\log_2(d'-d)$ where $(d,d')$ is the largest s.d.i.~starting at $d$ and $\log_2$ is the logarithm in base 2.
For example, $\ell(0)=0, \ell(1/2)=1, \ell(1/4)=\ell(3/4)=2,$ etc.

Note that the s.d.p associated to $t_n$ is $\{(\tfrac{j}{2^n},\tfrac{j+1}{2^n}) : 0\leq j \leq 2^n-1\}$.
Hence, $r_n$ sends $\tfrac{j}{2^n}$ to $\tfrac{j+1}{2^n}.$
The next lemma concerns the value of $\ell(\tfrac{j}{2^n})$ which will allow us to compare $\ell(\tfrac{j}{2^n})$ and $\ell(r(\tfrac{j}{2^n}))=\ell(\tfrac{j+1}{2^n}).$ 
The proof is a routine computation that is left to the reader.

\begin{lemma}\label{lem:ell-function}
Consider the function $\ell:\D\to \R$ and the $2^n$-tuple 
$$X_n=(\ell(d_1), \ell(d_2),\cdots, \ell(d_{2^n}))$$ where $d_j:=\tfrac{j-1}{2^n}$.
Then $X_{n+1}$ is equal to $$(\ell(d_1) , n+1 , \ell(d_2), n+1 , \cdots, \ell(d_{2^n}), n+1 ).$$
\end{lemma}

This implies the following proposition.

\begin{prop}
If $n\geq 1$, $d\in\D$ and $\ell(d)\neq \ell(r_n(d))$, then $d=\tfrac{j}{2^n}$ for some $j.$ 
In particular, for any rotation $r\in \Rot$ there exists only finitely many $d\in\D$ for which $\ell(d)\neq \ell(r(d)).$
\end{prop}

\begin{proof}
Consider $n\geq 1$ and $d\in \D$.
There exists $k\geq 0$,  $0\leq i\leq 2^n -1$ and $0\leq j\leq 2^k-1$ such that $d = \tfrac{i}{2^n} + \tfrac{j}{2^{k+n}}.$
Consider the tree $t_{n+k}$ that we write as the composition $t_{n+k} = (s_1\ot s_2\ot\cdots \ot s_{2^n}) \circ t_n$ where $s_m$ is a copy of $t_k$ for any $1\leq m\leq 2^n.$
Then $d$ corresponds to the $j+1$th leaf of the tree $s_{i+1}$ where we index the trees cyclicly in $\Z/2^n\Z$. 
Observe that $r_n(d)$ corresponds to the $j+1$th leaf of the tree $s_{i+2}$.
Let $Y_q=(y_1^q,\cdots, y_{2^k}^q)$ be the $2^k$-tuple such that $y_m^q=\ell(d_m)$ and $d_m$ corresponds to the $m$th leaf of $s_q$.
Hence, $\ell(d) = y_{j+1}^{i+1}$ and $\ell(r_n(d)) = y_{j+1}^{i+2}.$
Lemma \ref{lem:ell-function} implies that $y_{m}^\iota$ does not depend on $\iota$ if $m\neq 1.$
Therefore, if $\ell(d)\neq \ell(r_n(d))$, then necessarily $d$ is on the first leave of $s_{i+1}$ which means that $d=\tfrac{a}{2^n}$ for some $a\geq 0.$
Consider any rotation $r\in \Rot$. 
There exists $m\geq 1$ and $q\geq 0$ such that $r= r_m^q$.
The first part of the proof shows that if $\ell(d)\neq \ell(r(d))$, then $d=\tfrac{p}{2^m}$ for some $p\geq 0$ and thus there are only finitely many such $d\in \D.$
\end{proof}

This implies the following corollary about nonsingularity.

\begin{corollary}\label{cor:Rot}
Consider a map $\beta:\D\to\R_{>0}$ such that $\beta(d)$ only depends on $\ell(d)$, i.e.~$\beta=b\circ\ell$ for some function $b$.
Let $\mbD$ be the associated probability measure on $\Z^\D$ and consider the action of $V\act (\Z^\D,\mbD)$ that we restrict to the rotation subgroup $\Rot$.

Then the action $\Rot\act (\Z^\D,\mbD)$ is nonsingular and thus the restriction to $\Rot$ of the Jones' action $\Rot\act \scrM_0$ acting on the C*-algebra $\scrM_0$ extends to an action by automorphisms on the von Neumann algebra $\scrM.$

More generally, the action $\Rot\act (\Z^\D, \otimes_{d\in\D} \nu_d)$ is nonsingular if $(\nu_d,\ d\in\D)$ is any family of probability measures that are all mutually equivalent and such that $\nu_d = \nu_{d'}$ if $\ell(d)=\ell(d').$
\end{corollary}

\begin{remark}
\label{rem:tgauge}
Although, it is natural to have an action of the rotation subgroup $\Rot\subset T$ on our field algebra $\scrM$, we expect that it is not possible to define the generator of rotations as a strong limit,
$$s-\lim_{n\rightarrow\infty}\tfrac{1}{2^{n}}(U_{r_{n}}-1)\xi = L_{0}\xi,$$ 
as a consequence of gauge invariance, cp.~Remark \ref{rem:gausslaw} and \cite{Loeffelholz-03-Temporal-Gauge}, see also \cite{Jones16-Thompson, Brot-Jones18-2}. In the conformal field theory interpretation of Proposition \ref{prop:observables}, $L_{0}$ would be the conformal Hamiltonian as opposed to the Kogut-Susskind Hamiltonian which motivates our definition of the heat-kernel states. $L_{0}$ might exist on the level of observables, if we considered the gauge invariant subsector of our model. But, in the case of YM$_{1+1}$ this would result in a topological theory without local degrees of freedom, see e.g.~\cite{Dimock-96-Yang-Mills}.
\end{remark}

Example \ref{ex:beta} provides such $\beta$ and thus some examples of models $\scrM$ admitting the action of the rotation subgroup $\Rot.$

\begin{remark}\label{rem:generalM}
Most of the proofs provided in the special case of the circle group can be adapted to the case of a general abelian group $G$ and some measures $m_d, d\in\D$ on its Pontryagin dual.
However, we do not have any general criteria for nonsingularity and ergodicity of the action.
But, if the action $\wh G_{fr}\act ({\wh G}^\D, \ot_{d\in\D} m_d)$ is nonsingular, then $(\scrM,\varpi)$ is isomorphic to the crossed-product von Neumann algebra $L^\infty({\wh G}^\D, \ot_{d\in\D} m_d)\rtimes \wh G_{fr}$ equipped with the state given by the measure $\ot_{d\in\D} m_d$ that is thus faithful.
If the action is ergodic (which can be proved by showing that $\oplus_\D\wh G$ belongs to the weak completion of $\wh G_{fr}$ inside $\Aut({\wh G}^\D, \ot_{d\in\D} m_d)$), then $\scrM$ is always a hyperfinite type III factor.
\end{remark}

\subsection*{Acknowledgement}
Part of this project was done when both authors were working at the University of Rome, Tor Vergata thanks to the very generous support of Roberto Longo.
We are very grateful to Roberto for giving us this great opportunity besides his constant encouragement and support during various stages of the project. AS thanks the Alexander-von-Humboldt Foundation for generous financial support during his stay at the University of Rome, Tor Vergata. Moreover, AS acknowledges financial support and kind hospitality by the Isaac Newton Institute and the Banff International Research Station where parts of this work were developed.
Furthermore, we are grateful to Thomas Thiemann, Yoh Tanimoto, Luca Giorgetti and Vincenzo Morinelli for comments and discussions during various stages of this work.
{Finally, we are grateful to the referee for constructive comments which improved the clarity of this manuscript.}

\end{document}